\documentclass[journal,draftclsnofoot,onecolumn]{IEEEtran}

\usepackage{amsmath}
\allowdisplaybreaks[1]
\usepackage{setspace}

\usepackage{amssymb}
\usepackage{mathrsfs}
\usepackage{amsthm}
\usepackage{multirow}
\usepackage{color}
\usepackage{longtable}
\usepackage{array}
\usepackage{url}
\usepackage{comment}
\usepackage{enumerate}
\usepackage{eucal}
\usepackage{graphics}
\usepackage{verbatim}
\usepackage{subfigure}

\usepackage{algorithm}
\usepackage{algpseudocode}
\usepackage{cite}
\usepackage[table]{xcolor}

\usepackage{mathtools}
\usepackage{xparse}
\usepackage{bm}

\DeclarePairedDelimiterX{\set}[1]{\{}{\}}{\setargs{#1}}
\NewDocumentCommand{\setargs}{>{\SplitArgument{1}{;}}m}
{\setargsaux#1}
\NewDocumentCommand{\setargsaux}{mm}
{\IfNoValueTF{#2}{#1} {#1\,\delimsize|\,\mathopen{}#2}}

\DeclarePairedDelimiter\abs{\lvert}{\rvert}

\DeclarePairedDelimiter\parenv{\lparen}{\rparen}

\newtheorem{thm}{Theorem}

\newtheorem{lem}{Lemma}
\newtheorem{dfn}{Definition}
\newtheorem{prop}{Proposition}
\newtheorem{rmk}{Remark}
\newtheorem{cor}{Corollary}
\newtheorem{example}{Example}
\newenvironment{pf}{{\noindent\it Proof:}}{\hfill $\blacksquare$\par}

\newcommand{\RNum}[1]{\lowercase\expandafter{\romannumeral #1\relax}}
\newcommand{\Rnum}[1]{\uppercase\expandafter{\romannumeral #1\relax}}

\newcommand{\F}{\mathbb{F}}

\newcommand{\eqdef}{\triangleq}
\newcommand{\cA}{\mathcal{A}}
\newcommand{\cB}{\mathcal{B}}
\newcommand{\cC}{\mathcal{C}}

\newcommand{\cE}{\mathcal{E}}

\newcommand{\cG}{\mathcal{G}}
\newcommand{\cN}{\mathcal{N}}
\newcommand{\cO}{\mathcal{O}}
\newcommand{\cP}{\mathcal{P}}

\newcommand{\cV}{\mathcal{V}}
\newcommand{\cW}{\mathcal{W}}

\newcommand{\heta}{ \widehat{\eta}}
\newcommand{\hbc}{ \widehat{{\mathbf{C}}}}
\newcommand{\htheta}{\widehat{\theta}}
\newcommand{\hvarphi}{\widehat{\varphi}}

\newcommand{\ba}{\mathbf{a}}
\newcommand{\bb}{\mathbf{b}}
\newcommand{\bc}{\mathbf{c}}

\newcommand{\bff}{\mathbf{f}}
\newcommand{\bg}{\mathbf{g}}
\newcommand{\bk}{\mathbf{k}}

\newcommand{\bu}{\mathbf{u}}
\newcommand{\bw}{\mathbf{w}}
\newcommand{\bv}{\mathbf{v}}
\newcommand{\bx}{\mathbf{x}}
\newcommand{\by}{\mathbf{y}}

\newcommand{\bK}{\mathbf{K}}
\newcommand{\bU}{\mathbf{U}}

\newcommand{\bY}{\mathbf{Y}}
\newcommand{\bF}{\mathbf{F}}
\newcommand{\bG}{\mathbf{G}}
\newcommand{\bP}{\mathbf{P}}
\newcommand{\bI}{\mathbf{I}}

\newcommand{\Zero}{\mathbf{0}}

\newcommand{\In}{{\rm In}}
\newcommand{\Out}{{\rm Out}}
\newcommand{\tail}{{\rm tail}}
\newcommand{\head}{{\rm head}}

\newcommand{\Rank}{{\rm Rank}}

\newcommand{\Span}{{\rm Span}}

\begin{document}

\title{Network function computation with vector linear target function and security function}

\author{
  Min Xu, Qian Chen and  Gennian Ge%
  \thanks{This research was supported by the National Key Research and Development Program of China under Grant 2025YFC3409900, the National Natural Science Foundation of China under Grant 12231014, and Beijing Scholars Program.}
  \thanks{M. Xu (e-mail: minxu0716@qq.com) is with the Institute of Mathematics and Interdisciplinary Sciences, Xidian University, Xi'an 710126, China.}%
  \thanks{Q. Chen (e-mail: 17838590226@163.com) and G. Ge (e-mail: gnge@zju.edu.cn) are with the School of Mathematical Sciences, Capital Normal University, Beijing 100048, China.}
}

\maketitle
\begin{abstract}
   In this paper, we study the problem of securely computing a function over a network, where both the target function and the security function are vector linear. The network is modeled as a directed acyclic graph. A sink node wishes to compute a function of messages generated by multiple distributed sources, while an eavesdropper can access exactly one wiretap set from a given collection. The eavesdropper must be prevented from obtaining any information about a specified security function of the source messages. The secure computing capacity is the maximum average number of times that the target function can be securely computed with zero error at the sink node with the given collection of wiretap sets and security function for one use of the network. We establish two upper bounds on this capacity, which hold for arbitrary network topologies and for any vector linear target and security functions. These bounds generalize existing results and also lead to a new upper bound when the target function is the sum over a finite field. For the lower bound, when the target function is the sum, we extend an existing method, which transforms a non-secure network code into a secure one, to the case where the security function is vector linear. Furthermore, for a particular class of networks and a vector linear target function, we characterize the required properties of the global encoding matrix to construct a secure vector linear network code.
\end{abstract}
\begin{IEEEkeywords}
   Secure network function computation; vector linear target function; vector linear secure function; capacity
\end{IEEEkeywords}
\section{Introduction}
Network function computation is a fundamental primitive in distributed systems such as sensor networks and the Internet of Things \cite{2004koetternetwork,2012rainetwork,2013ramamoorthycommunicating,2013AFKZ,2014Appuswamy,2018HTYG,2019GYYL}.
The standard framework for network function computation is as follows. In a network $\cN$, which is modeled by a directed acyclic graph, there is a sink node and multiple source nodes. The sink aims to compute a function $f$ of the messages generated by the source nodes. We call $f$ as the \emph{target function}. Intermediate nodes may encode and forward their received messages along communication links to downstream nodes. We assume each link has unit capacity, which is a common and reasonable abstraction, especially when multiple parallel edges are allowed between nodes.
The computing capacity of a network code is the maximum average number of times that the target function can be reliably computed at the sink for one use of the network.
A central challenge is to characterize the computing capacity for a given network topology and a target function \cite{2011AFKZ,2018HTYG,2019GYYL,2014Appuswamy,2024GuangZhang}. In \cite{2011AFKZ}, Appuswamy {\it et al.} formalized this problem and provided a general upper bound on the computing capacity, which is not always valid. This bound was later refined and generalized by Guang {\it et al.} \cite{2018HTYG,2019GYYL}, and the refined upper bound is valid for  arbitrary networks and target functions. Although there is an example where this bound is not tight, it remains the best known general upper bound. 
Moreover, the upper bound is achievable for specific cases (e.g., when $f$ is the sum or identity function, or for certain network topologies \cite{2011AFKZ}).
In recent years, the case where the target function $f$ is a vector linear function has attracted increasing attention. Appuswamy et al.~\cite{2011AFKZ,2013AFKZ} studied computing vector linear target function over networks using linear codes, and established a necessary condition for the existence of rate $1$ linear codes. For general network topologies, Zhou and Fu~\cite{2024ZhouFu} derived two upper bounds on the computing capacity, which can be obtained from the upper bound in \cite{2019GYYL}. In more specialized settings, Li and Xu~\cite{2022LiXudiamond} showed that the upper bound from~\cite{2019GYYL} is tight for the diamond network, while Guang et al.~\cite{guang2025distributed} investigated distributed source coding for a three‑layer network with only three sources and three middle nodes. In summary, although two general upper bounds are available for arbitrary networks, the exact capacity has been determined only for very specific and simple topologies.

In practice, communications between nodes are susceptible to various challenges, including eavesdropping, transmission errors, and node collusion. These issues create distinct security and robustness requirements for network function computation \cite{2023WeiXuGe,2024Guangsourcesecure,2025Guangfunctionsecure,wei2024linear,xu2022network}.
In this paper, we focus on the problem of secure network function computation, which has been investigated in \cite{2024Guangsourcesecure,2025Guangfunctionsecure,xu2022network}. In the problem, there is an eavesdropper who has access to some links in the network and aims to learn a function $g$ of the source messages. We call $g$ as the \emph{security function}.  The goal of secure network function computation is to ensure that the sink node can correctly compute the target function while preventing the wiretapper from obtaining any information about the security function. In prior work \cite{2024Guangsourcesecure,2025Guangfunctionsecure}, Guang {\it et al.} investigated scenarios where $f$ is the sum function over finite field and $g$ is the identity function (referred to as \emph{source security}) or the same sum function (referred to as \emph{target-function security}). For each case, they derived an upper bound on the secure computing capacity. They also proposed a method to construct a linear secure network code by transforming a non-secure one. For certain network topologies, the constructed network code achieves the upper bound.  

When both the target function and the security function are the identity function, the secure network function computation problem reduces to the classical secure network coding problem \cite{cai2002secure,cai2010secure,el2012secure,guang2018alphabet,silva2011universal,guang2020local,bhattad2005weakly}. The fundamental model was first introduced by Cai and Yeung \cite{cai2002secure,cai2010secure}, where the goal is to prevent a wiretapper from gaining any information about the source messages. Several works have explored relaxed security constraints. For instance, Bhattad {\it et al.}  \cite{bhattad2005weakly} studied the weakly secure network coding, where the wiretapper is only prevented from obtaining complete information about the sources, a relaxation that can improve the achievable communication rate.
More recently, Bai {\it et al.} \cite{bai2023multiple} investigated the multiple linear combination security, which requires the wiretapper cannot obtain a predetermined vector linear function about the source messages. Their work analyzed the trade-off between the network's communication capacity and the rank of the coefficient matrix defining the security function. Building upon this line of research, we will study the secure computing capacity for the general scenario where both the target function and the security function are vector linear in this paper.

A closely related line of research is secure aggregation \cite{2022ZhaoSun,2024Zhaosun,2024WanYaoSunJiCaire,2024WanSunJiCaire,2025yuanSun,2024zhangwansunwangoptimal,2025LiZhangLvFan,hu2026capacity,zhang2025information,li2025capacity,so2022lightsecagg,jahani2023swiftagg+,xu2025hierarchical}. In its typical setting, a server (the sink) computes the sum of messages from multiple users (the sources), while being prevented from learning any individual message beyond the aggregated sum.
The secure aggregation model differs from the secure network function computation studied in this paper in two main aspects. 
First, secure aggregation often relies on a trusted third party to assign correlated random keys to users. In our model, each source node generates its randomness independently. Second, the adversary in secure aggregation is usually a set of colluding participating nodes, rather than an external wiretapper eavesdropping on communication links.
Most prior work on secure aggregation focuses on computing the sum function. From a network function computation viewpoint, this corresponds to the case where the target function is the algebraic sum\footnote{We use algebraic sum to represent the sum function over finite fields.} and the security function is the identity function. Recent studies have begun to generalize this setup. 
Notably, Yuan and Sun \cite{2025yuanSun} investigated the scenario where both the target and security functions are vector linear, which is closely related to ours. This direction was further extended by Hu and Ulukus in \cite{hu2026capacity}. It is important to note that their work remains within the secure aggregation framework and thus inherits the two key distinctions mentioned above. Furthermore, our work considers a more general and complex network setting, namely an arbitrary directed acyclic graph, which introduces additional challenges in code design and capacity analysis that do not exist in the typical aggregation topology.

The main contributions of this paper are summarized as follows.
\begin{enumerate}
    \item We establish two upper bounds on the secure computing capacity, which hold for arbitrary network topologies and for any vector linear target and security functions. 
    The first bound generalizes the prior results to the vector linear setting. The second bound is more novel and is derived by selecting cut sets that satisfy specific conditions. This new method yields a strictly better bound in certain cases. In particular, when the target function is the sum and the security function is the identity,  we construct an example to demonstrate that the new bound is tighter than pervious results.
    \item For general networks and sum target function, we extend the method of transforming a non‑secure network code into a secure one to the case where the security function is vector linear. Moreover, we characterize the required field size for the existence of such a secure code, which matches and generalizes existing existence results.
    \item For three‑layer multi‑edge tree networks with vector‑linear target and security functions, we propose a secure vector‑linear network code via explicit construction of the global encoding matrices. Unlike scalar‑linear codes, the messages transmitted on each link are vectors rather than single symbols.
\end{enumerate}

\emph{Organization.} The rest of this paper is organized as follows. Section~\ref{sec:system-model} introduces the system model for secure network function computation and reviews relevant results. In Section~\ref{sec:upper-bounds}, we present two upper bounds, which are derived by information-theoretic methods. Section~\ref{sec:lowerbound-sum} provides a construction of a (scalar) linear secure network code for the sum target function via a standard procedure and analyzes the required field size. In Section~\ref{sec:lowerbound-tree}, we consider three-layer multi-edge trees and propose a vector linear secure network code by characterizing the properties of the global encoding matrix. Finally, Section~\ref{sec:conclusion} summarizes the results and directions for future work.

\emph{Notation.} Throughout this paper, the following notations are used: Matrices are denoted by bold uppercase letters. Random variables and random vectors are denoted by uppercase letters. Deterministic vectors (non-random) are denoted by boldface, lowercase letters. For integers $m\leq n$, let $[m:n] \eqdef \{m, m + 1, \cdots , n\}$ if $m \leq n$ and
$[m : n] = \emptyset $, if $m > n$. $[1 : n]$ is written as $[n]$ for brevity. 

\section{System model and related results}\label{sec:system-model}
\subsection{The model of secure network function computation}
Consider a directed acyclic graph $\cG=(\mathcal{V},\mathcal{E})$ with a finite vertex set $\mathcal{V}$ and an edge set $\mathcal{E}$. The graph may contain multiple edges connecting the same pair of vertices. For an edge $e\in \mathcal{E}$, we use $\tail(e)$ and $\head(e)$ to denote the \emph{tail} node and the \emph{head} node of $e$.
For a vertex $v\in \mathcal{V}$, let $\In(v)=\set{e\in E;\head(e)=v}$ and $\Out(v)=\set{e\in E;\tail(e)=v}$, respectively. 
A sequence of edges $(e_1, e_2, \cdots, e_n)$ forms a \emph{path} from vertex $u$ to vertex $v$ if $\tail(e_1) = u$, $\head(e_n) = v$, and $\tail(e_{i+1}) = \head(e_i)$ for all $i = 1, 2, \dots, n-1$.
Given two disjoint vertex subsets $U, V \subseteq \mathcal{V}$, a \emph{cut separating $V$ from $U$} is a set of edges $C \subseteq \mathcal{E}$ such that after removing $C$, no path remains from any vertex in $U$ to any vertex in $V$. In particular, for two nodes $u,v\in\cV$, a cut separating $\{v\}$ from $\{u\}$ is called \emph{a cut separating $v$ from $u$}. A cut $C$ separating $V$ from $U$ is called a \emph{minimum cut} if no other cut $C'$ separating $V$ from $U$ satisfies $|C'| < |C|$.

In this paper, a \emph{network} $\cN$ over $\cG$ contains a set of \emph{source nodes} $S=\set{\sigma_1,\sigma_2,\ldots,\sigma_s}\subseteq \cV$, and a \emph{sink node} $\gamma \in \cV\setminus S$. We denote this network by $\cN = (\cG, S, \gamma)$. Without loss of generality, we assume that every source node has no incoming edges. Furthermore, we assume that every vertex $u \in \mathcal{V} \setminus {\gamma}$ has a directed path to $\gamma$ in $\cG$. Consequently, due to the acyclicity of $\cG$, the sink $\gamma$ has no outgoing edges. Any cut that separates $\gamma$ from a source node $\sigma_i$ is termed a \emph{cut of the network}. The collection of all such cuts for $\cN$ is denoted by $\Lambda(\mathcal{N})$.

In the network function computing problem, the sink node $\gamma$ needs to compute a \emph{target function} $f$ of the form \[f: \cA^s\longrightarrow\mathcal{O},\]
where $\cA$ and $\cO$ are finite alphabets, and the $i$-th argument of $f$ is generated at the source node $\sigma_i$.
For each $i \in [s]$, the information at source $\sigma_i$ is modeled as a random variable $M_i$, uniformly distributed over $\cA$. All the source messages $M_1,\dots,M_s$ are mutually independent.  The source node $\sigma_i$ generates $\ell$ independent identical distributed (i.i.d.) random variables $M_{i,1},M_{i,2},\cdots,M_{i,\ell}$, each distributed as $M_i$. Denote $M_i=(M_{i,1},M_{i,2},\cdots,M_{i,\ell})$ as the \emph{source message generated by $\sigma_i$}, and let $M_S=(M_1,M_2,\cdots,M_s)$ be the \emph{source message vector generated by $S$}. The sink node $\gamma$ must compute 
\[f(M_S)\eqdef(f(M_{1,j},M_{2,j},\cdots,M_{s,j}):j=1,2,\cdots,\ell)\]
with zero error.

In addition, let $\cW$ be a family of edge subsets, each $W \in \cW$ is called a wiretap set. Let $g:\cA^s\rightarrow\mathcal{Q}$ be a non-constant \emph{security function}, where $\mathcal{Q}$ is the image set of $g$. In the secure network function computation problem, the result $f(M_S)$ is required to be correctly computed at $\gamma$ through the network $\cN$, while 
\[g(M_S)\eqdef(g(M_{1,j},M_{2,j},\cdots,M_{s,j}):j=1,2,\cdots,\ell)\]
is required to not be leaked to a wiretapper who can access any one but not more than one wiretap set in $\cW$. Both $\cW$ and $g$ are known by the source nodes and the sink node, but which wiretap set in $\cW$ is eavesdropped by the wiretapper is unknown. We use $(\cN,f,g,\cW)$ to denote the model as specified above. 

To define the secure network codes for $(\cN,f,g,\cW)$, we introduce random keys generated at the source nodes, following the definition in \cite{2024Guangsourcesecure,2025Guangfunctionsecure}. For each $i\in[s]$, source node $\sigma_i$ has access to a random variable $K_i$, called a \emph{random key}, uniformly distributed over a finite set $\mathcal{K}_i$. As noted in \cite{2024Guangsourcesecure}, such randomness is necessary when the security function $g$ is the identity function. For general $g$, if randomness is not required, we simply set $\mathcal{K}_i=\emptyset$. Let $K_S=(K_1,K_2,\cdots,K_s)$. All the random keys $K_i,i\in[S]$ and the source messages $M_i,i\in[s]$ are mutually independent. 

An $(\ell,n)$ secure network code for the model $(\cN,f,g,\cW)$ is defined as follows. First, for every $i\in[s]$, let $\mathbf{m}_i\in\cA^\ell$ and $\bk_i\in\mathcal{K}_i$ be arbitrary realizations of the source message $M_i$ and the random key $K_i$, respectively. Accordingly, let $\mathbf{m}_S=(\mathbf{m}_1,\mathbf{m}_2,\cdots,\mathbf{m}_s)$ and $\bk_S=(\bk_1,\bk_2,\cdots,\bk_s)$, which can be regarded as arbitrary realizations of $M_S$ and $K_S$. We assume that the message transmitted by each edge is over the finite alphabet $\cB$.
An $(\ell,n)$ network code $\mathbf{\widehat{C}}$ consists of
\begin{itemize}
    \item a \emph{local encoding function} $\widehat{\theta}_e$ for each edge $e\in\mathcal{E}$ such that 
    \begin{equation}\label{eq:localencodingfunction}
        \widehat{\theta}_e:
\begin{cases}
 \mathcal{A}^\ell\times\mathcal{K}_i\mapsto\mathcal{B}^n, & \mbox{if $\tail(e)= \sigma_i$ for some $i$}; \\
    \prod\limits_{d\in \In(\tail(e))}\mathcal{B}^n\mapsto\mathcal{B}^n, & \mbox{otherwise},
  \end{cases}
    \end{equation}
  \item a \emph{decoding function} $\widehat{\varphi}:\prod_{\In(\gamma)}\cB^n\rightarrow\cO^\ell$ at the sink node $\gamma$, which is used to compute the target function $f$ with zero error.
\end{itemize}

Let $\by_e\in\cB^n$ denote the message transmitted on each edge $e\in\cE$ by using the code $\mathbf{\widehat{C}}$ with the source message vector $\mathbf{m}_S$ and the key vector $\bk_S$. From (\ref{eq:localencodingfunction}), it follows that $\by_e$ depends on $(\mathbf{m}_S,\bk_S)$. We write this dependence as $\by_e = \widehat{\eta}_e(\mathbf{m}_S,\bk_S)$. The function  $\widehat{\eta}_e$ can be obtained by recursively applying the local encoding functions $\widehat{\theta}_e, e \in \cE$. More precisely, for each $e\in\cE$,
\[\widehat{\eta}_e(\mathbf{m}_S,\bk_S)=
\begin{cases}
 \widehat{\theta}_e(\mathbf{m}_i,\bk_i), & \mbox{if $\tail(e)= \sigma_i$ for some $i$}; \\
    \widehat{\theta}_e(\widehat{\eta}_{\In(u)}(\mathbf{m}_S,\bk_S)), & \mbox{otherwise},
  \end{cases}\]
where $u = \tail(e)$ and $\widehat{\eta}_E(\mathbf{m}_S,\bk_S)=(\widehat{\eta}_e(\mathbf{m}_S,\bk_S):e\in E)$ for an edge subset $E\subseteq\cE$. We call $\widehat{\eta}_e$ the
\emph{global encoding function} of the edge $e$ for the code $\hbc$.

For the model $(\cN,f,g,\cW)$, an $(\ell, n)$ secure network code $\hbc = \{\htheta_e: e \in \cE\}\cup\{\hvarphi\}$ is called \emph{admissible} if it satisfies the following two conditions:
\begin{itemize}
    \item \textbf{Decodability}: The sink node $\gamma$ computes the target function $f$ with zero error, i.e.,
    \[\hvarphi(\heta_{\In(\gamma)}(\mathbf{m}_S,\bk_S))=f(\mathbf{m}_S),\ \ \forall\ \mathbf{m}_S\in\cA^{s\ell}\textup{ and } \bk_S\in\prod_{i=1}^s\mathcal{K}_i;\]
    \item \textbf{Function security}: For every wiretap set $W \in \cW$, the messages on $W$ must leak no information about $g(M_S)$, i.e.,
    \begin{equation}\label{eq:securitycondition}
        I(Y_W;g(M_S))=0,
    \end{equation}
    where $Y_W = (Y_e : e \in W)$ with $Y_e \eqdef \heta_e(M_S, K_S)$ being the random vector transmitted on the edge $e$, and $I(\cdot)$ denotes mutual information.
\end{itemize}
The \emph{secure computing rate} of an admissible $(\ell, n)$ secure network code $\hbc$ is defined as 
\[R(\hbc)=\frac{\ell}{n},\]
which measures how many function evaluations can be computed per use of the network, while protecting $g(M_S)$ from a wiretapper who may observe any single set $W\in\cW$. 
A nonnegative real number $R$ is said to be \emph{achievable} if for every $\epsilon>0$, there exists an admissible $(\ell,n)$ secure network code $\hbc$ for the model $(\cN,f,g,\cW)$ such that
\[R(\hbc)=\frac{\ell}{n}>R-\epsilon.\]
Accordingly, the secure computing rate region for the secure model $(\cN,f,g,\cW)$ is defined as
\[\mathcal{R}(\cN,f,g,\cW)=\left\{R:R \textup{ is achievable for } (\cN,f,g,\cW)\right\},\]
and the \emph{secure computing capacity} for $(\cN,f,g,\cW)$ is defined as
\[\widehat{\cC}(\cN,f,g,\cW)\eqdef \max \mathcal{R}(\cN,f,g,\cW).\]
In this paper, we focus on the case when $\cW=\{E\in\cE:|E|\leq r\}$. For simplicity, we denote the model as $(\cN,f,g,r)$.

In prior works \cite{2024Guangsourcesecure,2025Guangfunctionsecure}, Guang, Bai and Yeung studied secure network function computation for the case where the target function is the algebraic sum.  They defined the general model of secure network function and investigated the problem under two different types of security function, which is referred to as \emph{source security} and \emph{target function security}, respectively. Specifically, source security requires that the wiretapper obtains no information about the source messages $M_S$,
i.e., for any $W \in \mathcal{W}$,  
\[I\big(M_S; Y_W\big) = 0.\]
Target function security requires that the wiretapper learns nothing about $f(M_S)$, that is,
\[	I\big(f(M_S); Y_W\big) = 0.\] 
In \cite{2024Guangsourcesecure}, Guang, Bai and Yeung derived an upper bound on the secure computing capacity $\widehat{\cC}(\mathcal{N},f,g,r)$ and constructed a coding scheme that achieves this bound for certain network topologies. To state their result, we first introduce some notation. For an edge subset $C$, define three subsets of source nodes:
\begin{align*}
    D_C&\eqdef\{\sigma\in S:\text{there exists a directed path from $\sigma$ to $\gamma$ containig at least one edge in $C$}\};\\
    I_C   &\eqdef \set{\sigma\in S :  \mbox{there is no path from $\sigma$ to $\gamma$ after deleting the edges in $C$ from $\cE$}};\\
    J_C&\eqdef D_C\setminus I_C.
\end{align*}
The bounds for the secure computing capacity under source security are then given as follows.
\begin{thm}[{\cite[Theorem~1,10]{2024Guangsourcesecure}}]\label{thm:ub-source-secure}
	Consider a model of secure network function computation $(\cN,f,g,r)$, where $f$ is an algebraic sum function over a finite field and the $g$ is the identity function. Then,
	\[C_{\min}-r\leq \widehat{\cC}(\mathcal{N},f,g,r)\leq\min_{\substack{{\rm all\ pairs}\ (W,C)\in\mathcal{W}\times\Lambda(\mathcal{N}):\\W\subseteq C\ {\rm and}\ D_W\subseteq I_C}}(|C|-|W|),\]
	where $C_{\min}=\min\{|C|:C\in\Lambda(\mathcal{N})\}$.
\end{thm}

In \cite{2025Guangfunctionsecure}, a pair $(C,W)$ is called valid if it satisfies at least one of the following two conditions:
\begin{itemize}
    \item $I_C\setminus D_W\neq\emptyset$;
    \item $I_C\setminus D_W=\emptyset$ and $I_C=S$.
\end{itemize}
With this definition, the following upper bound was established for the case of target function security.
\begin{thm}[{\cite[Theorem~1,8]{2025Guangfunctionsecure}}]
    Consider a model of secure network function computation $(\cN,f,g,r)$, where $f$ is an algebraic sum function over a finite field and the $g=f$. Then,
	\[C_{min}-r\leq \widehat{\cC}(\mathcal{N},f,g,r)\leq\min_{{\rm all\ \text{valid}\ pairs}\ (C,W)}(|C|-|W|),\]
	where $C_{min}=\min\{|C|:C\in\Lambda(\mathcal{N})\}$.
\end{thm}

\subsection{Network function computation for linear functions}

In this subsection, we review prior results for network function computation without security constraints. When security requirements are removed, the use of random keys at the source nodes becomes unnecessary. The definitions of network code, admissible code, computing rate, and computing capacity can be adapted straightforwardly from the secure setting.
In \cite{2018HTYG,2019GYYL}, Guang {\it et al.} derived two general upper bounds on the computing capacity that hold for any network and any target function. Here we focus on the case where the target function is linear. A target function $f$ is called \emph{linear} if it has the form $f(\bx)=\bx \cdot\mathbf{F}$, where $\bx\in\F_q^{1\times s}$ and $\bF\in\F_q^{s\times r_f}$. 

A network code is called \emph{linear} if the message transmitted by each edge $e$ is a linear combination of the messages received by $\tail(e)$. Specifically, in an $(\ell,n)$ linear network code over $\F_q$,  
the message $\bu_e \in \F_q^{n\times 1}$ transmitted via edge $e$ has the form
\begin{equation}\label{eq:linlocalenc}
  \bu_e=
  \begin{cases}
    \sum\limits_{j=1}^{\ell}  m_{ij}\ba_{(i,j),e}, & \mbox{if $\tail(e)= \sigma_i$ for some $i$}; \\
    \sum\limits_{d\in \In(\tail(e))} \bu_d\mathbf{A}_{d,e}, & \mbox{otherwise},
  \end{cases}
\end{equation}
where  $\ba_{(i,j),e} \in \F_q^{1\times n},\mathbf{A}_{d,e} \in \F_{q}^{n\times n}$, and $\ba_{(i,j),e}$ is an all-zero vector  if $e$ is not an outgoing edge of some source node $\sigma_i\in S$ and   $\mathbf{A}_{d,e} $ is an all-zero matrix if $e$ is not an outgoing edge of $\head(d)$. Consequently, each $\bu_e$ can be expressed as a linear combination of the source messages:
\begin{equation*}
  \bu_e =\mathbf{m}_S\cdot \mathbf{H}_e,
\end{equation*}
where $ \mathbf{H}_e \in \F_q^{s\ell \times n}$.
If there exists a decoding function $\varphi :\prod_{\In(\gamma)} \F_q^n \to \F_q^{\ell r_f}$ such that for every $\mathbf{m}_S \in \F_q^{s\ell}$,
\begin{equation*}
    \varphi\parenv*{(\bu_e : e\in \In(\gamma))} =\mathbf{m}_S\cdot(\bF\otimes \bI_\ell),
\end{equation*}
where $\bI_\ell$ is the $\ell\times \ell$ identity matrix and $\otimes$ denotes the Kronecker product, then the linear network code enables the sink node to compute the target function exactly $\ell$ times. 

For a subset $ A \subseteq S$, let $\mathbf{F}_A$ denote an $s\times r_f$ matrix obtained by replacing all the rows corresponding to the source nodes in $S\setminus A$ by zeros in $\bF$.
For instance, let $S=\{\sigma_1,\sigma_2,\sigma_3\}$ and $A=\{\sigma_1,\sigma_2\}$. If the sink node wants to compute $m_1+m_2$ and $m_2+m_3$, then the
coefficient matrix $\bF$ and corresponding $\bF_A$ are
\[\bF^T=\begin{bmatrix}
    1&1&0\\0&1&1
\end{bmatrix},~~~\bF_A^T=\begin{bmatrix}
    1&1&0\\0&1&0
\end{bmatrix}.\]
As demonstrated in \cite{2023WeiXuGe}, the following upper bound can be obtained from the general bound in \cite{2018HTYG}.

\begin{thm}\label{thm:csbnd}[{\cite[Corollary~II.1]{2023WeiXuGe}}]
Given a network $\mathcal{N}$ and a linear target function $f(\bx)=\bx\cdot\mathbf{F}$.  If there exists a linear network code $\cC$ computing $f$ with rate $w$, then necessarily
\[ w  \leq  \min_{C\in \Lambda(\cN)} \frac{\abs{C}}{  {\rm{Rank}}(\mathbf{F}_{I_C})}.\]
\end{thm}

For the special case $n=\ell=1$, Appuswamy and Franceschetti  \cite{2014Appuswamy} proved that the condition which appears in Theorem~\ref{thm:csbnd}, is also sufficient when  $r_f\in\{1,s-1,s\}$.\footnote{For $r_f=s-1$, the sufficiency additionally requires $\mathbf{F}^T\sim\begin{pmatrix}
    \mathbf{I}_{s-1}&\mathbf{b}
\end{pmatrix}$, where $\mathbf{b}$ is a vector with no zero components.}

\subsection{Secure network coding}

In this subsection, we present the secure network coding model $(\cN, g, r)$ and the results in the secure network coding problem. In this setting, the network $\cN = (\cG, s, T)$ has a single source node $s$ and a set of sink nodes $T \subseteq \mathcal{V} \setminus {s}$.
We retain the notation $\tail(e)$, $\head(e)$, $\In(v)$, and $\Out(v)$ as defined previously.
The source $s$ generates a random source message $M=(M_1, M_2, \cdots, M_\ell)$ consisting of $\ell$ i.i.d. symbols uniformly drawn from a finite alphabet $\cA$. The parameter $\ell$ is called the \emph{information rate}. The goal is to multicast $M$ to all sinks $t \in T$ by using the network $\mathcal{N}$ multiple times.
As same as secure network function computation, there is a wiretapper who can eavesdrop any one edge set $W$ of size at most $r$, where $r$ is the \emph{security level}. The wiretapper is not allowed to obtain any information about $g(M)$. The overall model is denoted by $(\cN, g, r)$.
To achieve security, the source employs a random key $K$, uniformly distributed over a finite set $\mathcal{K}$, which is independent of $M$.

An $(\ell, n)$ secure network code for $(\cN, g, r)$ is defined as follows. Let $\mathbf{m} \in \cA^\ell$ and $k \in \mathcal{K}$ be realizations of $M$ and $K$, respectively. The code $\widehat{\mathbf{C}}$ consists of:
\begin{itemize}
    \item a \emph{local encoding function} $\widehat{\theta}_e$ for each edge $e\in\mathcal{E}$ such that 
    \begin{equation}
        \widehat{\theta}_e:
\begin{cases}
 \mathcal{A}^\ell\times\mathcal{K}\mapsto\mathcal{B}^n, & \mbox{if $\tail(e)= s$}; \\
    \prod\limits_{d\in \In(\tail(e))}\mathcal{B}^n\mapsto\mathcal{B}^n, & \mbox{otherwise},
  \end{cases}
    \end{equation}
  \item a \emph{decoding function} $\widehat{\varphi}_t:\prod_{\In(\gamma)}\cB^n\rightarrow\cA^\ell$ at each sink $t \in T$, which recovers $M$ with zero error.
\end{itemize}

Let $Y_e = \widehat{\eta}_e(M, K)$ be the random vector transmitted on edge $e$, where $\widehat{\eta}_e$ is the corresponding global encoding function obtained recursively from ${\widehat{\theta}_e}$. For an edge set $W$, denote $Y_W = (Y_e : e \in W)$.
The code $\widehat{\mathbf{C}}$ is admissible if it satisfies:
\begin{itemize}
    \item \textbf{Decodability}: For every sink $t \in T$, $\widehat{\varphi}_t\bigl(\widehat{\eta}_{\In(t)}(M, K)\bigr) = M$;
    \item \textbf{Security}: For every wiretap set $W$ with $|W| \leq r$, $I(Y_W;g(M))=0$.
\end{itemize}

When the security function is linear, i.e., $g(\bx) = \bx\cdot \mathbf{G}$, where $\bG\in\mathbb{F}_q^{\ell\times r_g}$ and $\Rank(\mathbf{G}) = r_g$, Bai, Guang and Yeung derived the following capacity result \cite{bai2023multiple}.
\begin{thm}[{\cite[Theorem~1]{bai2023multiple}}]
Consider the secure network coding model $(\cN, g, r)$ over $\F_q$, where $g(\bx)=\bx\cdot\mathbf{G}$. Let $\tau = r_g / \ell$, $\tau_0 = (C_{\min}-r)/C_{\min}$ with $0 < r < C_{\min}$, and assume $q> \max\{|T|, \binom{|\cE|}{r}\}$.
\begin{itemize}
\item If $0 \leq \tau \leq \tau_0$, then \[C(\cN,r,g) = \dfrac{l}{\big\lceil\frac{l}{C_{\min}}\big\rceil};\]
\item If $\tau_0 \leq \tau \leq 1$, then \[C(\cN,r,g) = \dfrac{l}{\big\lceil\frac{\tau l}{C_{\min}-r}\big\rceil}.\]
\end{itemize}
\end{thm}

In the following, we clarify the relationship between secure network coding and secure network function computation. Consider a secure network function computation model $(\cN,f,g,r)$ where $f(\bx)=\bx\cdot \bI_s$ and $g(\bx)=\bx\cdot\bG$, with $\bx\in\mathbb{F}_q^{s}$ and $\bG\in\mathbb{F}_q^{s\times r_g}$. The network $\cN$ has $s$ source nodes. Since secure network coding is commonly formulated for a single source, we transform $\cN$ into $\cN'$ by introducing a virtual source node $\sigma_0$ and connecting it to each original source $\sigma_i$, $i \in [s]$ via a perfectly secure edge of unlimited capacity. In the original model, each $\sigma_i$ generates $\ell$ independent symbols $M_i$ along with random keys. In $\cN'$, the message generation is centralized, that is, $\sigma_0$ produces a message $M_0 = (M_1, \dots, M_s)$ of length $s\ell$, as well as the required randomness. Accordingly, the target function becomes $f'=\bx'\cdot \bI_{s\ell}$ and the secure function becomes $g'=\bx'\cdot\bG'$, where $\bx'\in\mathbb{F}_q^{s\ell}$ and $\bG'=\bG\otimes\bI_\ell\in\mathbb{F}_q^{s\ell\times r_g\ell}$.
Based on the above network transformation, an $(\ell, n)$ secure network code for $(\cN, f, g, \mathcal{W})$ can be directly extended to an $(s\ell, n)$ secure network code for $(\cN', g', \mathcal{W})$. Note that the wiretap set $\mathcal{W}$ remains unchanged since the newly added edges are perfectly secure. Consequently, the following inequality must be satisfied:
\begin{equation}\label{eq:SNC}
    \frac{s\ell}{n} \leq \widehat{\mathcal{C}}(\cN', g', \mathcal{W}) \leq \min\left\{ |C|, \frac{s}{r_g}(|C| - r) \right\}.
\end{equation}

\section{Upper bounds for secure network computing capacity}\label{sec:upper-bounds}
In this section, we establish two upper bounds on the secure network computing capacity $\widehat{\cC}(\cN,f,g,r)$. The key idea is to identify an edge set $C$ (or $C \cup B$) with the following property: given certain source messages and random keys, the messages transmitted over $C$ contain the information the wiretapper wants to learn. At the same time, when provided with the same source messages and keys, any wiretap set $W \subseteq C$ (or $W \subseteq C \cup B$) must not leak any information about that same information to the wiretapper.

We first present a useful lemma.
\begin{lem}\label{lem:globalcut}
    For the secure network function computation problem $(\cN,f,g,r)$, if there exists a secure $(\ell,n)$ network code, then for any cut set $C\in\Lambda(\cN)$, we have
    \[H(M_S(\mathbf{F}_{I_C}\otimes \bI_\ell)|Y_C,M_{S\setminus I_C},K_{S\setminus I_C})=0.\]
\end{lem}
\begin{pf}
    Let $C'=\bigcup_{i\in S\setminus I_C} \Out(\sigma_i)$. Since $D_{C'}=S\setminus I_C$, the message $Y_{C'}$ is determined by $M_{S\setminus I_C}$ and $K_{S\setminus I_C}$. Hence, we have
    \begin{equation}\label{seq:lemma1} 
        H(M_S(\mathbf{F}_{I_C}\otimes \bI_\ell)|Y_C,M_{S\setminus I_C},K_{S\setminus I_C})=H(M_S(\mathbf{F}_{I_C}\otimes \bI_\ell)|Y_C,Y_{C'},M_{S\setminus I_C},K_{S\setminus I_C}).
    \end{equation}
    Note that $C\cup C'$ forms a global cut of the network, that is, $C\cup C'$ separates $\gamma$ from all source nodes. Consequently, the message received by the sink node $Y_{\In(\gamma)}$ is a function of $Y_C$ and $Y_{C'}$. By the decodability of the network code, the target function $M_S(\mathbf{F}\otimes \bI_\ell)$ can be recovered from $Y_C$ and $Y_{C'}$ with zero error, so
    \[H(M_S(\mathbf{F}\otimes \bI_\ell)|Y_C,Y_{C'})=0.\]
    Returning to \eqref{seq:lemma1}, we have
    \begin{subequations}
        \begin{align*}
            &H(M_S(\mathbf{F}_{I_C}\otimes \bI_\ell)|Y_C,M_{S\setminus I_C},K_{S\setminus I_C})\\
            =&H(M_S(\mathbf{F}_{I_C}\otimes \bI_\ell)|Y_C,Y_{C'},M_{S\setminus I_C},K_{S\setminus I_C})\\
            =&H(M_S(\mathbf{F}_{I_C}\otimes \bI_\ell)|Y_C,Y_{C'},M_S(\mathbf{F}\otimes \bI_\ell),M_{S\setminus I_C},K_{S\setminus I_C})=0,
        \end{align*}
    \end{subequations}
    where the last equality is because $M_S(\mathbf{F}\otimes\bI_\ell)$ equals to the sum of $M_S(\mathbf{F}_{I_C}\otimes \bI_\ell)$ and $M_S(\mathbf{F}_{S\setminus I_C}\otimes \bI_\ell)$, and $M_S(\mathbf{F}_{S\setminus I_C}\otimes \bI_\ell)$ is determined by $M_{S\setminus I_C}$.
\end{pf}

Next, we introduce the notion of a weak partition for a cut.
\begin{dfn}(Weak partition)
    A partition $\mathcal{P}_C = \{C_1, C_2, \dots, C_p\}$ of a cut $C$ is called a weak partition if
\begin{itemize}
    \item $I_{C_i}\neq \varnothing$, $\forall$ $1\leq i\leq p$;
    \item $D_{C_i}\subseteq I_{C_i}\cup J_C\cup L$, where $L\eqdef I_C\setminus(\cup_{i\in[p]}I_{C_i})$.
\end{itemize}
\end{dfn}
From the definition, the messages transmitted on each part $C_i$ depend only on the source messages and keys from nodes in $I_{C_i} \cup J_C \cup L$. A trivial weak partition is the whole cut itself, i.e., $\mathcal{P}_C = \{C\}$.
It is worth to note that the definition of weak partition differs from the \emph{strong partition} defined in \cite{2019GYYL}. A partition $\mathcal{P}_C=\{C_1,C_2,\cdots,C_p\}$ of a cut $C$ is called a strong partition, if
\begin{itemize}
    \item $I_{C_i}\neq \varnothing$, $\forall$ $1\leq i\leq p$;
    \item $D_{C_i}\cap I_{C_j}=\emptyset$, for any $i\neq j\in[s]$.
\end{itemize}
It is clear that every strong partition is also a weak partition, but the converse is not always true.
For instance, in Fig.~\ref{fig:weak_partition}, the cut $C=\{e_{3,1},e_{3,2},e_{1},e_{2}\}$ has a weak partition $\mathcal{P}_C=\{C_1,C_2\}$, where $C_1=\{e_{3,1},e_{3,2}\}$ and $C_2=\{e_{1},e_{2}\}$. Here $D_{C_1} = I_{C_1} = \{\sigma_3\}$ and $D_{C_2} = I_{C_2} = S$, satisfying the weak partition conditions. However, $\mathcal{P}_C$ is not a strong partition since $D_{C_2}\cap I_{C_1}\neq \emptyset$. 

\begin{figure}
    \centering
    \includegraphics[width=0.5\linewidth]{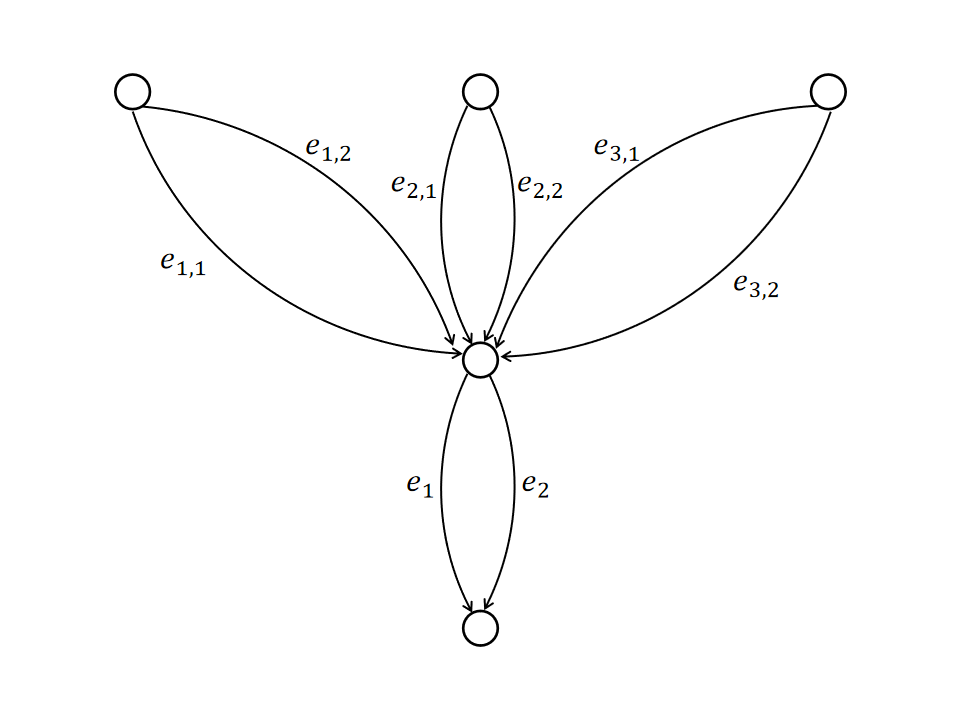}
    \caption{A three-layer network with $3$ source nodes.}
    \label{fig:weak_partition}
\end{figure}

\begin{dfn}
    For a cut set $C\in\Lambda(\cN)$ and a wiretapped set $W\in\cW$, the pair $(C,W)$ is called valid if the following conditions are satisfied:
\begin{enumerate}
    \item There exists a weak partition $\cP_C$ of $C$, $\cP_C=\{C_1,\cdots,C_p\}$ such that $\langle\mathbf{F}_{I_{C_1}}\rangle\bigoplus\cdots\bigoplus\langle\mathbf{F}_{I_{C_p}}\rangle$ $\bigcap\langle\mathbf{G}\rangle\neq\{\Zero\}$.\footnote{Here, we use $\langle\bG\rangle$ to denote the subspace spanned by the columns of a matrix $\bG$.}
    \item $W\subseteq C$, and $D_W\subseteq I_{C_1}\cup I_{C_2}\cup\cdots\cup I_{C_p}.$
\end{enumerate}
\end{dfn}

With these definitions, we can state the first upper bound on the secure computing capacity.
\begin{thm}\label{thm:upperbound1}
    For the secure network function computation problem $(\cN,f,g,r)$, the secure network computing capacity satisfies \[\widehat{\cC}(\cN,f,g,r)\leq\min_{(C,W)\ \text{valid}}\frac{|C|-|W|}{t_{C,f,g}},\]
    where $C$ has a weak partition $C_1,\cdots,C_p$, and $t_{C,f,g}=\dim\left(\langle\mathbf{F}_{I_{C_1}}\rangle\bigoplus\cdots\langle\mathbf{F}_{I_{C_p}}\rangle\bigcap\langle\mathbf{G}\rangle\right).$
\end{thm}
\begin{pf}
    Suppose that there exists a secure $(\ell,n)$ network code. For every valid pair $(C,W)$, we assume that $C=C_1\cup\cdots\cup C_p$ is a weak partition, $I_C=I_{C_1}\cup I_{C_2}\cup\cdots\cup I_{C_p}\cup L$. Let $\bP$ be an $s\times t$ matrix such that $\langle \bP\rangle\eqdef\langle \bF_{I_{C_1}}\rangle\bigoplus\cdots\bigoplus\langle \bF_{I_{C_p}}\rangle$ $\bigcap\langle \bG\rangle$, and $t=\dim(\langle \bP\rangle)$. From the definition of $\bP$, $M_S(\bP\otimes \bI_{\ell})$ is a function of $M_S(\bF_{I_{C_1}}\otimes\bI_\ell),\cdots,M_S(\bF_{I_{C_p}}\otimes\bI_\ell)$. Hence,
    \begin{equation*}
         I(M_S(\bP\otimes \bI_{\ell});M_{S\setminus \cup_{j\in[p]}I_{C_j}},K_{S\setminus \cup_{j\in[p]}I_{C_j}})=0.
    \end{equation*}
    From the security condition,
    \[I(M_S(\bP\otimes \bI_{\ell});Y_W)\leq I(M_S(\bG\otimes \bI_{\ell});Y_W)=0.\]
    Consequently, we have
    \begin{equation}\label{eq:ubpf_1}
         H(M_S(\bP\otimes \bI_{\ell}))+H(Y_W)=H(M_S(\bP\otimes \bI_{\ell}),Y_W).
    \end{equation}
    Since $D_W\subseteq\cup_{j\in[p]}I_{C_j}$, the joint entropy $H(M_S(\bP\otimes \bI_{\ell}),Y_W)$ can be decomposed as
    \begin{subequations}\label{seq:ubpf_2}
        \begin{align}
           & H(M_S(\bP\otimes \bI_{\ell}),Y_W)=H(M_S(\bP\otimes \bI_{\ell}),Y_W|M_{S\setminus \cup_{j\in[p]}I_{C_j}},K_{S\setminus \cup_{j\in[p]}I_{C_j}})\\
            =&H(Y_W|M_{S\setminus \cup_{j\in[p]}I_{C_j}},K_{S\setminus \cup_{j\in[p]}I_{C_j}})+H(M_S(\bP\otimes \bI_{\ell})|Y_W,M_{S\setminus \cup_{j\in[p]}I_{C_j}},K_{S\setminus \cup_{j\in[p]}I_{C_j}})\\
            =&H(Y_W)+H(M_S(\bP\otimes \bI_{\ell})|Y_W,M_{S\setminus \cup_{j\in[p]}I_{C_j}},K_{S\setminus \cup_{j\in[p]}I_{C_j}}),
        \end{align}
    \end{subequations}
    where the first equality is because that $M_S(\bP\otimes \bI_{\ell})$ and $Y_W$ are independent of $M_{S\setminus \cup_{j\in[p]}I_{C_j}},K_{S\setminus \cup_{j\in[p]}I_{C_j}}$.
    From (\ref{eq:ubpf_1}) and (\ref{seq:ubpf_2}), we obtain
    \begin{equation}\label{eq:ubpf_key_1}
        H(M_S(\bP\otimes \bI_{\ell}))=H(M_S(\bP\otimes \bI_{\ell})|Y_W,M_{S\setminus \cup_{j\in[p]}I_{C_j}},K_{S\setminus \cup_{j\in[p]}I_{C_j}}).
    \end{equation}
    
    Next, we show that, given $M_{S\setminus \cup_{j\in[p]}I_{C_j}}$ and $K_{S\setminus \cup_{j\in[p]}I_{C_j}}$, the message $Y_C$ completely determines $M_S(\bP\otimes \bI_{\ell})$.
    Since $\{C_1,\cdots,C_p\}$ is a weak partition, for each $j\in[p]$, the messages on $C_j$ depend only on messages and keys generated by the source nodes in $I_{C_j},J_C$ and $L$. Consequently,
\begin{subequations}\label{seq:ubpf_3}
    \begin{align}
        &I\!\left(M_S(\mathbf{F}_{I_{C_j}}\!\otimes \bI_\ell);~ M_{\bigcup_{j'\neq j}I_{C_{j'}}},K_{\bigcup_{j'\neq j}I_{C_{j'}}}\Big|\; M_{S\setminus\bigcup_{j'}I_{C_{j'}}},K_{S\setminus\bigcup_{j'}I_{C_{j'}}},Y_{C_j}\right)\\
        \leq &I\!\left(M_S(\bF_{I_{C_j}}\otimes\bI_\ell),M_{I_{C_j}},K_{I_{C_j}};M_{\bigcup_{j'\neq j}I_{C_{j'}}},K_{\bigcup_{j'\neq j}I_{C_{j'}}}\Big|M_{S\setminus\bigcup_{j'}I_{C_{j'}}},K_{S\setminus\bigcup_{j'}I_{C_{j'}}},Y_{C_j}\right)\\
        =&I\!\left(M_{I_{C_j}},K_{I_{C_j}};M_{\bigcup_{j'\neq j}I_{C_{j'}}},K_{\bigcup_{j'\neq j}I_{C_{j'}}}\Big|M_{S\setminus\bigcup_{j'}I_{C_{j'}}},K_{S\setminus\bigcup_{j'}I_{C_{j'}}},Y_{C_j}\right)\\
        =&H\!\left(M_{\bigcup_{j'\neq j}I_{C_{j'}}},K_{\bigcup_{j'\neq j}I_{C_{j'}}}\Big|M_{S\setminus\bigcup_{j'}I_{C_{j'}}},K_{S\setminus\bigcup_{j'}I_{C_{j'}}},Y_{C_j}\right)\nonumber\\&-H\!\left(M_{\bigcup_{j'\neq j}I_{C_{j'}}},K_{\bigcup_{j'\neq j}I_{C_{j'}}}\Big|M_{S\setminus\bigcup_{j'}I_{C_{j'}}},K_{S\setminus\bigcup_{j'}I_{C_{j'}}},Y_{C_j},M_{I_{C_j}},K_{I_{C_j}}\right)\\
        =&H\!\left(M_{\bigcup_{j'\neq j}I_{C_{j'}}},K_{\bigcup_{j'\neq j}I_{C_{j'}}}\Big|M_{S\setminus\bigcup_{j'}I_{C_{j'}}},K_{S\setminus\bigcup_{j'}I_{C_{j'}}},Y_{C_j}\right)\nonumber\\&-H\!\left(M_{\bigcup_{j'\neq j}I_{C_{j'}}},K_{\bigcup_{j'\neq j}I_{C_{j'}}}\Big|M_{S\setminus\bigcup_{j'}I_{C_{j'}}},K_{S\setminus\bigcup_{j'}I_{C_{j'}}},M_{I_{C_j}},K_{I_{C_j}}\right)\\
        =&H\!\left(M_{\bigcup_{j'\neq j}I_{C_{j'}}},K_{\bigcup_{j'\neq j}I_{C_{j'}}}\Big|M_{S\setminus\bigcup_{j'}I_{C_{j'}}},K_{S\setminus\bigcup_{j'}I_{C_{j'}}},Y_{C_j}\right)-H\!\left(M_{\bigcup_{j'\neq j}I_{C_{j'}}},K_{\bigcup_{j'\neq j}I_{C_{j'}}}\right)\\
        \leq&H\!\left(M_{\bigcup_{j'\neq j}I_{C_{j'}}},K_{\bigcup_{j'\neq j}I_{C_{j'}}}\right)-H\!\left(M_{\bigcup_{j'\neq j}I_{C_{j'}}},K_{\bigcup_{j'\neq j}I_{C_{j'}}}\right)=0,
    \end{align}
\end{subequations}
where (\ref{seq:ubpf_3}b) follows from a basic property of conditional mutual information, (\ref{seq:ubpf_3}c) holds because $M_S(\mathbf{F}_{I_{C_j}}\otimes \bI_\ell)$ is a function of $M_{I_{C_j}}$, (\ref{seq:ubpf_3}e) is due to the fact that messages on the cut set $C_j$ depend only on sources in $I_{C_j}\cup J_C\cup L$, which is a subset of $I_{C_j}\cup (S\setminus\bigcup_{j'}I_{C_{j'}})$,  (\ref{seq:ubpf_3}f) follows from the  independence of all source messages ${M_i}$ and keys ${K_i}$, and (\ref{seq:ubpf_3}g) results from the monotonicity of conditional entropy.
Moreover,  
\begin{subequations}
    \begin{align}
        &H(M_S(\bF_{I_{C_j}}\otimes\bI_\ell)|M_{S\setminus\bigcup_{j'}I_{C_{j'}}},K_{S\setminus\bigcup_{j'}I_{C_{j'}}},Y_{C_j})\\
        =&H(M_S(\bF_{I_{C_j}}\otimes\bI_\ell)|M_{S\setminus I_{C_j}},K_{S\setminus I_{C_j}},Y_{C_j})\nonumber\\&+I(M_S(\bF_{I_{C_j}}\otimes\bI_\ell);M_{\bigcup_{j'\neq j}I_{C_{j'}}},K_{\bigcup_{j'\neq j}I_{C_{j'}}}|M_{S\setminus\bigcup_{j'}I_{C_{j'}}},K_{S\setminus\bigcup_{j'}I_{C_{j'}}},Y_{C_j})\\
        \overset{(\ref{seq:ubpf_3})}{=}&H(M_S(\bF_{I_{C_j}}\otimes\bI_\ell)|M_{S\setminus I_{C_j}},K_{S\setminus I_{C_j}},Y_{C_j})=0,
    \end{align}
\end{subequations}
where the last equality follows from Lemma~\ref{lem:globalcut}.
Summing over $j\in[p]$, we have
\begin{subequations}\label{seq:ubpf_4}
    \begin{align}
        &H(\{M_S(\bF_{I_{C_j}}\otimes\bI_\ell)\}_{j\in[p]}|Y_C,M_{S\setminus\bigcup_{j'}I_{C_{j'}}},K_{S\setminus\bigcup_{j'}I_{C_{j'}}})\\
        \leq &\sum_{j\in[p]}H(M_S(\bF_{I_{C_j}}\otimes\bI_\ell)|Y_{C},M_{S\setminus\bigcup_{j'}I_{C_{j'}}},K_{S\setminus\bigcup_{j'}I_{C_{j'}}})\\
        \leq &\sum_{j\in[p]}H(M_S(\bF_{I_{C_j}}\otimes\bI_\ell)|Y_{C_j},M_{S\setminus\bigcup_{j'}I_{C_{j'}}},K_{S\setminus\bigcup_{j'}I_{C_{j'}}})=0,
    \end{align}
\end{subequations}
where the two inequalities are due to the sub-additivity and monotonicity of conditional entropy, respectively.
Note that $M_S(\bP\otimes \bI_{\ell})$ is determined by $M_S(\bF_{I_{C_1}}\otimes\bI_\ell),\cdots,M_S(\bF_{I_{C_p}}\otimes\bI_\ell)$, we have 
\begin{subequations}\label{seq:ubpf_key_2}
    \begin{align}
        &H(M_S(\bP\otimes \bI_{\ell})|Y_C,M_{S\setminus\bigcup_{j'}I_{C_{j'}}},K_{S\setminus\bigcup_{j'}I_{C_{j'}}})\\
        \leq&H(M_S(\bP\otimes \bI_{\ell}),\{M_S(\bF_{I_{C_j}}\otimes\bI_\ell)\}_{j\in[p]}|Y_C,M_{S\setminus\bigcup_{j'}I_{C_{j'}}},K_{S\setminus\bigcup_{j'}I_{C_{j'}}})\\
        =&H(\{M_S(\bF_{I_{C_j}}\otimes\bI_\ell)\}_{j\in[p]}|Y_C,M_{S\setminus\bigcup_{j'}I_{C_{j'}}},K_{S\setminus\bigcup_{j'}I_{C_{j'}}})\overset{(\ref{seq:ubpf_4})}{=}0,
    \end{align}
\end{subequations}
From (\ref{eq:ubpf_key_1}) and (\ref{seq:ubpf_key_2}), we have
\begin{subequations}\label{seq:ub1-last}
    \begin{align}
        \ell t\log q=&H(M_S(\bP\otimes \bI_{\ell}))\\
        \overset{\eqref{eq:ubpf_key_1}}{=}&H(M_S(\bP\otimes \bI_{\ell})|Y_W,M_{S\setminus \cup_{j\in[p]}I_{C_j}},K_{S\setminus \cup_{j\in[p]}I_{C_j}})\\
        \overset{\eqref{seq:ubpf_key_2}}{=}&H(M_S(\bP\otimes \bI_{\ell})|Y_W,M_{S\setminus \cup_{j\in[p]}I_{C_j}},K_{S\setminus \cup_{j\in[p]}I_{C_j}})\\&-H(M_S(\bP\otimes \bI_{\ell})|Y_C,M_{S\setminus\bigcup_{j'}I_{C_{j'}}},K_{S\setminus\bigcup_{j'}I_{C_{j'}}})\\
        =&I(M_S(\bP\otimes \bI_{\ell});Y_{C\setminus W}|M_{S\setminus\bigcup_{j'}I_{C_{j'}}},K_{S\setminus\bigcup_{j'}I_{C_{j'}}},Y_W)\\
        \leq &H(Y_{C\setminus W}|M_{S\setminus\bigcup_{j'}I_{C_{j'}}},K_{S\setminus\bigcup_{j'}I_{C_{j'}}},Y_W)\\
        \leq&H(Y_{C\setminus W})=n(|C|-|W|)\log q,
    \end{align}
\end{subequations}
where (\ref{seq:ub1-last}a) follows from that $M_S(\bP\otimes \bI_{\ell})$ is uniformly distributed over $\mathbb{F}_q^{t\ell}$,
which implies that 
\[\frac{\ell}{n}\leq\frac{|C|-|W|}{t}.\]
\end{pf}
\begin{figure}
    \centering
    \includegraphics[width=0.5\linewidth]{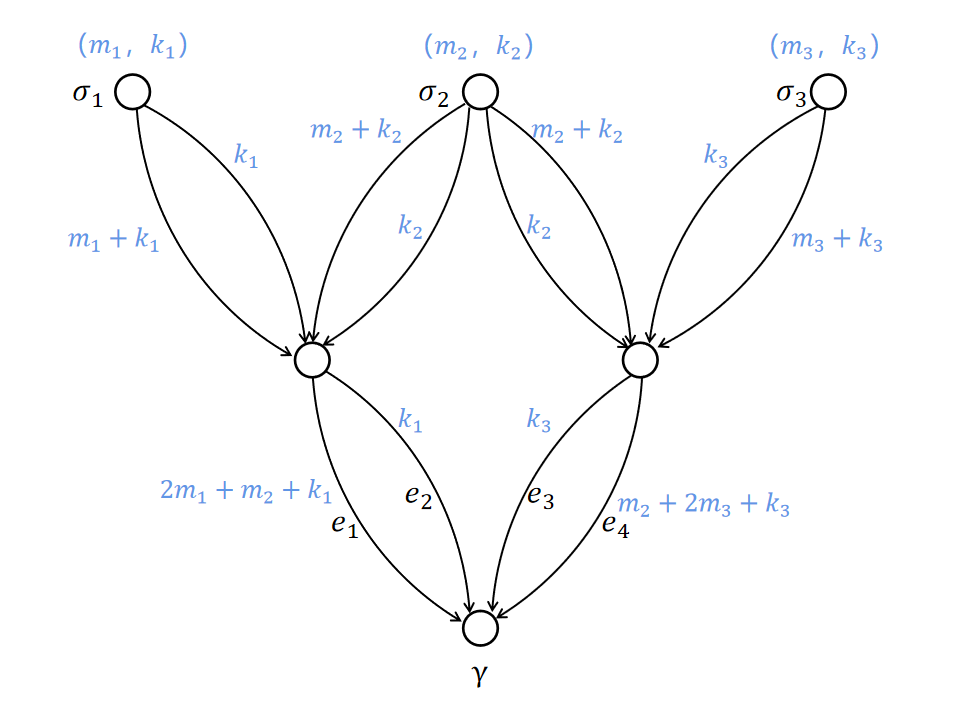}
    \caption{A secure network code for $(\cN,f,g,r)$ with $f=m_1+m_2+m_3,g=m_1+m_3$ (linear functions over $\mathbb{F}_3$) and $r=2$.}
    \label{fig:ML_necessary}
\end{figure}

\begin{figure}
    \centering
    \includegraphics[width=0.5\linewidth]{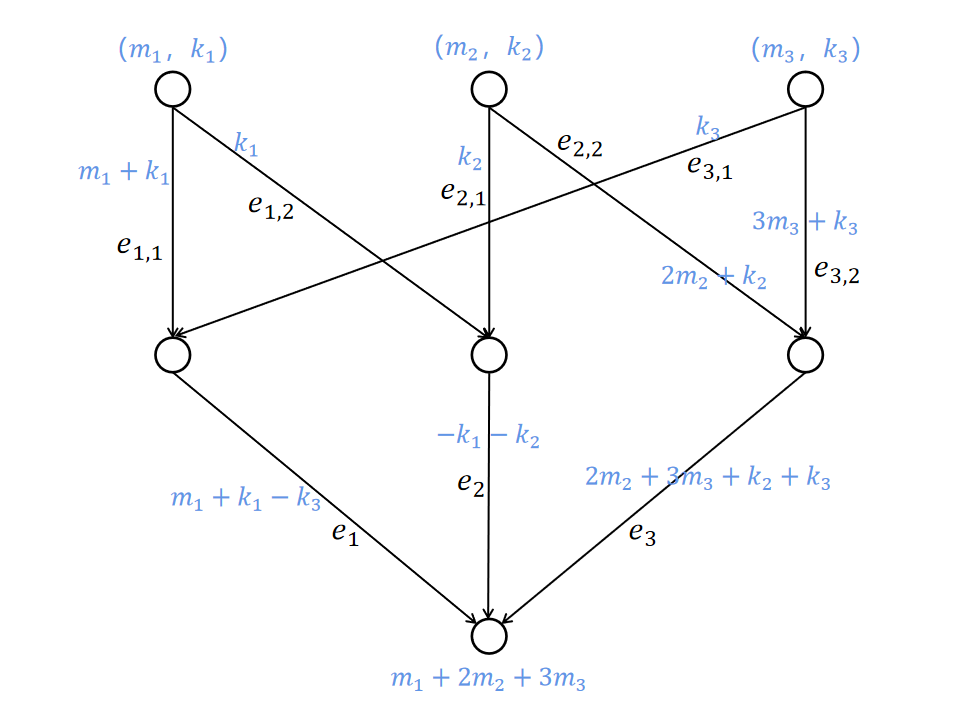}
    \caption{A secure network code for $(\cN,f,g,r)$ with $f=m_1+2m_2+3m_3,g=m_1+m_2+m_3$ (linear functions over $\mathbb{F}_5$) and $r=5$.}
    \label{fig:conjecture}
\end{figure}

\begin{rmk}
    The proof above relies on two key equations (\ref{eq:ubpf_key_1}) and (\ref{seq:ubpf_key_2}). For (\ref{seq:ubpf_key_2}), we aim to show that, after conditioning on certain source messages and keys, the messages $Y_C$ on the cut determine $M_S(\mathbf{P}\otimes \bI_\ell)$. Lemma~\ref{lem:globalcut} states that given $M_{S\setminus I_C}$ and $K_{S\setminus I_C}$, the value $M_S(\mathbf{F}_{I_C}\otimes \bI_\ell)$ can be recovered from $Y_C$. However, this may still be insufficient for recovering $M_S(\mathbf{P}\otimes I_\ell)$. For example, consider the network in Fig.~\ref{fig:ML_necessary} with the cut set $C=\{e_1,e_2,e_3,e_4\}$, $f=m_1+m_2+m_3$, and $g=m_1+m_3$. 
    Here $\mathbf{P}^T=\mathbf{G}^T=\begin{bmatrix}
        1&0&1
    \end{bmatrix}$, and $I_C=S$.
    From Lemma~\ref{lem:globalcut}, $Y_C$ completely determines $M_S(\mathbf{F}_{I_C}\otimes \bI_\ell)$, but it provides no information about $m_1+m_3$, i.e., $H(M_S(\mathbf{P}\otimes \bI_\ell)\mid Y_C)\neq 0$.
    Therefore, we need to add $M_L,K_L$ to the condition. 
    Furthermore,  to ensure that the same conditioning does NOT allow $Y_W$ to reveal $M_S(\mathbf{P}\otimes I_\ell)$, we require $D_W\subseteq\bigcup_{i\in[p]}I_{C_i}$. 
    
    Another essential observation is that the derivation of (\ref{seq:ubpf_key_2}) relies on bounding the entropy by terms of the form   $H(M_S(\bF_{I_{C_j}}\otimes\bI_\ell)|Y_{C_j},M_{S\setminus\bigcup_{j'}I_{C_{j'}}},K_{S\setminus\bigcup_{j'}I_{C_{j'}}})$. This step motivates the requirement that $C$ admits a weak partition. However, this condition is only sufficient, not necessary.  
    As an illustration, consider the secure network code shown in Fig.~\ref{fig:conjecture}. The cut set $C=\{e_1,e_2,e_3,e_{1,2},e_{2,2},e_{3,1}\}$ has a partition $\mathcal{P}_C=\{C_1,C_2,C_3\}$, where $C_1=\{e_1,e_{1,2}\},C_2=\{e_2,e_{2,2}\}$ and $C_{3}=\{e_3,e_{3,1}\}$. This is not a weak partition, but we can still recover $m_1,m_2,m_3$ ($\{M_S(\bF_{I_{C_j}}\otimes\bI_\ell)\}_{j\in[p]}$) from $Y_C$, that is, $H(\{M_S(\bF_{I_{C_j}}\otimes\bI_\ell)\}_{j\in[p]}|Y_C,M_{S\setminus\bigcup_{j'}I_{C_{j'}}},K_{S\setminus\bigcup_{j'}I_{C_{j'}}})=0$ also holds. Whether (\ref{seq:ubpf_key_2}) holds for arbitrary partitions (each component is a cut) in general networks remains unknown.
\end{rmk}

Note that when the target function $f$ is the algebraic sum, and the security function $g$ is the identity function or algebraic sum, then the upper bound in Theorem~\ref{thm:upperbound1} reduces to the bounds in \cite{2024Guangsourcesecure} and \cite{2025Guangfunctionsecure}. The result is summarized as follows.
\begin{cor}\label{cor:ub1}
    \begin{enumerate}
        \item If the target function $f$ is algebraic sum and the security function $g$ is the identity function, then 
        \[\widehat{\cC}(\cN,f,g,r)\leq\min_{(C,W):D_W\subseteq I_C}(|C|-|W|).\]
        \item If both the target function $f$ and the security function $g$ are algebraic sum, then 
        \[\widehat{\cC}(\cN,f,g,r)\leq\min_{\substack{(C,W):I_C\setminus D_W\neq \emptyset\\\textup{or }I_C=D_W=S}}(|C|-|W|).\]
    \end{enumerate}
\end{cor}
\begin{pf}
    First, suppose $f$ is the algebraic sum and  $g$ is the identity function. For any cut set $C\in\Lambda(\cN)$, consider the trivial weak partition $C=C$. Then $M_S(\bP\otimes \bI_{\ell})=M_S(\bF_{I_C}\otimes\bI_\ell)=\sum_{i\in I_C} M_i$, and the pair $(C,W)$ is valid precisely when $D_W\subseteq I_C$. Applying Theorem~\ref{thm:upperbound1}, we obtain
    \[\widehat{\cC}(\cN,f,g,r)\leq\min_{(C,W):D_W\subseteq I_C}(|C|-|W|).\]

     Next, assume both $f$ and $g$ are algebraic sum. If $(C,W)$ is valid, then we must have $I_C=S$. Consequently, from Theorem~\ref{thm:upperbound1}, we directly get
     \[\widehat{\cC}(\cN,f,g,r)\leq\min_{(C,W):I_C=S}(|C|-|W|).\]
     When $I_C\setminus D_W\neq \emptyset$ and $I_C\neq S$, the edge subset $C\setminus W$ is a cut set that separates $\gamma$ from the source nodes $I_{C}\setminus D_W$. This is because the source nodes in $I_C\setminus D_W$ have no path to the edges in $W$, but all the paths from them to $\gamma$ pass through $C$. Hence, all such paths contain an edge of $C\setminus W$. Using the upper bound for non‑secure network computing capacity in Theorem~\ref{thm:csbnd}, we have
     \[\widehat{\cC}(\cN,f,g,r)\leq|C\setminus W|.\]
     Combining the two cases, we can obtain 
     \[\widehat{\cC}(\cN,f,g,r)\leq\min_{\substack{(C,W):I_C\setminus D_W\neq \emptyset\\\textup{or }I_C=D_W=S}}(|C|-|W|).\]
\end{pf}

When the target function $f$ is the identity function and the security function $g$ is vector linear, we can combine the upper bounds given in Theorem~\ref{thm:upperbound1} and Theorem~\ref{thm:csbnd} to obtain a matching upper bound from the perspective of secure network coding.
\begin{cor}
    If the target function $f$ is the identity function and the security function $g$ is vector linear, then 
        \[\widehat{\cC}(\cN,f,g,r)\leq\min_{C:I_C=S}\left\{\frac{|C|}{s},\frac{|C|-r}{r_g}\right\}.\]
\end{cor}
\begin{pf}
     Consider any cut $C$ with $I_C = S$. From Theorem~\ref{thm:csbnd} we obtain
    \[\widehat{\cC}(\cN,f,g,r)\leq\frac{|C|}{\Rank(F_{I_C})}=\frac{|C|}{s}.\]
    From Theorem~\ref{thm:upperbound1}, we have
    \[\widehat{\cC}(\cN,f,g,r)\leq\frac{|C|-r}{\dim(\langle\bG\rangle)}=\frac{|C|-r}{r_g}.\]
    Combining these two bounds, we obtain \[\widehat{\cC}(\cN,f,g,r)\leq\min_{C:I_C=S}\left\{\frac{|C|}{s},\frac{|C|-r}{r_g}\right\}.\]
    Since this holds for every cut $C$ with $I_C=S$, the inequality in the corollary follows. This bound coincides with the secure network coding bound in \eqref{eq:SNC}.
\end{pf}

We now present the second upper bound on $\widehat{\cC}(\cN,f,g,r)$. The proof adjusts the conditioning used in the entropy term of (\ref{seq:ubpf_key_2}a). 
Before presenting the formal result, we need to introduce some additional notation.
For a cut set $C\in\Lambda(\cN)$, define $\tail(C)=\{\tail(e):e\in C\}$. An edge subset $B$ is called \emph{a cut between $J_C$ and $C$} if it separates $\tail(C)$ from $J_C$, that is, after removing the edges in $B$, no path exists from any node $u\in J_C$ to any node $v\in\tail(C)$. Denote by $\cB(C)$ the collection of all such cuts.

\begin{dfn}
    For a cut set $C\in\Lambda(\cN)$, an edge subset $B\in\cB(C)$ and a wiretapped set $W\in\cW$, we say the triple $(C,B,W)$ is valid if it satisfies the following three conditions:
\begin{enumerate}
    \item There exists a weak partition $\cP_C$ of $C$, $\cP_C=\{C_1,\cdots,C_p\}$ such that $\langle \bF_{I_{C_1}}\rangle\bigoplus\cdots\bigoplus\langle \bF_{I_{C_p}}\rangle$ $\bigcap\langle \bG\rangle\neq\{\Zero\}.$
    \item $D_B=J_C$.
    \item $W\subseteq C\cup B$, and $D_W\subseteq I_{C_1}\cup I_{C_2}\cup\cdots\cup I_{C_p}\cup J_C.$
\end{enumerate}
\end{dfn}

\begin{thm}\label{thm:upperbound2}
    For the secure network function computation problem $(\cN,f,g,r)$, the secure network computing capacity\[\widehat{\cC}(\cN,f,g,r)\leq\min_{(C,B,W)\ \text{valid}}\frac{|C\cup B|-|W|}{t_{C,f,g}},\]
    where $t_{C,f,g}=\dim\left(\langle \bF_{I_{C_1}}\rangle\bigoplus\cdots\langle \bF_{I_{C_p}}\rangle\bigcap\langle \bG\rangle\right).$
\end{thm}
\begin{pf}
    Suppose that there exists a secure $(\ell,n)$ network code. Consider a valid triple $(C,B,W)$ and let $C = C_1 \cup \cdots \cup C_p$ be a weak partition, with $I_C = I_{C_1} \cup \cdots \cup I_{C_p} \cup L$. Define $\langle \bP\rangle\eqdef\langle \bF_{I_{C_1}}\rangle\bigoplus\cdots\bigoplus\langle \bF_{I_{C_p}}\rangle$ $\bigcap\langle \bG\rangle$, and let $t=\dim(\langle P))$. Then $M_S(\bP\otimes \bI_{\ell})$ is a function of $M_S(\bF_{I_{C_1}}\otimes\bI_\ell),\cdots,M_S(\bF_{I_{C_p}}\otimes\bI_\ell)$, which implies that $M_S(\bP\otimes \bI_{\ell})$ is determined by $M_{\cup_{j\in[p]} I_{C_j}}$. Therefore,
    \begin{subequations}
        \begin{align}
            &I(M_S(\bP\otimes \bI_{\ell});M_{S\setminus D_C},M_L,K_{S\setminus D_C},K_L)\\
            \leq&I(M_S(\bP\otimes \bI_{\ell}),M_{\cup_{j\in[p]} I_{C_j}};M_{S\setminus D_C},M_L,K_{S\setminus D_C},K_L)\\
            =& I(M_{\cup_{j\in[p]} I_{C_j}};M_{S\setminus D_C},M_L,K_{S\setminus D_C},K_L)=0.
        \end{align}
    \end{subequations}
    Consider the entropy $H(M_S(\bP\otimes \bI_{\ell}),Y_W)$,
\begin{subequations}\label{seq:ubpf_5}
    \begin{align}
        &H(M_S(\bP\otimes \bI_{\ell}))+H(Y_W)\\=&H(M_S(\bP\otimes \bI_{\ell}),Y_W)\\
        =&H(M_S(\bP\otimes \bI_{\ell}),Y_W|M_{S\setminus D_C},M_L,K_{S\setminus D_C},K_L)\\
        =&H(M_S(\bP\otimes \bI_{\ell})|M_{S\setminus D_C},M_L,K_{S\setminus D_C},K_L,Y_W)+H(Y_W|M_{S\setminus D_C},M_L,K_{S\setminus D_C},K_L)\\
        =&H(M_S(\bP\otimes \bI_{\ell})|M_{S\setminus D_C},M_L,K_{S\setminus D_C},K_L,Y_W)+H(Y_W),
    \end{align}
\end{subequations}
where (\ref{seq:ubpf_5}b) is due to the security constraint, (\ref{seq:ubpf_5}c) is because $D_W\subseteq I_{C_1}\cup I_{C_2}\cdots\cup I_{C_p}\cup J_C=D_C\setminus L$ and both $M_S(\bP\otimes \bI_{\ell}),Y_W$ are independent of $M_{S\setminus D_C},M_L,K_{S\setminus D_C},K_L$.
From (\ref{seq:ubpf_5}), we have
\begin{equation}\label{eq:ubpf_key_3}
    H(M_S(\bP\otimes \bI_{\ell}))=H(M_S(\bP\otimes \bI_{\ell})|M_{S\setminus D_C},M_L,K_{S\setminus D_C},K_L,Y_W).
\end{equation}
For the cut $C_j,j\in[p]$, the messages on $C_j$ depend on the messages and keys generated by $I_{C_j},L$ and $J_C$. Since $B$ is a cut of $J_C$ and $C$, we obtain
\begin{equation}
    Y_{C_j}=\widehat{\eta}_{C_j}(M_S,K_S)=\widehat{\eta}_{C_j}(M_{I_{C_j}},K_{I_{C_j}},M_L,K_L,Y_B),
\end{equation}
where $\widehat{\eta}_{C_j}$ is the global encoding function of the cut set $C_j$.
Let $\widehat{\eta}'_{C_j}(M_{I_{C_j}},K_{I_{C_j}})$ be the symbols in $M_{I_{C_j}}$ and $K_{I_{C_j}}$, which appear in the encoding of $Y_{C_j}$. Consequently, we have 
\begin{equation}\label{eq:ubpf_7}
    H(\widehat{\eta}'_{C_j}(M_{I_{C_j}},K_{I_{C_j}}))=H(Y_{C_j}|M_L,K_L,Y_B),
\end{equation}
\begin{equation}\label{eq:ubpf_6}
    H(Y_{C_j}|M_L,K_L,Y_B,\heta'_{C_j}(M_{I_{C_j}},K_{I_{C_j}}))=0.
\end{equation}
Since $D_B=J_C$, $Y_B$ is a function of $M_{J_C},K_{J_C}$. Hence, we have
\begin{subequations}\label{seq:ubpf_10}
    \begin{align}
        &H(Y_{C_j}|\heta'_{C_j}(M_{I_{C_j}},K_{I_{C_j}}),M_L,K_L,M_{J_C},K_{J_C})\\
        = &H(Y_{C_j}|\heta'_{C_j}(M_{I_{C_j}},K_{I_{C_j}}),M_L,K_L,M_{J_C},K_{J_C},Y_B)\\
        \leq&H(Y_{C_j}|\heta'_{C_j}(M_{I_{C_j}},K_{I_{C_j}}),M_L,K_L,Y_B)\overset{(\ref{eq:ubpf_6})}{=}0.
    \end{align}
\end{subequations}
Moreover, we show that $\heta'_{C_j}(M_{I_{C_j}},K_{I_{C_j}})$ can be determined by $M_L,K_L,Y_B$ and $Y_{C_j}$. Specifically, 
\begin{subequations}\label{seq:up2pf-r1}
    \begin{align}
        H(\heta'_{C_j}(M_{I_{C_j}},K_{I_{C_j}}))\geq &H(\heta'_{C_j}(M_{I_{C_j}},K_{I_{C_j}})|M_L,K_L,Y_B)\\
        =&H(\heta'_{C_j}(M_{I_{C_j}},K_{I_{C_j}}),M_L,K_L,Y_B)-H(M_L,K_L,Y_B)\\
        \overset{(\ref{eq:ubpf_6})}{=}&H(\heta'_{C_j}(M_{I_{C_j}},K_{I_{C_j}}),M_L,K_L,Y_B,Y_{C_j})-H(M_L,K_L,Y_B)\\
        =&H(\heta'_{C_j}(M_{I_{C_j}},K_{I_{C_j}}),M_L,K_L,Y_B,Y_{C_j})-H(M_L,K_L,Y_B,Y_{C_j})\nonumber\\&+H(M_L,K_L,Y_B,Y_{C_j})-H(M_L,K_L,Y_B)\\
        =&H(\heta'_{C_j}(M_{I_{C_j}},K_{I_{C_j}})|M_L,K_L,Y_B,Y_{C_j})-H(Y_{C_j}|M_L,K_L,Y_B)\\
        \overset{(\ref{eq:ubpf_7})}{=}&H(\heta'_{C_j}(M_{I_{C_j}},K_{I_{C_j}})|M_L,K_L,Y_B,Y_{C_j})+H(\heta'_{C_j}(M_{I_{C_j}},K_{I_{C_j}})),
    \end{align}
\end{subequations}
where (\ref{seq:up2pf-r1}b) and (\ref{seq:up2pf-r1}e) follow from the definition of conditional entropy.
Thus, we have
\begin{equation}\label{eq:ubpf_gc'}
    H(\heta'_{C_j}(M_{I_{C_j}},K_{I_{C_j}})|M_L,K_L,Y_B,Y_{C_j})=0.
\end{equation}
From the independence of source messages and random keys, we have
\begin{subequations}\label{seq:ubpf_8}
    \begin{align}
        H(M_{J_c},K_{J_c})=&H(M_{J_c},K_{J_c}|\heta'_{C_j}(M_{I_{C_j}},K_{I_{C_j}}),M_{S\setminus D_C},K_{S\setminus D_C},M_L,K_L)\\
        =&H(M_{J_c},K_{J_c}|M_S(\bF_{I_{C_j}}\otimes\bI_\ell),\heta'_{C_j}(M_{I_{C_j}},K_{I_{C_j}}),M_{S\setminus D_C},K_{S\setminus D_C},M_L,K_L),
    \end{align}
\end{subequations}
where the above equalities hold since $\heta'_{C_j}(M_{I_{C_j}},K_{I_{C_j}})$ and $M_S(\bF_{I_{C_j}}\otimes\bI_\ell)$ are functions of $M_{I_{C_j}}$ and $K_{I_{C_j}}$.
Note that the conditional entropy $H(M_{J_c},K_{J_c},M_S(\bF_{I_{C_j}}\otimes\bI_\ell)|\heta'_{C_j}(M_{I_{C_j}},K_{I_{C_j}}),M_{S\setminus D_C},K_{S\setminus D_C},M_L,K_L)$ can be decomposed in two different ways as follows,
\begin{subequations}
    \begin{align}
        &H(M_{J_c},K_{J_c}|\heta'_{C_j}(M_{I_{C_j}},K_{I_{C_j}}),M_{S\setminus D_C},K_{S\setminus D_C},M_L,K_L)\nonumber\\&+H(M_S(\bF_{I_{C_j}}\otimes\bI_\ell)|\heta'_{C_j}(M_{I_{C_j}},K_{I_{C_j}}),M_{S\setminus D_C},K_{S\setminus D_C},M_L,K_L,M_{J_c},K_{J_c})\\
        =&H(M_{J_c},K_{J_c},M_S(\bF_{I_{C_j}}\otimes\bI_\ell)|\heta'_{C_j}(M_{I_{C_j}},K_{I_{C_j}}),M_{S\setminus D_C},K_{S\setminus D_C},M_L,K_L)\\
        =&H(M_S(\bF_{I_{C_j}}\otimes\bI_\ell)|\heta'_{C_j}(M_{I_{C_j}},K_{I_{C_j}}),M_{S\setminus D_C},K_{S\setminus D_C},M_L,K_L)\nonumber\\&+H(M_{J_c},K_{J_c}|M_S(\bF_{I_{C_j}}\otimes\bI_\ell),\heta'_{C_j}(M_{I_{C_j}},K_{I_{C_j}}),M_{S\setminus D_C},K_{S\setminus D_C},M_L,K_L).
    \end{align}
\end{subequations}
Combining with (\ref{seq:ubpf_8}), we have
\begin{align}\label{eq:ubpf_9}
    &H(M_S(\bF_{I_{C_j}}\otimes\bI_\ell)|\heta'_{C_j}(M_{I_{C_j}},K_{I_{C_j}}),M_{S\setminus D_C},K_{S\setminus D_C},M_L,K_L,M_{J_c},K_{J_c})\nonumber\\=&H(M_S(\bF_{I_{C_j}}\otimes\bI_\ell)|\heta'_{C_j}(M_{I_{C_j}},K_{I_{C_j}}),M_{S\setminus D_C},K_{S\setminus D_C},M_L,K_L).
\end{align}
Furthermore, for each $j\in[p]$, 
\begin{subequations}\label{seq:ubpf_12}
    \begin{align}
        &H(M_S(\bF_{I_{C_j}}\otimes\bI_\ell)|Y_{C_j},Y_B,M_{S\setminus D_C},M_L,K_{S\setminus D_C},K_L)\\
        \overset{(\ref{eq:ubpf_gc'})}{=}&H(M_S(\bF_{I_{C_j}}\otimes\bI_\ell)|Y_{C_j},Y_B,M_{S\setminus D_C},M_L,K_{S\setminus D_C},K_L,\heta'_{C_j}(M_{I_{C_j}},K_{I_{C_j}}))\\
        \leq&H(M_S(\bF_{I_{C_j}}\otimes\bI_\ell)|M_{S\setminus D_C},M_L,K_{S\setminus D_C},K_L,\heta'_{C_j}(M_{I_{C_j}},K_{I_{C_j}}))\\
        \overset{(\ref{eq:ubpf_9})}{=}&H(M_S(\bF_{I_{C_j}}\otimes\bI_\ell)|M_{S\setminus D_C},M_L,K_{S\setminus D_C},K_L,\heta'_{C_j}(M_{I_{C_j}},K_{I_{C_j}}),M_{J_c},K_{J_c})\\
        \overset{(\ref{seq:ubpf_10})}{=}&H(M_S(\bF_{I_{C_j}}\otimes\bI_\ell)|M_{S\setminus D_C},M_L,K_{S\setminus D_C},K_L,\heta'_{C_j}(M_{I_{C_j}},K_{I_{C_j}}),M_{J_c},K_{J_c},Y_{C_j})\\
        \leq&H(M_S(\bF_{I_{C_j}}\otimes\bI_\ell)|M_{S\setminus D_C},M_L,K_{S\setminus D_C},K_L,M_{J_c},K_{J_c},Y_{C_j})\\
        =&H(M_S(\bF_{I_{C_j}}\otimes\bI_\ell)|M_{S\setminus \bigcup_{j'}I_{C_{j'}}},K_{S\setminus \bigcup_{j'}I_{C_{j'}}},Y_{C_j})\\
        =&H(M_S(\bF_{I_{C_j}}\otimes\bI_\ell)|M_{S\setminus I_{C_j}},K_{S\setminus I_{C_j}},Y_{C_j})\nonumber\\&+I(M_S(\bF_{I_{C_j}}\otimes\bI_\ell);M_{\bigcup_{j'\neq j}I_{C_{j'}}},K_{\bigcup_{j'\neq j}I_{C_{j'}}}|M_{S\setminus\bigcup_{j'}I_{C_{j'}}},K_{S\setminus\bigcup_{j'}I_{C_{j'}}},Y_{C_j})\\
        \overset{(\ref{seq:ubpf_3})}{=}&H(M_S(\bF_{I_{C_j}}\otimes\bI_\ell)|M_{S\setminus I_{C_j}},K_{S\setminus I_{C_j}},Y_{C_j})=0,
    \end{align}
\end{subequations}
where (\ref{seq:ubpf_12}g) is due to $(S\setminus D_C)\cup L\cup J_C=S\setminus \bigcup_{j'}I_{C_{j'}}$.
Consequently, for each $j\in[p]$ we have
\begin{align}\label{eq:ubpf_11}
    &H(M_S(\bF_{I_{C_j}}\otimes\bI_\ell)|Y_{C},Y_B,M_{S\setminus D_C},M_L,K_{S\setminus D_C},K_L)\\\leq& H(M_S(\bF_{I_{C_j}}\otimes\bI_\ell)|Y_{C_j},Y_B,M_{S\setminus D_C},M_L,K_{S\setminus D_C},K_L)=0.
\end{align}
Therefore, we have
\begin{subequations}\label{seq:ubpf_key_4}
    \begin{align}
        &H(M_S(\bP\otimes \bI_{\ell})|Y_{C},Y_B,M_{S\setminus D_C},M_L,K_{S\setminus D_C},K_L)\\
        \leq&H(M_S(\bP\otimes \bI_{\ell}),\{M_S(\bF_{I_{C_j}}\otimes\bI_\ell)\}_{j\in[p]}|Y_{C},Y_B,M_{S\setminus D_C},M_L,K_{S\setminus D_C},K_L)\\
        = &H(\{M_S(\bF_{I_{C_j}}\otimes\bI_\ell)\}_{j\in[p]}|Y_{C},Y_B,M_{S\setminus D_C},M_L,K_{S\setminus D_C},K_L)\\
        \leq&\sum_{j\in[p]}H(M_S(\bF_{I_{C_j}}\otimes\bI_\ell)|Y_{C},Y_B,M_{S\setminus D_C},M_L,K_{S\setminus D_C},K_L)\overset{(\ref{eq:ubpf_11})}{=}0,
    \end{align}
\end{subequations}
where (\ref{seq:ubpf_key_4}c) is because $M_S(\bP\otimes \bI_{\ell})$ is determined by $F_{I_{C_1}}M_S,\cdots,F_{I_{C_p}}M_S$.
Combining (\ref{eq:ubpf_key_3}) and (\ref{seq:ubpf_key_4}), we have
\begin{subequations}
    \begin{align}
        \ell t\log q=&H(M_S(\bP\otimes \bI_{\ell}))\overset{(\ref{eq:ubpf_key_3})}{=}H(M_S(\bP\otimes \bI_{\ell})|M_{S\setminus D_C},M_L,K_{S\setminus D_C},K_L,Y_W)\\
        \overset{(\ref{seq:ubpf_key_4})}{=}&H(M_S(\bP\otimes \bI_{\ell})|M_{S\setminus D_C},M_L,K_{S\setminus D_C},K_L,Y_W)\\&-H(M_S(\bP\otimes \bI_{\ell})|Y_{C},Y_B,M_{S\setminus D_C},M_L,K_{S\setminus D_C},K_L)\\
        =&I(M_S(\bP\otimes \bI_{\ell});Y_{C\cup B\setminus W}|M_{S\setminus D_C},M_L,K_{S\setminus D_C},K_L,Y_W)\\
        \leq &H(Y_{C\cup B\setminus W}|M_{S\setminus D_C},M_L,K_{S\setminus D_C},K_L,Y_W)\\
        \leq&H(Y_{C\cup B\setminus W})=n(|C\cup B|-|W|)\log q,
    \end{align}
\end{subequations}
which implies that 
\[\frac{\ell}{n}\leq\frac{|C\cup B|-|W|}{t}.\]
\end{pf}

Note that the main difference between Theorem~\ref{thm:upperbound1} and Theorem~\ref{thm:upperbound2} lies in the conditioning set of the entropy term (\ref{seq:ubpf_key_2}a), in detail, $S\setminus I_C$ is replaced by $S\setminus D_C$. To ensure the inequality still holds, we add an edge subset $B$, which contains all the symbols in $Y_C$ from $J_C$. Interestingly, this adjustment can provide a tighter bound in certain cases, as shown in the following example.
\begin{figure}
  \centering
  \includegraphics[width=0.4\textwidth]{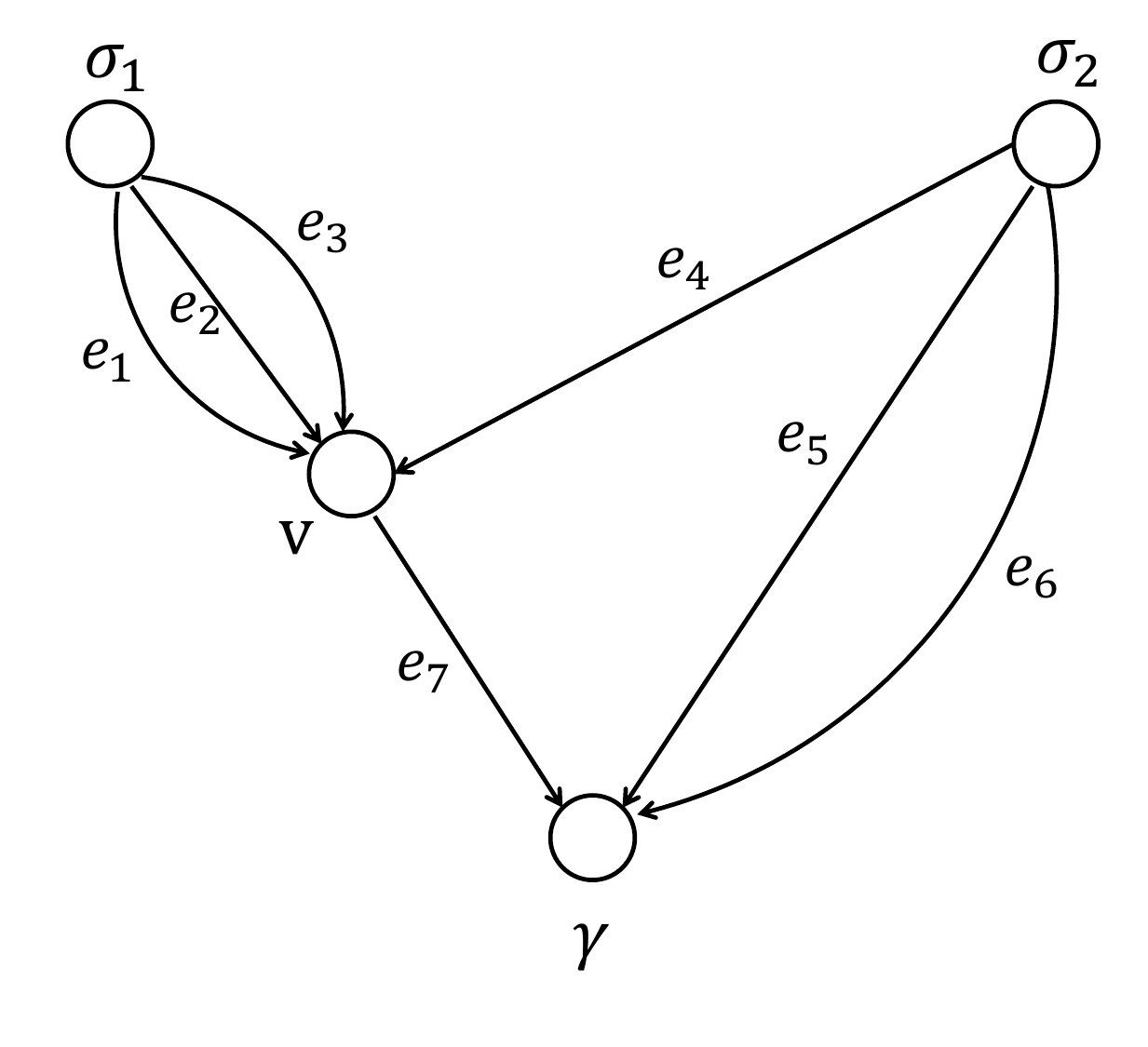}
  \caption{The network $\mathcal{N}$ has $2$ source nodes with binary source messages and the sink node desires the algebraic sum of the two source nodes. }\label{Fig:motivated_example}
\end{figure}
\begin{example}\label{example:motivated_example}
    Consider the network in Fig.~\ref{Fig:motivated_example} with two source nodes $\sigma_1$, $\sigma_2$ and a sink node $\gamma$. The target function $f$ is the algebraic sum, and the security function $g$ is the identity function. When the security level is $r=2$, Theorem~\ref{thm:ub-source-secure} gives the upper bound
    \[\widehat{\cC}(\cN,f,g,r)\leq 1.\]    
    However, there does not exist a secure network code with security level $2$ for this network. To see why, we can assume a wiretapper correctly obtains the messages $y_{e_4},y_{e_7}$ transmitted by edge $e_4$ and $e_7$. Suppose a network code exists that allows $\gamma$ to decode the sum without error. Then the message transmitted by each edge $e_i$ can be represented in a global form $\heta_i(m_1,m_2,k_1,k_2)$. Because $\sigma_2$ has no access to $m_1$ or $k_1$, we have $y_{e_4} = \widehat{\eta}_4(m_2,k_2)$. The message $y_{e_7}$ can be represented as a function of $m_1,k_1$ and $y_{e_4}$, i.e., $y_{e_7}=\heta_7(m_1,m_2,k_1,k_2)=\heta_7'(m_1,k_1,y_{e_4})$. Since $\gamma$ can decode the sum correctly and $y_{e_5},y_{e_6}$ carry no information about $k_1$, the same must hold for $y_{e_7}$.  Therefore, $y_{e_7}=\heta_7'(m_1,y_{e_4})$. Consequently, after observing $y_{e_4}$ and $y_{e_7}$, the wiretapper can obtain information about $m_1$, which violates the security requirement.

    Applying Theorem~\ref{thm:upperbound2} yields a tighter bound. Take the cut $C = \{e_7\}$ with the trivial weak partition $\mathcal{P}_C = \{C\}$. Here $I_C = \{\sigma_1\}$ and $J_C = \{\sigma_2\}$.   Choosing $B = \{e_4\}$ and $W = \{e_4, e_7\}$ gives a valid triple $(C,B,W)$.  Consequently, we have
    $\widehat{\cC}(\cN,f,g,r)\leq 0$,
    which matches the fact that no positive secure rate is achievable when $r=2$.
\end{example}

Similar to Corollary~\ref{cor:ub1}, when the target function $f$ is the algebraic sum, and the security function $g$ is the identity function or algebraic sum, the upper bound in Theorem~\ref{thm:upperbound2} can be simplified as follows.
\begin{cor}\label{cor:ub2}
    \begin{enumerate}
        \item If the target function $f$ is algebraic sum and the security function $g$ is the identity function, then 
        \[\widehat{\cC}(\cN,f,g,r)\leq\min_{\substack{(C,B,W):D_B= J_C\\W\subseteq C\cup B}}(|C\cup B|-|W|).\]
        \item If both the target function $f$ and the security function $g$ are algebraic sum, then 
        \[\widehat{\cC}(\cN,f,g,r)\leq\min_{\substack{(C,W):I_C\setminus D_W\neq \emptyset\\\textup{or }I_C=D_W=S}}(|C|-|W|).\]
    \end{enumerate}
\end{cor}
\begin{pf}
    First, suppose $f$ is the algebraic sum and $g$ is the identity function. For any cut set $C\in\Lambda(\cN)$, consider the trivial weak partition $C=C$. Then $M_S(\bP\otimes \bI_{\ell})=M_S(\bF_{I_C}\otimes\bI_\ell)=\sum M_i$, and the triple $(C,B,W)$ is valid exactly when $D_B= J_C$. By Theorem~\ref{thm:upperbound2}, we obtain
   \[\widehat{\cC}(\cN,f,g,r)\leq\min_{\substack{(C,B,W):D_B= J_C\\W\subseteq C\cup B}}(|C\cup B|-|W|).\]
   
   Next, assume both $f$ and $g$ are the algebraic sum. If $(C,B,W)$ is valid, then necessarily $I_C=S$ and $J_C=\emptyset$. Theorem~\ref{thm:upperbound2} then yields
   \[\widehat{\cC}(\cN,f,g,r)\leq\min_{I_C=S}(|C|-|W|).\]
   Moreover, with the same discussion in the proof of Corollary~\ref{cor:ub1}, we can obtain the result.
\end{pf}

\section{Linear Secure Network Codes: Sum Target Function and Vector Linear Security}\label{sec:lowerbound-sum}
In this section, we consider the scenario where the target function $f$ is the algebraic sum, the security function $g$ is an arbitrary linear function, and the network topology is arbitrary.
\subsection{Linear secure network codes}
We begin by defining linear secure network codes for the model $(\cN, f, g, \cW)$, following the framework in \cite{2024Guangsourcesecure,2025Guangfunctionsecure}.

Let $f,g$ be linear functions over a finite field $\mathbb{F}_q$, that is, $f(\bx)=\bx\cdot\mathbf{F}$ and $g(\bx)=\bx\cdot\mathbf{G}$, where $\mathbf{F}\in\mathbb{F}_q^{s\times r_f},\mathbf{G}\in\mathbb{F}_q^{s\times r_g}$ and $\bx\in\mathbb{F}_q^{1\times s}$. A secure network code is called \emph{linear} if the message transmitted by each edge $e$ is a linear combination of the messages received by $\tail(e)$. Specifically, in an $(\ell,n)$ secure linear network code over $\F_q$, the message $\by_e \in \F_q^{1\times n}$ transmitted via edge $e$ has the form
\begin{equation}
  \by_e=\widehat{\eta}_e(\mathbf{m}_S,\bk_S)=
  \begin{cases}
    \sum\limits_{j=1}^{\ell}m_{ij}  \widehat{\ba}_{(i,j),e}+\sum\limits_{j=1}^{z_i}k_{ij}  \widehat{\bb}_{(i,j),e}, & \mbox{if $\tail(e)= \sigma_i$ for some $i$}; \\
    \sum\limits_{d\in \In(\tail(e))} \by_d\widehat{\mathbf{A}}_{d,e}, & \mbox{otherwise},
  \end{cases}
\end{equation}
where  $\widehat{\ba}_{(i,j),e},\widehat{\bb}_{(i,j),e}\in\mathbb{F}_q^{1\times n}, \widehat{\mathbf{A}}_{d,e} \in \F_{q}^{n\times n}$, $\widehat{\ba}_{(i,j),e},\widehat{\bb}_{(i,j),e}$ are all-zero vectors if $e$ is not an outgoing edge of some source node $\sigma_i\in S$,  $\widehat{\mathbf{A}}_{d,e}$ is an all-zero matrix if $e$ is not an outgoing edge of $\head(d)$, and $z_i$ is the length of random key generated by $\sigma_i$. So, each $\by_e$ can be written as a linear combination of the source messages:
\begin{equation*}
  \by_e = ((\mathbf{m_1}~\bk_1)~(\mathbf{m_2}~\bk_2)~\cdots~(\mathbf{m_s}~\bk_s))\cdot\widehat{\mathbf{H}}_e,
\end{equation*}
where $\widehat{\mathbf{H}}_e \in \F_q^{\left(\sum_{i\in[s]}(\ell+z_i)\right) \times n}$ and for notational simplicity, we follow \cite{2024Guangsourcesecure} to use $(\mathbf{m}_S,\bk_S)$ to represent \[((\mathbf{m_1}~\bk_1)~(\mathbf{m_2}~\bk_2)~\cdots~(\mathbf{m_s}~\bk_s)).\]
In spite of an abuse of notation, we call $\widehat{\mathbf{H}}_e$ the \emph{global encoding matrix} of the edge $e$. Furthermore, we write 
\begin{equation}\label{eq:global-encoding-matrix-form}
    \widehat{\mathbf{H}}_e=\begin{bmatrix}
        \widehat{\mathbf{H}}_e^{(\sigma_1)}\\\widehat{\mathbf{H}}_e^{(\sigma_2)}\\\vdots\\\widehat{\mathbf{H}}_e^{(\sigma_s)}
    \end{bmatrix},
\end{equation}
where $\widehat{\mathbf{H}}_e^{(\sigma_i)}\in\mathbb{F}_q^{\ell+z_i}$ for each $i\in[s]$.
Also, for an edge subset $E\subseteq \cE$, we let $\widehat{\mathbf{H}}_E\eqdef\begin{bmatrix}
    \widehat{\mathbf{H}}_e:e\in E
\end{bmatrix}\in\mathbb{F}_q^{\left(\sum_{i\in[s]}(\ell+z_i)\right)\times n|E|}$ and write
\begin{equation}\label{eq:Eglobal-matrix-form}
    \widehat{\mathbf{H}}_E=\begin{bmatrix}
        \widehat{\mathbf{H}}_E^{(\sigma_1)}\\\widehat{\mathbf{H}}_E^{(\sigma_2)}\\\vdots\\\widehat{\mathbf{H}}_E^{(\sigma_s)}
    \end{bmatrix},
\end{equation}
where $\widehat{\mathbf{H}}_E^{(\sigma_i)}\eqdef\begin{bmatrix}
    \widehat{\mathbf{H}}_e^{(\sigma_i)}:e\in E
\end{bmatrix}\in\mathbb{F}_q^{\left(\ell+z_i\right)\times n|E|}$ for each $i\in[s]$. 

Next, for such an $(\ell,n)$ linear secure network code $\hbc$ for the model $(\cN,f,g,r)$, we define a matrix $\mathbf{S}(\hbc)$ as follows:

\begin{equation}\label{eq:S-def}
    \mathbf{S}(\hbc)\eqdef\begin{bmatrix}
        \mathbf{S}^{(\sigma_1)}\\\mathbf{S}^{(\sigma_2)}\\\vdots\\\mathbf{S}^{(\sigma_s)}
    \end{bmatrix}\text{ with } \mathbf{S}^{(\sigma_i)}=\begin{bmatrix}
        \bg_i\otimes\mathbf{I}_\ell\\\Zero_{z_i\times \ell r_g}
    \end{bmatrix}, \textbf{ } \forall i\in[s],
\end{equation} 
where $\mathbf{I}_\ell$ is the $\ell\times \ell$ identity matrix, $\bg_i$ is the $i$-th row of $\mathbf{G}$, and $\otimes$ is the Kronecker product. Evidently, we have  $\mathbf{S}(\hbc)\in\mathbb{F}_q^{\left(\sum_{i\in[s]}(\ell+z_i)\right) \times \ell r_g}$. In the rest of this paper, we write $\mathbf{S}(\hbc)$ as $\mathbf{S}$ for notational simplicity when there is no ambiguity on the code $\hbc$. With this, we can see that
\begin{equation}\label{eq:secure-function-matrixform}
    (\mathbf{m}_S~\bk_S)\cdot \mathbf{S}=\mathbf{m}_S\cdot(\mathbf{G}\otimes \mathbf{I}_\ell)=g(\mathbf{m}_S).
\end{equation}

We now extend Theorem~7 in \cite{2025Guangfunctionsecure} to the case where $g$ is a vector linear function rather than the algebraic sum. We prove the result for the sake of completeness and adapt the notation to our setting. The proof follows exactly the same line as \cite{2025Guangfunctionsecure}.

\begin{thm}\label{thm:secure-condition-space-form}
  Consider the secure model $(\cN, f, g,r )$ with $f$ being the algebraic sum over $\mathbb{F}_q$. Let $\hbc$ be an $(\ell, n)$ linear secure network code, of which the global encoding matrices are $\widehat{\mathbf{H}}_e,~e\in\cE$. Then, the function security condition \eqref{eq:securitycondition} is satisfied for the code $\hbc$ if and only if
  \begin{align}\label{eq:def_sec_condition_space}
  \big\langle \widehat{\mathbf{H}}_W \big\rangle \cap \big\langle \mathbf{S} \big\rangle = \big\{\Zero\big\},~~\forall~W \in \cW_r. 
  \end{align}
\end{thm}
\begin{pf}
    For the ``only if '' part, we prove it by contradiction. Suppose the contrary holds that there exists a wiretap set $W\in\cW_r$ that does not satisfy the condition in \eqref{eq:def_sec_condition_space}, i.e.,
    \begin{equation}\label{eq:def_sec_condition_space_contrary}
        \big\langle \widehat{\mathbf{H}}_W \big\rangle \cap \big\langle \mathbf{S} \big\rangle \neq \big\{\Zero\big\}. 
    \end{equation}
    Then, it suffices to show that $I(g(M_S);Y_W)>0$. From \eqref{eq:def_sec_condition_space_contrary}, there exist two non-zero vectors $\bm{\alpha}\in\mathbb{F}_q^{n|W|},\bm{\beta}\in\mathbb{F}_q^{\ell r_g}$ such that 
    \begin{equation}\label{eq:Hw-a-S-b}
        \widehat{\mathbf{H}}_W\cdot\bm{\alpha}=\mathbf{S}\cdot\bm{\beta}.
    \end{equation}
    Consequently, we have
    \begin{subequations}\label{eq:onlyif}
        \begin{align}
            &I(g(M_S);Y_W)\\
            =&I(g(M_S);(M_S~K_S)\cdot\widehat{\mathbf{H}}_W)\\
            =&H(g(M_S))-H(g(M_S)|(M_S~K_S)\cdot\widehat{\mathbf{H}}_W)\\
            =&H(g(M_S))-H(g(M_S)|(M_S~K_S)\cdot\widehat{\mathbf{H}}_W,(M_S~K_S)\cdot\widehat{\mathbf{H}}_W\cdot\bm{\alpha})\\
            \geq&H(g(M_S))-H(g(M_S)|(M_S~K_S)\cdot\widehat{\mathbf{H}}_W\cdot\bm{\alpha})\\
            \overset{\eqref{eq:Hw-a-S-b}}{=}&H(g(M_S))-H(g(M_S)|(M_S~K_S)\cdot\mathbf{S}\cdot\bm{\beta})\\
            =&I(g(M_S);(M_S~K_S)\cdot\mathbf{S}\cdot\bm{\beta})\\
            =&H((M_S~K_S)\cdot\mathbf{S}\cdot\bm{\beta})-H((M_S~K_S)\cdot\mathbf{S}\cdot\bm{\beta}|g(M_S))\\
            \overset{\eqref{eq:secure-function-matrixform}}{=}&H((M_S~K_S)\cdot\mathbf{S}\cdot\bm{\beta})-H((M_S~K_S)\cdot\mathbf{S}\cdot\bm{\beta}|(M_S~K_S)\cdot\mathbf{S})\\
            =&H((M_S~K_S)\cdot\mathbf{S}\cdot\bm{\beta})>0,
        \end{align}
    \end{subequations}
    where (\ref{eq:onlyif}j) follows from $\bm{\beta}\neq \Zero$. Then, we have proved $I(g(M_S);Y_W)>0$, which contradicts to the security condition \eqref{eq:securitycondition}.

    For the ``if'' part, we need to show \eqref{eq:securitycondition} or equivalently 
    \begin{equation}
        H(g(M_S)|Y_W)=H(g(M_S))
    \end{equation}
    based on the condition in \eqref{eq:def_sec_condition_space}. Specifically, it suffices to show that for every $W\in\cW_r$, 
    \begin{equation}
        \Pr(g(M_S)=\bu|Y_W=\by)=\Pr(g(M_S)=\bu),
    \end{equation}
    for any row vector $\bu\in\mathbb{F}_q^{\ell r_g}$ and $\by\in\mathbb{F}_q^{n|W|}$ with $\Pr(Y_W=\by)>0$.
    Since the source messages $\{M_{i,j}\}_{i\in[s],j\in[\ell]}$ are i.i.d random variables according to the uniform distribution on $\mathbb{F}_q$ and $\bG$ is a full rank matrix with $\Rank(\bG)=r_g$, we have $M_S\mathbf{G}$ are $\ell r_g$ i.i.d. random variables according to the uniform distribution on $\mathbb{F}_q$. Consequently, for every $\bu\in\mathbb{F}_q^{\ell r_g}$
    \begin{equation}
        \Pr(g(M_S)=\bu)=\frac{1}{q^{\ell r_g}}.
    \end{equation}
    Next, we consider
    \begin{subequations}
        \begin{align}
            &\Pr(g(M_S)=\bu|Y_W=\by)\\
            =&\frac{\Pr(g(M_S)=\bu,Y_W=\by)}{\Pr(Y_W=\by)}\\
            =&\frac{\Pr((\mathbf{m}_S~\bk_S)\mathbf{S}=\bu,(\mathbf{m}_S~\bk_S)\widehat{\mathbf{H}}_W=\by)}{\Pr((\mathbf{m}_S~\bk_S)\widehat{\mathbf{H}}_W=\by)}\\
            =&\frac{\Pr((\mathbf{m}_S~\bk_S)[\mathbf{S}~\widehat{\mathbf{H}}_W]=(\bu~\by))}{\Pr((\mathbf{m}_S~\bk_S)\widehat{\mathbf{H}}_W=\by)}\\
            =&\frac{\sum_{(\mathbf{m}_S~\bk_S):(\mathbf{m}_S~\bk_S)[\mathbf{S}~\widehat{\mathbf{H}}_W]=(\bu~\by)}\Pr(M_S=\mathbf{m}_S,K_S=\bk_S)}{\sum_{(\mathbf{m}'_S~\bk'_S):(\mathbf{m}'_S~\bk'_S)\widehat{\mathbf{H}}_W=\by}\Pr(M_S=\mathbf{m}'_S,K_S=\bk'_S)}\\
            =&\frac{\#\{(\mathbf{m}_S~\bk_S):(\mathbf{m}_S~\bk_S)[\mathbf{S}~\widehat{\mathbf{H}}_W]=(\bu~\by)\}}{\#\{(\mathbf{m}'_S~\bk'_S):(\mathbf{m}'_S~\bk'_S)\widehat{\mathbf{H}}_W=\by\}},
        \end{align}
    \end{subequations}
    where $\#\{\cdot\}$ denotes the cardinality of the set and the equality holds because $M_S$ and $K_S$ are independent and uniformly distributed on $\mathbb{F}_q^{s\ell}$ and $\mathbb{F}_q^{\Sigma_{i\in[s]}z_i}$, respectively. For the denominator, we have
    \begin{equation}\label{eq:probability1}
        \#\{(\mathbf{m}'_S~\bk'_S):(\mathbf{m}'_S~\bk'_S)\widehat{\mathbf{H}}_W=\by\}=q^{s\ell+\sum_{i\in[s]}z_i-\Rank(\widehat{\mathbf{H}}_W)}.
    \end{equation}
    And for the numerator, we obtain
    \begin{subequations}\label{seq:probability2}
        \begin{align}
            &\#\{(\mathbf{m}_S~\bk_S):(\mathbf{m}_S~\bk_S)[\mathbf{S}~\widehat{\mathbf{H}}_W]=(\bu~\by)\}\\
            =&q^{s\ell+\sum_{i\in[s]}z_i-\Rank([\mathbf{S}~\widehat{\mathbf{H}}_W])}\\
            =&q^{s\ell+\sum_{i\in[s]}z_i-\Rank(\mathbf{S})-\Rank(\widehat{\mathbf{H}}_W)}\\
            =&q^{s\ell+\sum_{i\in[s]}z_i-\ell r_g-\Rank(\widehat{\mathbf{H}}_W)},
        \end{align}
    \end{subequations}
    where (\ref{seq:probability2}c) follows from the condition $\big\langle \widehat{\mathbf{H}}_W \big\rangle \cap \big\langle \mathbf{S} \big\rangle = \big\{\Zero\big\}$, and $(\ref{seq:probability2}d)$ is due to the definition of $\mathbf{S}$.
    Combining \eqref{eq:probability1} and \eqref{seq:probability2}, we obtain
    \[\Pr(g(M_S)=\bu|Y_W=\by)=\frac{1}{q^{\ell r_g}}=\Pr(g(M_S)=\bu).\]
    Thus, the theorem is proved.
    
\end{pf}

\subsection{Secure network coding for sum function}
Let the security function be $g(M_S)= M_S\cdot \bG$, where $\bG\in\mathbb{F}_q^{r_g\times s}$.
In \cite{2024Guangsourcesecure,2025Guangfunctionsecure}, a linear secure network code was constructed by transforming a non-secure code into a secure one. We extend this idea to the case where $g$ is a vector linear function.

First, consider the model of computing the algebraic sum $f$ over the network $\cN$ without any security constraint. As illustrated in \cite{2024Guangsourcesecure}, for this model there exists an $(R,1)$ linear network code $\mathbf{C}$ that allows the sink to compute $f$ with zero error. Denote the global encoding vectors in $\mathbf{C}$ as $\{\mathbf{h}_e:e\in\cE\}$. Similar to \eqref{eq:global-encoding-matrix-form}, we write
 \begin{equation}
    \mathbf{h}_e=\begin{bmatrix}
        \mathbf{h}_e^{(\sigma_1)}\\\mathbf{h}_e^{(\sigma_2)}\\\vdots\\\mathbf{h}_e^{(\sigma_s)}
    \end{bmatrix},
\end{equation}
where $\mathbf{h}_e\in\mathbb{F}_q^{Rs}$ and $\mathbf{h}_e^{(\sigma_i)}\in\mathbb{F}_q^{R}$.
Let $\mathbf{B}$ be an $R\times R$ matrix over $\mathbb{F}_q$ and $\widehat{\mathbf{B}}$ be an $Rs\times Rs$ matrix with the form of
\[\widehat{\mathbf{B}}=\begin{bmatrix}
    \mathbf{B}&\Zero_{R\times R}&\cdots&\Zero_{R\times R}\\
    \Zero_{R\times R}&\mathbf{B}&\cdots&\Zero_{R\times R}\\
    \vdots&\vdots&\ddots&\vdots\\
    \Zero_{R\times R}&\Zero_{R\times R}&\cdots&\mathbf{B}
\end{bmatrix}.\]

We select the matrix $\mathbf{B}$ such that the following two conditions are satisfied:
\begin{enumerate}
    \item $\mathbf{B}$ is invertible;
    \item For every $W\in\cW_r$, 
    \begin{equation}\label{eq:B-condition}
    \langle\widehat{\mathbf{B}}^{-1}\mathbf{S}\rangle\cap\langle\mathbf{H}_W\rangle=\{\Zero\},
    \end{equation}
    where $\mathbf{H}_W=\begin{bmatrix}
        \mathbf{h}_e:e\in W
    \end{bmatrix}$.
\end{enumerate}
Then, we can construct an $(R-r,1)$ secure network code $\hbc$ for the model $(\cN,f,g,r)$ as follows. 
Each source node $\sigma_i$ generates $R-r$ i.i.d. random variables $M_i=(M_{i,1},M_{i,2},\cdots,M_{i,R-r})$ as source messages and $r$ i.i.d. random variables $K_i=(K_{i,1},K_{i,2},\cdots,K_{i,r})$ as random keys. Let $\mathbf{m}_i,\bk_i$ be the realizations of $M_i$ and $K_i$.
Let the global encoding vector for the edge $e$ in $\hbc$ be $\widehat{\mathbf{h}}_e=\widehat{\mathbf{B}}\cdot\mathbf{h}_e$. We denote the corresponding secure network code $\hbc\eqdef \widehat{\mathbf{B}}\cdot\mathbf{C}$. 

In the following, we will prove that if the matrix $\mathbf{B}$ satisfies the two conditions, then the constructed network code $\hbc$ satisfies the decodability and the security condition.

\textbf{Verification of decodability}: At the sink node $\gamma$, the messages $y_e,e\in\In(\gamma)$ are received. Since $y_e=(\mathbf{m}_S~\bk_S)\cdot \widehat{\mathbf{h}}_e$, we have
\begin{align*}
    \mathbf{y}_{\In(\gamma)}\eqdef&(y_e:e\in\In(\gamma))=(\mathbf{m}_S~\bk_S)\cdot(\widehat{\mathbf{h}}_e:e\in\In(\gamma))\\
    =&(\mathbf{m}_S~\bk_S)\cdot\widehat{\mathbf{B}}\cdot(\mathbf{h}_e:e\in\In(\gamma))\\
    =&((\mathbf{m}_1~\bk_1)\mathbf{B}~(\mathbf{m}_2~\bk_2)\mathbf{B}~\cdots~(\mathbf{m}_s~\bk_s)\mathbf{B})\cdot \mathbf{H}_{\In(\gamma)}.
\end{align*}
Note that $\mathbf{C}$ is an $(R,1)$ network code for the non-secure network computing problem, that is, from $(\bx_1,\cdots,\bx_s)\cdot\mathbf{H}_{\In(\gamma)}$, we can decode $\sum_{i\in[s]}\bx_i$. Then, from $\by_{\In(\gamma)}$, we can decode 
\[\sum_{i\in[s]}(\mathbf{m}_i~\bk_i)\mathbf{B}=\left(\sum_{i\in[s]}(\mathbf{m}_i~\bk_i)\right)\cdot\mathbf{B}.\]
Since $\mathbf{B}$ is invertible, the sink node can obtain $\sum_{i\in[s]}(\mathbf{m}_i~\bk_i)$ and hence obtain $\sum_{i\in[s]}\mathbf{m}_i$.

\textbf{Verification of security condition}: For every wiretap set $W\in\cW_r$, in the construted secure network code $\hbc$, the global encoding matrix for $\cW$ is $\widehat{\mathbf{B}}\mathbf{H}_W$. From \eqref{eq:B-condition}, we can obtain $\langle\mathbf{S}\rangle\cap\langle\widehat{\mathbf{B}}\mathbf{H}_W\rangle=\{\Zero\}$, because otherwise there exist two non-zero column vectors $\bm{\alpha}\in\mathbb{F}_q^{(R-r)r_g}$ and $\bm{\beta}\in\mathbb{F}_q^{|W|}$ such that
\begin{equation}
    \mathbf{S}\cdot\bm{\alpha}=\widehat{\mathbf{B}}\mathbf{H}_W\cdot\bm{\beta}.
\end{equation}
This implies that $\widehat{\mathbf{B}}^{-1}\mathbf{S}\cdot\bm{\alpha}=\mathbf{H}_W\cdot\bm{\beta}$, which is a contradiction to \eqref{eq:B-condition}. By Theorem~\ref{thm:secure-condition-space-form}, the security condition is satisfied.

To complete the construction, it suffices to construct $\mathbf{B}$ such that the above conditions are satisfied. Recall that the security function $g(\mathbf{m}_S)=\mathbf{m}_S\bG$, where $\bG$ is full rank matrix with rank $r_g$. Without loss of generality, we can assume that the first $r_g$ rows of $\bG$ are linearly independent. Consequently, we can focus on the case where $\bG=[\bI_{r_g};\bG']$.
This is because the wiretapper cannot get any information about $M_S\cdot\bG$ if and only if the wiretapper cannot get any information about $M_S\cdot[\bI_{r_g};\bG']$. For this case, let $\mathbf{B}^{-1}=[\bb_1~\bb_2~\cdots~\bb_R]$, $\bb_i\in\mathbb{F}_q^{R\times 1}$. We select $\bb_1,\bb_2,\cdots,\bb_R$ be $R$ linearly independent vectors in $\mathbb{F}_q^{R}$ such that 
\begin{equation}\label{eq:b-condition}
    \langle\bb_1,\cdots,\bb_{R-r}\rangle\cap\langle\mathbf{H}_W^{(\sigma_i)}\rangle=\{\Zero\},~\forall~1\leq i\leq r_g,~\forall~W\in\cW_r.
\end{equation}
Since $\bb_1,\bb_2,\cdots,\bb_R$ are linearly independent, the invertibility of $\mathbf{B}$ holds. Next, we will show that from \eqref{eq:b-condition}, we can obtain \eqref{eq:B-condition}.
We prove this by contradiction. Suppose that \eqref{eq:B-condition} does not hold, that is, there exists a wiretap set $W$ and two non-zero column vectors $\bm{\alpha}\in\mathbb{F}_q^{(R-r)r_g},\bm{\beta}\in\mathbb{F}_q^{|W|}$ such that
\[\widehat{\mathbf{B}}^{-1}\cdot\mathbf{S}\cdot\bm{\alpha}=\mathbf{H}_W\cdot\bm{\beta}.\]
Recall that $\mathbf{S}$ can be written as \eqref{eq:S-def} and the first $r_g$ rows of $\bG$ form an identity matrix. Then, $\widehat{\mathbf{B}}^{-1}\cdot\mathbf{S}$ can be computed as
\begin{equation}
    \widehat{\mathbf{B}}^{-1}\cdot\mathbf{S}=\begin{bmatrix}
        \bg_1\otimes[\bb_1~\bb_2~\cdots~\bb_{R-r}]\\\bg_2\otimes[\bb_1~\bb_2~\cdots~\bb_{R-r}]\\\vdots\\\bg_s\otimes[\bb_1~\bb_2~\cdots~\bb_{R-r}]
    \end{bmatrix}=\begin{bmatrix}\begin{matrix}
        \mathbf{B}^{-1}_{R-r}&&&\\&\mathbf{B}^{-1}_{R-r}&&\\&&\ddots&\\&&&\mathbf{B}^{-1}_{R-r}
    \end{matrix}\\
        \\\bg_{r_g+1}\otimes\mathbf{B}^{-1}_{R-r}\\\vdots\\\bg_{s}\otimes\mathbf{B}^{-1}_{R-r}
    \end{bmatrix}.\end{equation}
Consequently, $\widehat{\mathbf{B}}^{-1}\cdot\mathbf{S}\cdot\bm{\alpha}$ can be written as
\begin{equation}
    \begin{bmatrix}\mathbf{B}^{-1}_{R-r}\cdot\bm{\alpha}_1\\\mathbf{B}^{-1}_{R-r}\cdot\bm{\alpha}_2\\\vdots\\\mathbf{B}^{-1}_{R-r}\cdot\bm{\alpha}_{r_g}
        \\\bg_{r_g+1}\otimes\mathbf{B}^{-1}_{R-r}\cdot\bm{\alpha}\\\vdots\\\bg_{s}\otimes\mathbf{B}^{-1}_{R-r}\cdot\bm{\alpha}
    \end{bmatrix} \text{ with } \bm{\alpha}=\begin{bmatrix}
        \bm{\alpha}_1\\\bm{\alpha}_2\\\vdots\\\bm{\alpha}_{r_g}
    \end{bmatrix},
\end{equation}
where $\bm{\alpha}_i\in\mathbb{F}_q^{R-r}$ for every $i\in[r_g]$. Since $\bm{\alpha}$ is a non-zero vector, then at least one of $\bm{\alpha}_i,1\leq i\leq r_g$ is non-zero. From the assumption, we obtain that for $1\leq i\leq r_g$, 
\[\mathbf{B}^{-1}_{R-r}\cdot\bm{\alpha}_i=\mathbf{H}_W^{(\sigma_i)}\cdot\bm{\beta}.\]
This implies that for some $i\in[r_g]$,$ \langle\bb_1,\cdots,\bb_{R-r}\rangle\cap\langle\mathbf{H}_W^{(\sigma_i)}\rangle\neq\{\Zero\}$, which is a contradiction.

Finally, we will analyze the required field size such that we can select $R$ linearly independent vectors satisfying \eqref{eq:b-condition}. 

\begin{lem}\label{lem:field-size-lemma}
    If $q>r_g\cdot|\cW_r|$, then there exist $R$ linearly independent vectors in $\mathbb{F}_q^{R}$ satisfying \eqref{eq:b-condition}.
\end{lem}
\begin{pf}
    We define 
    \begin{align*}
        &\mathscr{B}_j\eqdef\langle\bb_1,\bb_2,\cdots,\bb_j\rangle,\\
        &\mathscr{L}_W^{(\sigma_i)}\eqdef\langle\mathbf{h}_e^{(\sigma_i)}\rangle.
    \end{align*}
    We can select $\bb_j$ as follows. For $1\leq j\leq R-r$, we choose $\bb_j\in\mathbb{F}_q^R\setminus\cup_{W\in\cW_r}\cup_{i\in[r_g]}(\mathscr{L}_W^{(\sigma_i)}+\mathscr{B}_{j-1})$. For $R-r+1\leq j\leq R$, we choose $\bb_j\in\mathbb{F}_q^R\setminus\mathscr{B}_{j-1}$.
    Now, it suffices to show that if $q>r_g|\cW_r|$, then $\mathbb{F}_q^R\setminus\cup_{W\in\cW_r}\cup_{i\in[r_g]}(\mathscr{L}_W^{(\sigma_i)}+\mathscr{B}_{j-1})$ and $\mathbb{F}_q^R\setminus\mathscr{B}_{j-1}$ are not empty for $1\leq j\leq R-r$ and $R-r+1\leq j\leq R$, respectively. Consider 
    \begin{subequations}\label{seq:field-size-lemma-1}
        \begin{align}
            |\mathbb{F}_q^R\setminus\cup_{W\in\cW_r}\cup_{i\in[r_g]}(\mathscr{L}_W^{(\sigma_i)}+\mathscr{B}_{j-1})|&\geq q^R-\sum_{W\in\cW_r}\sum_{i\in[r_g]}|\mathscr{L}_W^{(\sigma_i)}+\mathscr{B}_{j-1}|\\
            &\geq q^R-\sum_{W\in\cW_r}\sum_{i\in[r_g]}q^{R-1}\\
            &=q^{R-1}(q-r_g|\cW_r|)>0,
        \end{align}
    \end{subequations}
    where (\ref{seq:field-size-lemma-1}b) is because $\dim(\mathscr{L}_W^{(\sigma_i)})\leq |W|\leq r$ and $\dim(\mathscr{B}_{j-1})\leq R-r-1$. For $|\mathbb{F}_q^R\setminus\mathscr{B}_{j-1}|$, since $j\leq R$, we have $|\mathbb{F}_q^R\setminus\mathscr{B}_{j-1}|=q^R-q^{j-1}>0$. Then, the proof of the lemma is completed. 
\end{pf}

Combining the code construction above and Lemma~\ref{lem:field-size-lemma}, we obtain the following theorem.
\begin{thm}
    Consider the secure network function computation problem $(\cN,f,g,r)$, where $f$ is the algebraic sum over a finite field $\mathbb{F}_q$, $g$ is a vector linear function and the security level $r$ satisfies $0\leq r\leq C_{\min}$. Then there exists a $(C_{\min}-r,1)$ secure linear network code if the field size $q>r_g|\cW_r|$, where $r_g$ is the rank of the coefficient matrix of $g$.
\end{thm}
\begin{pf}
    Based on the results in \cite{2024Guangsourcesecure,2025Guangfunctionsecure}, there exists a $(C_{\min},1)$ linear network code for the network function computation problem $(\cN,f)$ without the security constraint. Then, from the discussion in this subsection and Lemma~\ref{lem:field-size-lemma}, we can obtain a $(C_{\min}-r,1)$ secure network code for $(\cN,f,g,r)$.
\end{pf}

Next, we will give an example to illustrate our code construction. 
\begin{figure}
    \centering
    \begin{minipage}{0.49\linewidth}
        \centering
        \includegraphics[width=1\linewidth]{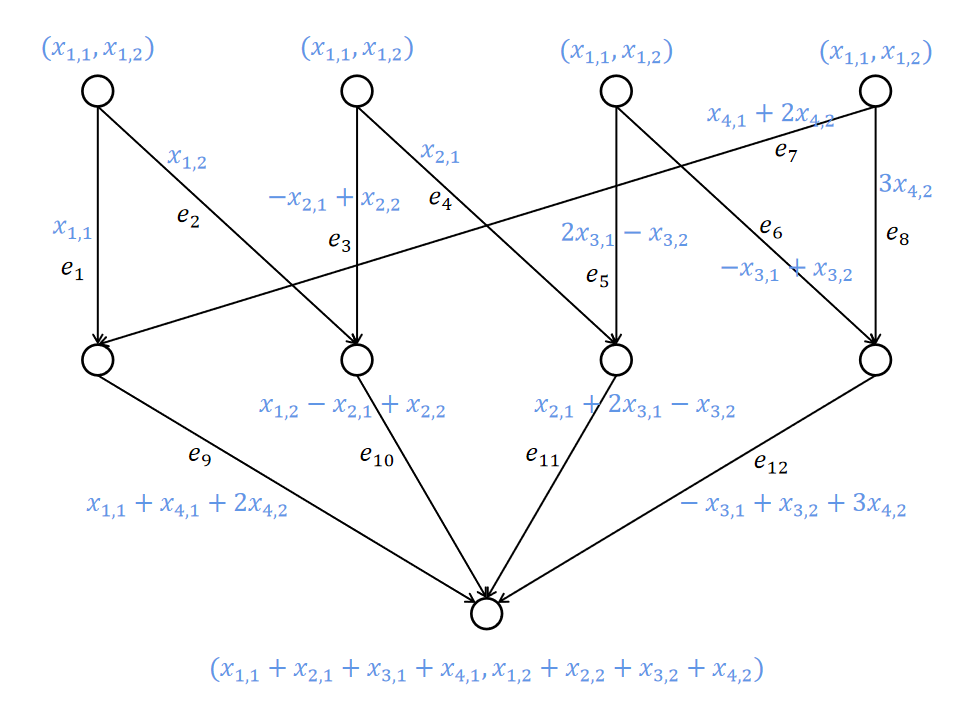}
        \caption{A $3$-layer network $\cN$ and a $(2,1)$ non-secure network code.}
        \label{fig:scheme1_1}
    \end{minipage}
    \begin{minipage}{0.49\linewidth}
        \centering
        \includegraphics[width=1\linewidth]{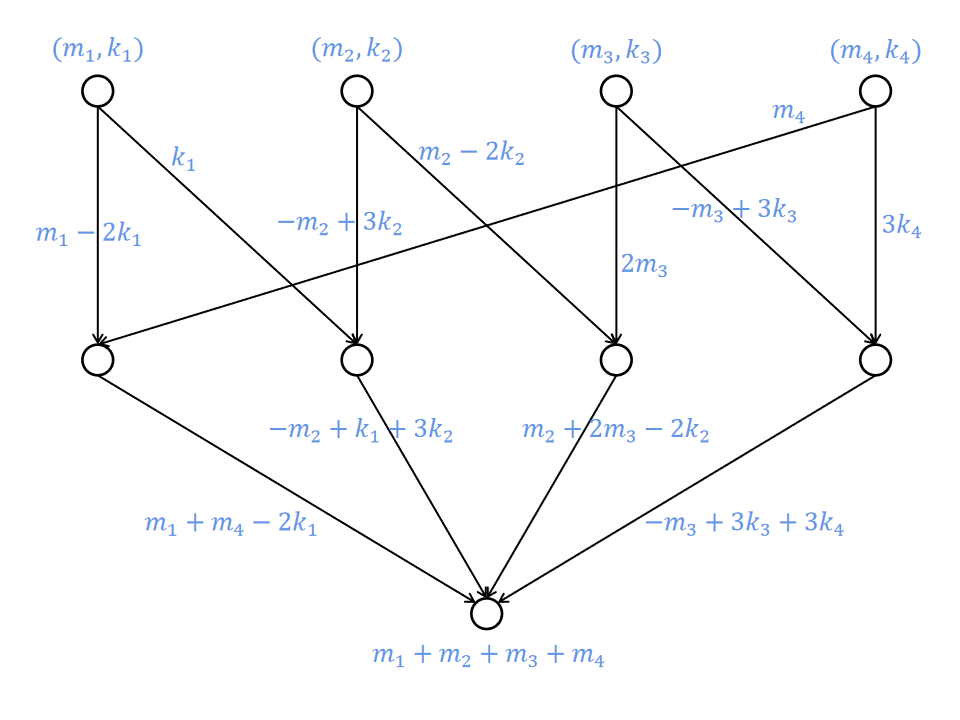}
        \caption{The secure network code in Example~\ref{eg:scheme1}.}
        \label{fig:scheme1_2}
    \end{minipage}
\end{figure}

\begin{example}\label{eg:scheme1}
    Consider the secure network function problem $(\cN,f,g,r)$. The network $\cN$ is a $3$-layer network as depicted in Fig.~\ref{fig:scheme1_1}. The target function $f$ is the algebraic sum over $\mathbb{F}_5$ and the security function $g=\bx\cdot \bG$, where 
    \[\bG^T=\begin{bmatrix}
        1&0&1&1\\0&1&1&0
    \end{bmatrix}.\] That is, the sink node desires $\bx_1+\bx_2+\bx_3+\bx_4$ and the wiretapper wants $\bx_1+\bx_3+\bx_4$ and $\bx_2+\bx_3$. Let the security level $r=1<C_{\min}=2$. 
    
    In the following, we will construct a secure network code with rate $C_{\min}-r=1$. First, we have a $(2,1)$ linear non-secure network code $\mathbf{C}$ as shown in Fig.~\ref{fig:scheme1_1}. In $\mathbf{C}$, the global encoding vectors are
    \begin{align*}
        &\mathbf{h}_1=\text{\footnotesize$\begin{bmatrix}
            1\\0\\0\\0\\0\\0\\0\\0
            \end{bmatrix}$},~\mathbf{h}_2=\text{\footnotesize$\begin{bmatrix}
            0\\1\\0\\0\\0\\0\\0\\0
            \end{bmatrix}$},~\mathbf{h}_3=\text{\footnotesize$\begin{bmatrix}
            0\\0\\-1\\1\\0\\0\\0\\0
            \end{bmatrix}$},~\mathbf{h}_4=\text{\footnotesize$\begin{bmatrix}
            0\\0\\1\\0\\0\\0\\0\\0
            \end{bmatrix}$},~\mathbf{h}_5=\text{\footnotesize$\begin{bmatrix}
            0\\0\\0\\0\\2\\-1\\0\\0
            \end{bmatrix}$},~\mathbf{h}_6=\text{\footnotesize$\begin{bmatrix}
            0\\0\\0\\0\\-1\\1\\0\\0
            \end{bmatrix}$},\\&\mathbf{h}_7=\text{\footnotesize$\begin{bmatrix}
            0\\0\\0\\0\\0\\0\\1\\2
            \end{bmatrix}$},~\mathbf{h}_8=\text{\footnotesize$\begin{bmatrix}
            0\\0\\0\\0\\0\\0\\0\\3
            \end{bmatrix}$},~\mathbf{h}_9=\text{\footnotesize$\begin{bmatrix}
            1\\0\\0\\0\\0\\0\\1\\2
            \end{bmatrix}$},~\mathbf{h}_{10}=\text{\footnotesize$\begin{bmatrix}
            0\\1\\-1\\1\\0\\0\\0\\0
            \end{bmatrix}$},~\mathbf{h}_{11}=\text{\footnotesize$\begin{bmatrix}
            0\\0\\1\\0\\2\\-1\\0\\0
            \end{bmatrix}$},~\mathbf{h}_{12}=\text{\footnotesize$\begin{bmatrix}
            0\\0\\0\\0\\-1\\1\\0\\3
            \end{bmatrix}$}.          
    \end{align*}
    
    Using the code $\mathbf{C}$, the sink node can decode $\bx_1+\bx_2+\bx_3+\bx_4$. 
    Now, we select a $2\times 2$ matrix $\mathbf{B}$ to transform $\mathbf{C}$ to a $(1,1)$ secure network code $\hbc$. Let $\mathbf{B}^{-1}=(\bb_1~\bb_2)$, such that for every edge $e$ and $i=1,2$, 
    \[\langle\bb_1\rangle\cap\langle\mathbf{h}_e^{(\sigma_i)}\rangle=\{\Zero\}.\]
    Since for $i=1,2$, $\mathbf{h}_e^{(\sigma_i)}$ only has three cases, $(1~0)^T,(0~1)^T$ and $(-1~1)^T$,
    to satisfy above condition, we can select 
    \[\bb_1=\text{\footnotesize$\begin{bmatrix}1\\2\end{bmatrix}$},~\bb_2=\text{\footnotesize$\begin{bmatrix}
            0\\1\end{bmatrix}$}.\]
    Then, we have $\mathbf{B}=\begin{bmatrix}
        1&0\\-2&1
    \end{bmatrix}$ and 
    \[\widehat{\mathbf{B}}=\begin{bmatrix}
        \mathbf{B}&\Zero&\Zero&\Zero\\\Zero&\mathbf{B}&\Zero&\Zero\\ \Zero&\Zero&\mathbf{B}&\Zero\\\Zero&\Zero&\Zero&\mathbf{B}
    \end{bmatrix}.\] Consequently, the global encoding vectors in $\hbc$ are computed by $\widehat{\mathbf{h}}_e=\widehat{\mathbf{B}}\mathbf{h}_e$,
    \begin{align*}
        &\widehat{\mathbf{h}}_1=\text{\footnotesize$\begin{bmatrix}
            1\\-2\\0\\0\\0\\0\\0\\0
            \end{bmatrix}$},~\widehat{\mathbf{h}}_2=\text{\footnotesize$\begin{bmatrix}
            0\\1\\0\\0\\0\\0\\0\\0
            \end{bmatrix}$},~\widehat{\mathbf{h}}_3=\text{\footnotesize$\begin{bmatrix}
            0\\0\\-1\\3\\0\\0\\0\\0
            \end{bmatrix}$},~\widehat{\mathbf{h}}_4=\text{\footnotesize$\begin{bmatrix}
            0\\0\\1\\-2\\0\\0\\0\\0
            \end{bmatrix}$},~\widehat{\mathbf{h}}_5=\text{\footnotesize$\begin{bmatrix}
            0\\0\\0\\0\\2\\0\\0\\0
            \end{bmatrix}$},~\widehat{\mathbf{h}}_6=\text{\footnotesize$\begin{bmatrix}
            0\\0\\0\\0\\-1\\3\\0\\0
            \end{bmatrix}$},\\&\widehat{\mathbf{h}}_7=\text{\footnotesize$\begin{bmatrix}
            0\\0\\0\\0\\0\\0\\1\\0
            \end{bmatrix}$},~\widehat{\mathbf{h}}_8=\text{\footnotesize$\begin{bmatrix}
            0\\0\\0\\0\\0\\0\\0\\3
            \end{bmatrix}$},~\widehat{\mathbf{h}}_9=\text{\footnotesize$\begin{bmatrix}
            1\\-2\\0\\0\\0\\0\\1\\0
            \end{bmatrix}$},~\widehat{\mathbf{h}}_{10}=\text{\footnotesize$\begin{bmatrix}
            0\\1\\-1\\3\\0\\0\\0\\0
            \end{bmatrix}$},~\widehat{\mathbf{h}}_{11}=\text{\footnotesize$\begin{bmatrix}
            0\\0\\1\\-2\\2\\0\\0\\0
            \end{bmatrix}$},~\widehat{\mathbf{h}}_{12}=\text{\footnotesize$\begin{bmatrix}
            0\\0\\0\\0\\-1\\3\\0\\3
            \end{bmatrix}$}.          
    \end{align*}
    The sink node can decode $m_1+m_2+m_3+m_4$ from $y_9+2y_{10}+3y_{11}$ (see Fig.~\ref{fig:scheme1_2}). And for every edge $e$, $\langle\widehat{\mathbf{h}}_e\rangle\cap\langle\mathbf{S}\rangle=\{\Zero\}$, where 
    \[\mathbf{S}^T=\begin{bmatrix}
        1&0&0&0&1&0&1&0\\0&0&1&0&1&0&0&0
    \end{bmatrix}.\]
    Therefore, $\hbc$ satisfies the decodability and security conditions. 
    Note that the code $\hbc$ is not source secure, since from $y_5$ or $y_7$, the wiretapper can directly obtain $m_3$ or $m_4$.
\end{example}
Finally, we analyze the gap between the upper bound and the lower bound for secure computing capacity.
\begin{cor}
    When $f$ is the algebraic sum and $g$ is vector linear, the upper bound in Theorem~\ref{thm:upperbound1} satisfies
    \begin{equation}
        C_{\min}-r\leq \min_{(C,W)\ \text{valid}}(|C|-|W|)\leq C_{\min}.
    \end{equation}
    Moreover, 
    \begin{itemize}
        \item if $r\geq\min\{|C|:C\in\Lambda(\cN),D_C=I_C,\text{ and }\exists~\bg\in\langle\bG\rangle,~{\rm{supp}}(\bg)=I_C\}$, then $\widehat{\cC}(\cN,f,g,r)=0$;
        \item if $r=0$, then $\widehat{\cC}(\cN,f,g,r)=C_{\min}$.
    \end{itemize}
\end{cor}
\begin{pf}
    First, we consider
    \begin{subequations}
        \begin{align*}
            \min_{(C,W)\ \text{valid}}(|C|-|W|)&=\min_{C\in\Lambda(\cN)}\min_{\substack{W\in\cW_r:\\(C,W)~\text{valid}}}(|C|-|W|)\\
            &=\min_{C\in\Lambda(\cN)}\left(|C|-\max_{\substack{W\in\cW_r:\\(C,W)~\text{valid}}}|W|\right)\\
            &\geq \min_{C\in\Lambda(\cN)}\left(|C|-r\right)\\
            &=C_{\min}-r.
        \end{align*}
    \end{subequations}
    And the upper bound follows immediately by choosing $W=\emptyset$. 

    Now let $C_0$ be the minimum cut such that $D_{C_0}=I_{C_0}$, and $\exists~\bg\in\langle\bG\rangle,~{\rm{supp}}(\bg)=I_{C_0}$.  For any $W\subseteq C_0$, the pair $(C_0,W)$ is valid.  
    If $r\ge |C_0|$, we may take $W=C_0$, which yields $\widehat{\cC}(\cN,f,g,r)=0$.

    Finally, when $r=0$, $\min_{(C,W)\ \text{valid}}(|C|-|W|)$ reduces to $\min_{C\in\Lambda(\cN)}|C|=C_{\min}$.
\end{pf}

\section{Secure network coding for $(U,V,\alpha)$-trees}\label{sec:lowerbound-tree}
\begin{figure}
    \centering
    \begin{minipage}{0.49\linewidth}
        \centering
        \includegraphics[width=1\linewidth,trim={1cm 1cm 1cm 1cm}]{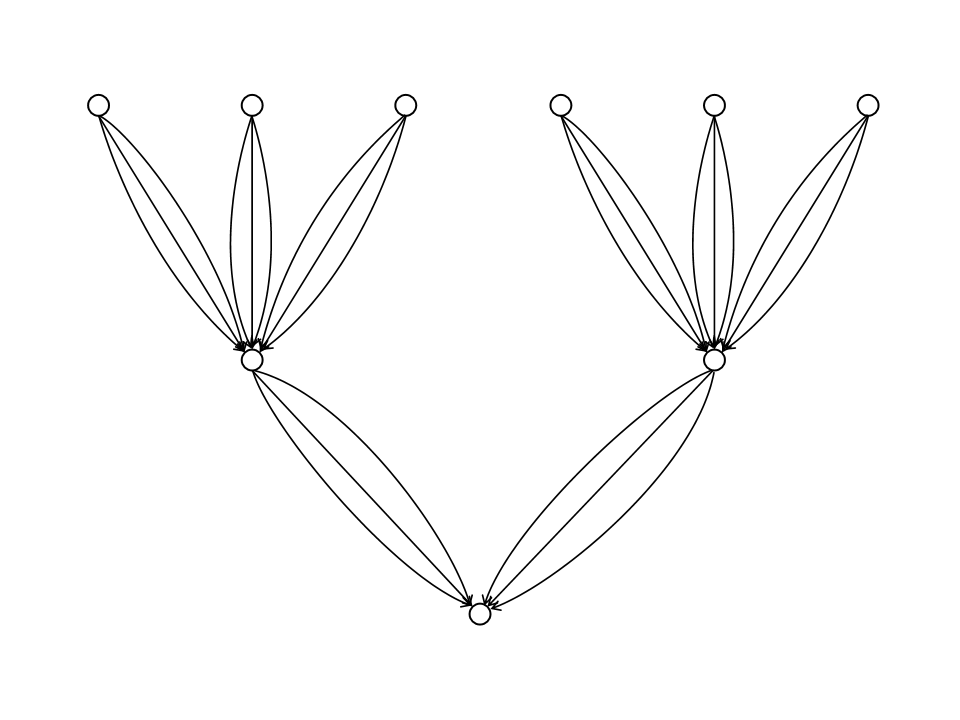}
        \caption{A $(2,3,3)$-tree}
        \label{fig:233tree}
    \end{minipage}
    \begin{minipage}{0.49\linewidth}
        \centering
        \includegraphics[width=1\linewidth,trim={1cm 1cm 1cm 1cm}]{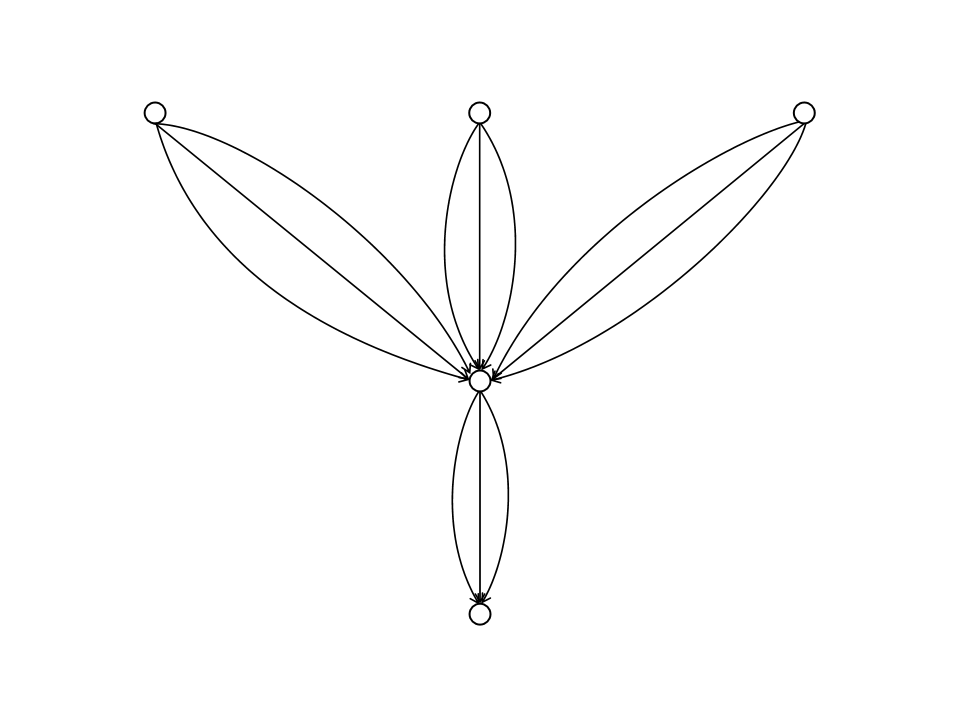}
        \caption{A branch of a $(2,3,3)$-tree.}
        \label{fig:233branch}
    \end{minipage}
\end{figure}
A network $\cN$ is defined as a \emph{$(U,V,\alpha)$-tree} (see Fig.~\ref{fig:233tree}), if it is a three-layer network with $U$ middle nodes and $UV$ source nodes, each middle node connects to exactly $V$ source nodes, and there are $\alpha$ edges between each pair of connected nodes. 
In this subsection, we will consider the secure network function computation model $(\cN,f,g,r)$, where $\cN$ is an $(U,V,\alpha)$-tree, and $f(\bx)=\bx\cdot \mathbf{F},g(\bx)=\bx\cdot \mathbf{G}$ with $\langle\mathbf{F}\rangle\cap\langle\mathbf{G}\rangle\neq\{\Zero\}$. 
This specific network is motivated by two considerations. First, from a theoretical perspective, when the target function is vector linear, the computing capacity can be determined only for multi-edge trees in the network function computation problem. Second, such tree topologies are of practical relevance, for instance, in widely studied hierarchical federated learning systems \cite{2024zhangwansunwangoptimal,xu2025hierarchical,2025LiZhangLvFan}.

The main idea of the code construction is to build secure network codes for subnetworks. Consider a $(U,V,\alpha)$-tree with source nodes $\{\sigma_{i,j}:i\in[U],j\in[V]\}$ and middle nodes $\{v_i:i\in[U]\}$. A subnetwork $\cN_i$ is called a \emph{branch} if its vertex set is $\cV_i=\{v_i\}\cup\{\gamma\}\cup\{\sigma_{i,j}:j\in[V]\}$ and it contains all edges between these vertices from the original network $\cN$ (see Fig.~\ref{fig:233branch}).
Therefore, a $(U,V,\alpha)$-tree can be decomposed into $U$ branches $\{\cN_i:i\in[U]\}.$ The topology ensures that messages on edges within one branch are independent of the source messages generated in any other branch. This independence allows a corresponding decomposition of the target and security functions. Recall that for a subset of source nodes $A\subseteq S$, $\bF_{A}$ is a $UV\times r_f$ matrix which is obtained by replacing the rows corresponding to $S\setminus A$ with $\Zero_{1\times r_f}$. For every $i\in[U]$, let $S_i=\{\sigma_{i,j}:j\in[V]\}$. For notation simplicity, we use $\bF_{i}$ and $\bG_i$ to denote $\bF_{S_i}$ and $\bG_{S_i}$, respectively. Similarly, we let $f_i(\bx)=\bx \bF_i$ and $g_i(\bx)=\bx \bG_i$.

The following key observation reduces the secure coding problem for the entire tree $(\cN,f,g,r)$ to those for its branches $(\cN_i,f_i,g_i,r_i)$.

\begin{prop}\label{prop:tree2branch}
    Consider the secure computation model $(\cN,f,g,r)$ on a $(U,V,\alpha)$-tree $\cN$, and let $(\cN_i,f_i,g_i,r_i)$ denote the corresponding problem on branch $\cN_i$. Suppose an $(\ell,n)$ secure network code exists for each $(\cN_i,f_i,g_i,r_i)$, $i\in[U]$. Then an $(\ell,n)$ secure network code exists for $(\cN,f,g,r)$, when the wiretapper accesses fewer than $r_i$ edges in each branch $\cN_i$.
\end{prop}
\begin{pf}
    For each $i \in [U]$, let $\hbc_i$ be an $(\ell, n)$ linear secure network code for $(\cN_i, f_i, g_i, r_i)$, and denote its global encoding matrix on edge $e \in \cE_i$ by $\widehat{\mathbf{H}}_e(\hbc_i)$. Then we can construct an $(\ell,n)$ network code $\hbc$ over $\cN$ by designing the global encoding matrix for each edge $e\in\cE$ as $\widehat{\mathbf{H}}_e(\hbc)=\widehat{\mathbf{H}}_e(\hbc_i)$ if $e\in\cE_i$. Since each $\hbc_i$ allows the sink $\gamma$ to decode $f_i(M_S)$, $\gamma$ can decode all $\{f_i(M_S): i \in [U]\}$, which leads to the decodability condition.

    We now prove the security condition. Suppose the wiretapper accesses at most $r_i$ edges in each branch $\cN_i$, i.e., $W = \bigcup_{i=1}^U W_i$ with $|W_i| \leq r_i$. First, observe that in $\hbc$, for every $i\in[U]$, $Y_{W_i}$ is determined by $M_{S_i},K_{S_i}$, that is,
    \begin{equation}\label{eq:prop1-6}
        H(Y_{W_i}|M_{S_i},K_{S_i})=0.
    \end{equation}
    Moreover, by the definition of $f_i, g_i$,
    \begin{equation}
        H(f_i(M_S)|M_{S_i},K_{S_i})=0,
    \end{equation}
    \begin{equation}\label{eq:prop1-2}
        H(g_i(M_S)|M_{S_i},K_{S_i})=0.
    \end{equation}
    Due to the security of $\hbc_i$, for each $i\in[U]$
    \begin{equation}\label{eq:propsecurity}
        I(g_i(M_S);Y_{W_i})=0.
    \end{equation}
    From the independence of $\{M_{S_i}:i\in[U]\},\{K_{S_i}:i\in[U]\}$, we obtain
    \begin{subequations}\label{seq:prop1-1}
        \begin{align}
            &I(g_{i}(M_S);M_{S\setminus S_i},K_{S\setminus S_i})\\\leq &I(g_{i}(M_S),M_{S_i},K_{S_i};M_{S\setminus S_i},K_{S\setminus S_i})\\
            \overset{\eqref{eq:prop1-2}}{=}&I(M_{S_i},K_{S_i};M_{S\setminus S_i},K_{S\setminus S_i})=0,
        \end{align}
    \end{subequations}
    where (\ref{seq:prop1-1}b) is due to the property of mutual information. Similarly, we have
    \begin{equation}\label{eq:prop1-3}
        I(Y_{W_i}; M_{S\setminus S_i},K_{S\setminus S_i})=0,
    \end{equation}
    \begin{equation}\label{eq:prop1-4}
        I(Y_{W_i},g_{i}(M_S); M_{S\setminus S_i},K_{S\setminus S_i})=0.
    \end{equation}
    Combining \eqref{seq:prop1-1}, \eqref{eq:prop1-3}, \eqref{eq:prop1-4} with the property of mutual information, we have
    \begin{subequations}
        \begin{align}
            &I(g_i(M_S);Y_{W_i})-I(g_i(M_S);Y_{W_i}|M_{S\setminus S_i},K_{S\setminus S_i})\\
            =&I(g_{i}(M_S);M_{S\setminus S_i},K_{S\setminus S_i})+I(Y_{W_i}; M_{S\setminus S_i},K_{S\setminus S_i})-I(Y_{W_i},g_{i}(M_S); M_{S\setminus S_i},K_{S\setminus S_i})\\
            =&0,
        \end{align}
    \end{subequations}
    which implies 
    \begin{equation}\label{eq:prop1-8}
        I(g_i(M_S);Y_{W_i}|M_{S\setminus S_i},K_{S\setminus S_i})=I(g_i(M_S);Y_{W_i})\overset{\eqref{eq:propsecurity}}{=}0.
    \end{equation}
    Consequently, for $i\in[U]$, 
    \begin{subequations}\label{seq:prop1-5}
        \begin{align}
            H(g_i(M_S))\geq &H(g_i(M_S)|Y_{W},g_1(M_S),\cdots,g_{i-1}(M_S))\\
            \geq&H(g_i(M_S)|Y_{W},g_1(M_S),\cdots,g_{i-1}(M_S),M_{S\setminus S_i},K_{S\setminus S_i})\\
            \overset{\eqref{eq:prop1-6}\eqref{eq:prop1-2}}{=}&H(g_i(M_S)|Y_{W_i},M_{S\setminus S_i},K_{S\setminus S_i})\\
            =&H(g_i(M_S)|M_{S\setminus S_i},K_{S\setminus S_i})-I(g_i(M_S);Y_{W_i}|M_{S\setminus S_i},K_{S\setminus S_i})\\
            \overset{\eqref{eq:prop1-8}}{=}&H(g_i(M_S)|M_{S\setminus S_i},K_{S\setminus S_i})=H(g_i(M_S)),
        \end{align}
    \end{subequations}
    where (\ref{seq:prop1-5}b) uses monotonicity of conditional entropy, and (\ref{seq:prop1-5}e) follows from the security of $\hbc_i$. The same holds for $i=1$ without the conditioning on $\{g_1(M_S),\dots,g_{i-1}(M_S)\}$. Therefore,
    \begin{subequations}\label{seq:prop1-7}
        \begin{align}
            &H(g_1(M_S),\cdots,g_U(M_S)|Y_{W})\\
            =&H(g_1(M_S)|Y_W)+H(g_2(M_S)|Y_W,g_1(M_S))+\cdots+H(g_U(M_S)|Y_{W},g_1(M_S),\cdots,g_{U-1}(M_S))\\
            \overset{\eqref{seq:prop1-5}}{=}&H(g_1(M_S))+H(g_2(M_S))+\cdots+H(g_U(M_S))\\
            =&H(g_1(M_S),\cdots,g_{U}(M_S)),
        \end{align}
    \end{subequations}
    where (\ref{seq:prop1-7}b) follows from the chain rule of the conditional entropy, and (\ref{seq:prop1-7}d) is due to the independence of $M_{S_i},K_{S_i}$ and \eqref{eq:prop1-2}.
    In other words,
    \begin{equation}
        I(g_1(M_S),\cdots,g_U(M_S);Y_W)=0.
    \end{equation}
    Since $g(M_S)=g_1(M_S)+\cdots+g_U(M_S)$, 
    \begin{subequations}
        \begin{align}
            I(g(M_S);Y_W)\leq&I(g(M_S),g_1(M_S),\cdots,g_U(M_S);Y_W)\\
            =&I(g_1(M_S),\cdots,g_U(M_S);Y_W)=0.
        \end{align}
    \end{subequations}
    Thus $\hbc$ satisfies both decodability and security under the stated wiretapping constraint.
\end{pf}
Proposition~\ref{prop:tree2branch} establishes a sufficient condition for securing the global network, namely, securing every branch individually. However, this condition is not necessary. In the following, we extend this result to a more general form.

\begin{prop}\label{prop:extension}
    Consider the secure computation model $(\cN,f,g,r)$ on a $(U,V,\alpha)$-tree $\cN$, and let $(\cN_i,f_i,g_i,r_i)$ denote the corresponding problem on branch $\cN_i$. Let $T\subseteq[U]$ such that 
    \begin{equation}\label{eq:prop-condition}
        (\bigoplus_{i\in[U]\setminus T}\langle\bG_i\rangle)\cap\langle\bG\rangle=\{\Zero\}.
    \end{equation}
    If for every $i\in T$, there exists an $(\ell,n)$ secure network code for $(\cN_i,f_i,g_i,r_i)$, and for $i\in [U]\setminus T$, there exists an $(\ell,n)$ network code enabling $\gamma$ to finish the computation, then we can construct an $(\ell,n)$ secure network code over $\cN$, when the wiretapper accesses fewer than $r_i$ edges in the branch $\cN_i,i\in T$.
\end{prop}
\begin{pf}
    The construction is similar to that in Proposition~\ref{prop:tree2branch}. For each $i \in [U]$, let $\hbc_i$ be an $(\ell, n)$ linear secure network code for $(\cN_i, f_i, g_i, r_i)$. For $i\in [U]\setminus T$, $\hbc_i$ only satisfies the decodability condition. For $i\in T$, $\hbc_i$ satisfies both the decodability condition and secure condition with security level $r_i$. Then we can construct an $(\ell,n)$ network code $\hbc$ over $\cN$ by designing the global encoding matrix for each edge $e\in\cE$ as $\widehat{\mathbf{H}}_e(\hbc)=\widehat{\mathbf{H}}_e(\hbc_i)$ if $e\in\cE_i$.
    The decodability condition naturally follows from the decodability of $\hbc_i$.

    To establish security, we must show that for $W=W_1\cup\cdots\cup W_U$, $|W_i|\leq r_i,i\in T$,
    \[I(g(M_S);Y_W)=0.\]
    As before, for each $i \in [U]$,
    \begin{equation}
        H(g_i(M_S)|M_{S_i},K_{S_i})=0,
    \end{equation}
    \begin{equation}\label{eq:prop2-2}
        H(Y_{W_i}|M_{S_i},K_{S_i})=0.
    \end{equation}
    For $i\in T$, due to the security of $\hbc_i$, we have
    \begin{equation}
        I(g_i(M_S);Y_{W_i})=0.
    \end{equation}
    Moreover, with the same argument as in \eqref{seq:prop1-7}, 
    \begin{equation}\label{eq:prop2-6}
        H(\{g_i(M_S)\}_{i\in T}|\{Y_{W_i}\}_{i\in[T]})=H(\{g_i(M_S)\}_{i\in T}).
    \end{equation}
    For branches outside $T$, we may assume the worst-case scenario where the wiretapper learns the source messages $M_{S_i}$ entirely, i.e.,
    \begin{equation}
        H(g_i(M_S)|Y_{W_i})=0.
    \end{equation}
    Consequently, we have
    \begin{subequations}\label{seq:prop2-1}
        \begin{align}
            H(g(M_S)|Y_W)=&H(g(M_S)|\{Y_{W_i}\}_{i\in[U]})\\
            \geq&H(g(M_S)|\{Y_{W_i}\}_{i\in[U]},\{M_{S_i},K_{S_i}\}_{i\in [U]\setminus T})\\
            \overset{\eqref{eq:prop2-2}}{=}&H(g(M_S)|\{Y_{W_i}\}_{i\in T},\{M_{S_i},K_{S_i}\}_{i\in [U]\setminus T})\\
            =&H(M_S\sum_{i\in[U]}(\bG_i\otimes \bI_\ell)|\{Y_{W_i}\}_{i\in T},\{M_{S_i},K_{S_i}\}_{i\in [U]\setminus T})\\
            =&H(M_S\sum_{i\in T}(\bG_i\otimes \bI_\ell)|\{Y_{W_i}\}_{i\in T},\{M_{S_i},K_{S_i}\}_{i\in [U]\setminus T}),
        \end{align}
    \end{subequations}
    where (\ref{seq:prop2-1}b) is due to the property of entropy, and (\ref{seq:prop2-1}d) is because $g(M_S)=M_S(\bG\otimes\bI_\ell)=M_S((\sum_{i\in[U]}\bG_i)\otimes\bI_\ell)=M_S\sum_{i\in[U]}(\bG_i\otimes \bI_\ell)$.
    
    Under the uniform distribution of $M_S$,
    \begin{equation}
        H(M_S(\bG\otimes\bI_\ell))=r_g\ell\log q,\text{ and }H(M_S((\sum_{i\in T}\bG_i)\otimes\bI_\ell))=\Rank(\sum_{i\in T}\bG_i)\ell\log q. 
    \end{equation}
    In fact, we have $\Rank(\sum_{i\in T}\bG_i)=\Rank(\bG)=r_g$, otherwise there exists a nonzero vector $\bm{\alpha}\in\mathbb{F}_q^{r_g}$ such that $(\sum_{i\in T}\bG_i)\bm{\alpha}=0$. This implies that $\bG\bm{\alpha}=(\sum_{i\in [U]}\bG_i)\bm{\alpha}=(\sum_{i\in [U]\setminus T}\bG_i)\bm{\alpha}$, which is a contradiction to \eqref{eq:prop-condition}. 
    Thus,
    \begin{equation}\label{eq:prop2-3}
        H(M_S((\sum_{i\in T}\bG_i)\otimes\bI_\ell))=r_g\ell\log q=H(M_S(\bG\otimes\bI_\ell)).
    \end{equation}
    Moreover, we have
    \begin{subequations}\label{seq:prop2-4}
        \begin{align}
            &I(M_S((\sum_{i\in T}\bG_i)\otimes\bI_\ell);\{M_{S_i},K_{S_i}\}_{i\in [U]\setminus T}|\{Y_{W_i}\}_{i\in T})\\
            \leq&I(M_S((\sum_{i\in T}\bG_i)\otimes\bI_\ell),\{M_{S_i},K_{S_i}\}_{i\in T};\{M_{S_i},K_{S_i}\}_{i\in [U]\setminus T}|\{Y_{W_i}\}_{i\in T})\\
            =&I(\{M_{S_i},K_{S_i}\}_{i\in T};\{M_{S_i},K_{S_i}\}_{i\in [U]\setminus T}|\{Y_{W_i}\}_{i\in T})=0,
        \end{align}
    \end{subequations}
    where (\ref{seq:prop2-4}b) is due to the property of mutual information, and (\ref{seq:prop2-4}c) is because that $M_S((\sum_{i\in T}\bG_i)\otimes\bI_\ell)$ is determined by $\{M_{S_i},K_{S_i}\}_{i\in T}$.
    Thus, continue (\ref{seq:prop2-1}), we have
    \begin{subequations}\label{seq:prop2-5}
        \begin{align}
            H(g(M_S)|Y_W)\geq &H(M_S((\sum_{i\in T}\bG_i)\otimes\bI_\ell)|\{Y_{W_i}\}_{i\in T},\{M_{S_i},K_{S_i}\}_{i\in [U]\setminus T})\\
            =&H(M_S((\sum_{i\in T}\bG_i)\otimes\bI_\ell)|\{Y_{W_i}\}_{i\in T})\nonumber\\&-I(M_S((\sum_{i\in T}\bG_i)\otimes\bI_\ell);\{M_{S_i},K_{S_i}\}_{i\in [U]\setminus T}|\{Y_{W_i}\}_{i\in T})\\
            \overset{\eqref{seq:prop2-4}}{=}&H(M_S((\sum_{i\in T}\bG_i)\otimes\bI_\ell)|\{Y_{W_i}\}_{i\in T})\\
            \overset{\eqref{eq:prop2-6}}{=}&H(M_S((\sum_{i\in T}\bG_i)\otimes\bI_\ell))\\
            \overset{\eqref{eq:prop2-3}}{=}&H(M_S(\bG\otimes\bI_\ell))=H(g(M_S)),
        \end{align}
    \end{subequations}
    where (\ref{seq:prop2-5}b) is due to the property of conditional entropy.
    Hence, the security of $\hbc$ is proved.
\end{pf}

While Proposition~\ref{prop:tree2branch} requires security on all branches, Proposition~\ref{prop:extension} demonstrates that security on a carefully chosen subset $T$ of branches is sufficient, subject to condition \eqref{eq:prop-condition}. Both results ultimately reduce the problem to constructing secure codes on individual branches. In the following, we therefore concentrate on designing secure network codes for the secure network computation  $(\cN_i,f_i,g_i,r_i)$ over each branch.
\begin{rmk}
    The core idea of Proposition~\ref{prop:extension} is to find a set $T$ of branches that the network codes should be secure. The required condition for $T$ is characterized by \eqref{eq:prop-condition}. Notably, a very similar idea appears in the recent work on vector linear secure aggregation by Hu and Ulukus \cite{hu2026capacity}. In their work, they determined a set of users who need to encrypt their data. While the idea of finding a critical subset is similar, the network model differs from ours. In their aggregation model, each user connects via a single link. In our network model, this single link is replaced by an entire branch, which may contain multiple links and internal network coding. Accordingly, the security function in our work is described by a matrix over each branch, while in their model each user is associated with a vector.
\end{rmk}

\subsection{Linear secure network codes for branches}\label{subsec:branch}
In this subsection, we construct a linear secure network code for $(\cN_i, f_i, g_i, r_i)$. To simplify notation, we omit the subscript $i$ throughout this subsection. Consequently, the network $\cN$ considered here has $V$ source nodes, one middle node, and one sink node. The target function and the security function are $f(\bx)=\bx\bF$ and $g(\bx)=\bx\bG$, where $\bx\in\mathbb{F}_q^{V},\bF\in\mathbb{F}_q^{V\times r_f}$ and $\bG\in\mathbb{F}_q^{V\times r_g}$ with $r_f=\Rank(\bF)$ and $r_g=\Rank(\bG)$. Note that in the general case, $\bF$ and $\bG$ are obtained by truncating the original functions to the branch and may not be full-rank. However, we can focus on their full-rank components as they span the relevant subspaces.

From the upper bounds in Theorem~\ref{thm:csbnd}, Theorem~\ref{thm:upperbound1} and Theorem~\ref{thm:upperbound2}, we obtain the following upper bound for the secure network computation model over the branch $(\cN,f,g,r)$.
\begin{cor}\label{cor:ub-branch}
    For the secure network function computation problem over the branch $(\cN,f,g,r)$, the secure computing capacity 
    \[\widehat{\cC}(\cN,f,g,r)\leq\min\left\{\frac{\alpha}{r_f},\min_{(C,W)\ \text{valid}}\frac{|C|-|W|}{t_{C,f,g}},\min_{(C,B,W)\ \text{valid}}\frac{|C\cup B|-|W|}{t_{C,f,g}}\right\},\]
    where $t_{C,f,g}=\dim\left(\langle \bF_{I_{C_1}}\rangle\bigoplus\cdots\langle \bF_{I_{C_p}}\rangle\bigcap\langle \bG\rangle\right).$
\end{cor}
\begin{pf}
    From Theorem~\ref{thm:csbnd} and the network topology, we obtain
    \[\widehat{\cC}(\cN,f,g,r)\leq \cC(\cN,f)\leq\frac{|\In(\gamma)|}{\Rank(\bF)}=\frac{\alpha}{r_f}.\]
    From Theorem~\ref{thm:upperbound1} and Theorem~\ref{thm:upperbound2}, we can obtain the rest of the result.
\end{pf}

Let $\widehat{\cC}^*\eqdef\min\left\{\frac{\alpha}{r_f},\min_{(C,W)\ \text{valid}}\frac{|C|-|W|}{t_{C,f,g}},\min_{(C,B,W)\ \text{valid}}\frac{|C\cup B|-|W|}{t_{C,f,g}}\right\}$. 
Based on the upper bound, the construction will be presented according to the following two cases: 1) when $\widehat{\cC}^*=\alpha/r_f$ and 2) when $\widehat{\cC}^*=\ell/n<\alpha/r_f$.

\textbf{Case 1:}
We now present an $(\alpha, r_f)$ linear secure network code for a branch under the condition $\widehat{\cC}^* = \alpha/r_f$. This condition implies the inequality
\begin{equation}
    \frac{\alpha}{r_f}\leq\frac{|\In(\gamma)|-r}{t_{\In(\gamma),f,g}}\leq \alpha-r,
\end{equation}
from which it follows that $\alpha + r r_f < \alpha r_f$.

We first present an example to illustrate the main idea of the construction.
\begin{example}
    Consider the network shown in Fig.~\ref{fig:233branch}, which is a branch for $(U,3,3)$-tree. Let the field size $q=7$ and the coefficient matrices of the target security function be 
    \[\bF^T=\begin{bmatrix}
        1&1&2\\0&1&1
    \end{bmatrix},~\bG^T=\begin{bmatrix}
        1&0&1
    \end{bmatrix}.\]
    We assume that the security level $r=1$, then we can obtain an upper bound $\widehat{\cC}(\cN,f,g,r)\leq \alpha/r_f=3/2$. Note that each source node can transmit at most $\alpha r_f=6$ symbols to the intermediate node, and the wiretapper could obtain at most $rr_f=2$ of them. Hence, to guarantee the security and decodability, each source node generates a source message of length $\alpha=3$ and a random key of length $2$. The source node $\sigma_i$ sends $\by_{\Out(\sigma_i)}=(\mathbf{m}_i~\bk_i)\widehat{\mathbf{H}}_{\Out(\sigma_i)}^{(\sigma_i)}$, where $\widehat{\mathbf{H}}_{\Out(\sigma_i)}^{(\sigma_i)}\in\mathbb{F}_7^{5\times 6}$ is selected as
    \[\widehat{\mathbf{H}}_{\Out(\sigma_i)}^{(\sigma_i)}=\begin{bmatrix}
        1&0&0&0&0&0\\0&1&0&0&0&0\\0&0&1&0&0&0\\1&1&1&1&1&1\\1&2&3&4&5&6
    \end{bmatrix}.\]
    Since $\widehat{\mathbf{H}}_{\Out(\sigma_i)}^{(\sigma_i)}$ has full row rank, the intermediate node can decode $\mathbf{m}_i,i\in[3]$. Next, the middle node can transmit at most $6$ symbols to $\gamma$, and the value of target function also contains $6$ symbols. Therefore, to guarantee the decodability, no random keys can be involved in the transmission on the edges $\In(\gamma)$. The intermediate node transmits $\by_{\In(\gamma)}=(\mathbf{m}_S~\bk_S)\widehat{\mathbf{H}}_{\In(\gamma)}=\mathbf{m}_S\widehat{\mathbf{H}}_{\In(\gamma)}^{(M)}$, where $\widehat{\mathbf{H}}_{\In(\gamma)}^{(M)}\in\mathbb{F}_q^{9\times 6}$ is selected as
    \[\widehat{\mathbf{H}}_{\Out(\sigma_i)}^{(\sigma_i)}=\begin{bmatrix}
        1&0&0&0&0&0\\0&1&0&0&0&0\\0&0&1&0&0&0\\1&1&1&1&1&1\\1&2&3&4&5&6
    \end{bmatrix},~~\widehat{\mathbf{H}}_{\In(\gamma)}^{(M)}=\begin{bmatrix}
        1&0&0&0&0&1\\0&0&1&0&0&0\\0&0&0&0&0&1\\1&1&0&0&0&1\\0&1&1&1&0&0\\0&0&1&0&1&1\\2&1&0&0&0&2\\0&1&2&1&0&0\\0&0&1&0&1&2
    \end{bmatrix}.\]
    In fact, $\widehat{\mathbf{H}}_{\In(\gamma)}^{(M)}$ is obtained by applying column transformation to $\bF\otimes\bI_{\alpha}$ such that any $2$ columns are linearly independent with $\bG\otimes\bI_\alpha$. 
    The sink node receives 
    \begin{align*}
        &y_1=m_{11}+m_{21}+m_{31},~~~y_2=m_{21}+m_{31}+m_{22}+m_{32},~~~y_3=m_{12}+m_{22}+2m_{32}+m_{33},\\
        &y_4=m_{22}+m_{32},~~~y_5=m_{23}+m_{33},~~~y_6=m_{11}+m_{21}+2m_{31}+m_{13}+m_{23}+2m_{33}.
    \end{align*}
    Then, the sink node can decode $f(\mathbf{m}_S)$ from $y_1,y_2-y_4,y_3-y_5,y_4,y_6-y_1$ and $y_5$. 
    The security can be checked as follows. If the wiretapper accesses to an edge in $\Out(\sigma_i)$, then the random keys in the symbols can not be eliminated due to the MDS property of the last $2$ rows of $\widehat{\mathbf{H}}_{\Out(\sigma_i)}$. If the wiretapper accesses to an edge in $\In(\gamma)$, due to the property of $\widehat{\mathbf{H}}_{\In(\gamma)}^{(M)}$, it cannot obtain any information about $\bG\otimes\bI_\alpha$.
\end{example} 
The formal construction is stated as follows.

\textbf{Notations}:
Each source node $\sigma_i$, $i \in [V]$ independently generates a source message $M_i$ of length $\alpha$ and a random key $K_i$ of length $r r_f$, both uniformly distributed over $\mathbb{F}_q$. For each edge $e\in\Out(\sigma_i)$, the global encoding matrix of $e$ is denoted as $\widehat{\mathbf{H}}_e\in\mathbb{F}_q^{V(\alpha+rr_f)\times r_f}$. For $\Out(\sigma_i)$, let $\widehat{\mathbf{H}}_{\Out(\sigma_i)}=\begin{bmatrix}
    \widehat{\mathbf{H}}_e:e\in\Out(\sigma_i)
\end{bmatrix}\in\mathbb{F}_q^{V(\alpha+rr_f)\times \alpha r_f}$. Then, $\widehat{\mathbf{H}}_{\Out(\sigma_i)}$ can be written in the form of \eqref{eq:Eglobal-matrix-form}, which is represented by $\{\widehat{\mathbf{H}}_{\Out(\sigma_i)}^{(\sigma_j)}:j\in[V]\}$. By construction, $\widehat{\mathbf{H}}_{\Out(\sigma_i)}^{(\sigma_j)} = \Zero$ for $j \neq i$. For $j=i$, we write $\widehat{\mathbf{H}}_{\Out(\sigma_i)}^{(\sigma_j)}\in\mathbb{F}_q^{(\alpha+rr_f)\times \alpha r_f}$ as
\begin{equation}\label{eq:scheme2-case1-out}
    \widehat{\mathbf{H}}_{\Out(\sigma_i)}^{(\sigma_i)}=\begin{bmatrix}
        \widehat{\mathbf{H}}_{\Out(\sigma_i)}^{(\sigma_j,M)}\\
        \widehat{\mathbf{H}}_{\Out(\sigma_i)}^{(\sigma_j,K)}
    \end{bmatrix},
\end{equation}
where $\widehat{\mathbf{H}}_{\Out(\sigma_i)}^{(\sigma_j,M)}\in\mathbb{F}_q^{\alpha\times \alpha r_f}$ and $\widehat{\mathbf{H}}_{\Out(\sigma_i)}^{(\sigma_j,K)}\in\mathbb{F}_q^{rr_f\times \alpha r_f}$ are the coefficient matrices for $M_i$ and $K_i$, respectively. 

Next, consider the transmission on the edges in $\In(\gamma)$. Since $\gamma$ desires $M_S(\bF\otimes \bI_{\alpha})$, which contains $\alpha r_f$ symbols, and it can receive at most $\alpha r_f$ symbols, the messages transmitted by $\In(\gamma)$ cannot contain any information about $\{K_i:i\in[V]\}$. In other words, the global encoding matrix for $\In(\gamma)$ has the following form,
\begin{equation}\label{eq:scheme2-case1In}
    \widehat{\mathbf{H}}_{\In(\gamma)}=\begin{bmatrix}
        \widehat{\mathbf{H}}_{\In(\gamma)}^{(\sigma_1)}\\\widehat{\mathbf{H}}_{\In(\gamma)}^{(\sigma_2)}\\\vdots\\\widehat{\mathbf{H}}_{\In(\gamma)}^{(\sigma_V)}
    \end{bmatrix}~\text{with } \widehat{\mathbf{H}}_{\In(\gamma)}^{(\sigma_i)}=\begin{bmatrix}
        \widehat{\mathbf{H}}_{\In(\gamma)}^{(\sigma_i,M)}\\
        \widehat{\mathbf{H}}_{\In(\gamma)}^{(\sigma_i,K)}
    \end{bmatrix}=\begin{bmatrix}
        \widehat{\mathbf{H}}_{\In(\gamma)}^{(\sigma_i,M)}\\
        \Zero
    \end{bmatrix},
\end{equation}
where $\widehat{\mathbf{H}}_{\In(\gamma)}\in\mathbb{F}_q^{V(\alpha+rr_f)\times \alpha r_f}$ and $\widehat{\mathbf{H}}_{\In(\gamma)}^{(\sigma_i,M)}\in\mathbb{F}_q^{\alpha\times \alpha r_f}$. Let $\widehat{\mathbf{H}}_{\In(\gamma)}^{(M)}\in\mathbb{F}_q^{\alpha V\times \alpha r_f}$ be the matrix 
\begin{equation}
    \widehat{\mathbf{H}}_{\In(\gamma)}^{(M)}=\begin{bmatrix}
        \widehat{\mathbf{H}}_{\In(\gamma)}^{(\sigma_1,M)}\\\widehat{\mathbf{H}}_{\In(\gamma)}^{(\sigma_2,M)}\\\vdots\\\widehat{\mathbf{H}}_{\In(\gamma)}^{(\sigma_V,M)}
    \end{bmatrix}.
\end{equation}

\textbf{Construction}:
The code is specified by the following conditions on the encoding matrices.
\begin{enumerate}
    \item For each $i \in [V]$, $\widehat{\mathbf{H}}_{\Out(\sigma_i)}^{(\sigma_i)}$ is of full row rank;
    \item For each $i \in [V]$, $\widehat{\mathbf{H}}_{\Out(\sigma_i)}^{(\sigma_i,K)}$ is a maximum distance separable (MDS) matrix; i.e., any set of $rr_f$ of its columns is linearly independent;
    \item There exists an invertible matrix $\mathbf{A} \in \mathbb{F}_q^{\alpha r_f \times \alpha r_f}$ such that $\widehat{\mathbf{H}}_{\In(\gamma)}^{(M)} = (\bF \otimes \bI_{\alpha}) \cdot \mathbf{A}$;
    \item For any submatrix $\mathbf{B}$ formed by $rr_f$ columns of $\widehat{\mathbf{H}}_{\In(\gamma)}$, we have $\langle \mathbf{B} \rangle \cap \langle \mathbf{S} \rangle = \{\Zero\}$, where $\mathbf{S}$ is defined in \eqref{eq:S-def}.
\end{enumerate}



\textbf{Existence}: 
We now prove that matrices satisfying the conditions exist for sufficiently large field size $q$.
Matrices satisfying 1) and 2) can be constructed explicitly. For each $i\in[V]$, we can select
\[\widehat{\mathbf{H}}_{\Out(\sigma_i)}^{(\sigma_i)}=\begin{bmatrix}
        \widehat{\mathbf{H}}_{\Out(\sigma_i)}^{(\sigma_j,M)}\\
        \widehat{\mathbf{H}}_{\Out(\sigma_i)}^{(\sigma_j,K)}
    \end{bmatrix}=\begin{bmatrix}
        \begin{matrix}
            \bI_{\alpha}&\Zero_{\alpha\times\alpha(r_f-1)}
        \end{matrix}\\\mathbf{V}_{rr_f\times\alpha r_f}(\ba)
    \end{bmatrix},\]
where $\ba=(a_1,\cdots,a_{\alpha r_f})$, and $\mathbf{V}_{rr_f\times\alpha r_f}(\ba)$ is a Vandermonde matrix with the form
\[\mathbf{V}_{rr_f\times\alpha r_f}(\ba)=\begin{bmatrix}
    1&1&\cdots&1\\
    a_1&a_2&\cdots&a_{\alpha r_f}\\
    \vdots&\vdots&\ddots&\vdots\\
    a_1^{rr_f-1}&a_2^{rr_f-1}&\cdots&a_{\alpha r_f}^{rr_f-1}
\end{bmatrix},\]
with distinct non-zero $a_j \in \mathbb{F}_q$. This ensures full row rank and the MDS property.

The existence of a suitable $\mathbf{A}$ for 3) and 4) is guaranteed by the following lemma, applied to $\bF' = \bF \otimes \bI_\alpha$ and $\bG' = \bG \otimes \bI_\alpha$.
\begin{lem}
     Given $\bF \in \mathbb{F}_{q}^{n \times n_1}$ and $\bG \in \mathbb{F}_{q}^{n \times n_2}$ with $\Rank(\bF) = n_1$ and $\Rank(\bG) = n_2$, let $t = \dim(\langle \bF \rangle \cap \langle \bG \rangle)$ satisfy $0 < t < n_1$. If $r \leq n_1 - t$ and $q^{n_1-t-r+1} > \binom{n_1-1}{r-1}$, then there exists an invertible $\mathbf{A} \in \mathbb{F}_{q}^{n_1 \times n_1}$ such that any subspace spanned by $r$ columns of $\bF \mathbf{A}$ intersects $\langle \bG \rangle$ only at $\Zero$.
\end{lem}
\begin{pf}
    Let $\mathscr{U} \eqdef \langle \bF \rangle \cap \langle \bG \rangle$ and $\mathscr{W} \eqdef \langle \bF \rangle / \mathscr{U}$, which implies $\dim(\mathscr{W}) = n_1 - t$. Then, we can select $\bw_1, \dots, \bw_{n_1} \in \mathscr{W}$ such that any $r$ of them are linearly independent. We first start with a basis $\bw_1, \dots, \bw_{n_1-t}$ for $\mathscr{W}$. For $j = n_1-t+1, \dots, n_1$, iteratively choose \[\bw_j\in \mathscr{W}\setminus\cup_{R=\{i_1,\cdots,i_{r-1}\}:R\subseteq[j-1]}\Span(\bw_{i_1},\bw_{i_2},\cdots,\bw_{i_{r-1}}).\]
     This is possible because at each step the union covers at most $\binom{j-1}{r-1} q^{r-1}$ vectors, and our field size condition guarantees $q^{n_1-t} - \binom{n_1-1}{r-1} q^{r-1} > 0$.
     
    Let $\pi:\langle\bF\rangle\rightarrow\langle\bF\rangle/\mathscr{U}$ be the quotient map.
    Next, we lift $\bw_1,\cdots,\bw_{n_1}$ to linearly independent vectors $\bff_1,\cdots,\bff_{n_1}$ in $\langle \bF\rangle$.
    Since $\bw_1,\cdots,\bw_{n_1-t}$ form a basis of $\mathscr{W}$, we can randomly select $\bff_1,\cdots,\bff_{n_1-t}$ such that they are linearly independent and $\pi(\bff_{i})=\bw_i$, for $i\in[n_1-t]$. Otherwise, if $\bff_1,\cdots,\bff_{n_1-t}$ are linearly dependent, then $\bw_1,\cdots,\bw_{n_1-t}$ are linearly dependent in the quotient space, which is a contradiction. 
    Let $\mathscr{V}\eqdef\Span(\bff_1,\cdots,\bff_{n_1-t})$, then $\langle\bF\rangle=\mathscr{V}\bigoplus \mathscr{U}$. For every $n_1-t+1\leq j\leq n_1$, we write $\bw_j=\sum_{i=1}^{n_1-t}c_{i,j}\bw_i$. Define $\bv_j=\sum_{i=1}^{n_1-t}c_{i,j}\bff_i\in \mathscr{V}$, so $\pi(\bv_j)=\bw_j$. Select a basis of $\mathscr{U}$, and denote it as $\{\bu_{n_1-t+1},\cdots,\bu_{n_1}\}$. Let $\bff_j=\bv_j+\bu_j$. Then, we have $\pi(\bff_j)=\bw_j$ and $\bff_1,\cdots,\bff_n$ are linearly independent. Otherwise, for the direct sum decomposition of $\langle\bF\rangle=\mathscr{V}\bigoplus \mathscr{U}$, we can write $\bff_i=(\bff_i,\Zero),i\leq n_1-t$ and $\bff_i=(\bv_i,\bu_i)$ for $n_1-t+1\leq i\leq n_1$. Then, for the linear combination 
    \[\sum_{i=1}^{n-t}a_i\bff_i+\sum_{j=n_1-t+1}^{n_1}b_j\bff_j=\Zero,\]
    we can obtain that 
    \[\sum_{i=1}^{n-t}a_i\bff_i+\sum_{j=n_1-t+1}^{n_1}b_j\bv_j=\Zero,~\text{and }\sum_{j=n_1-t+1}^{n_1}b_j\bu_j=\Zero.\]
    Since $\{\bu_j:n_1-t+1\leq j\leq n_1\}$ is a basis of $\mathscr{U}$, then we can obtain that $b_j=0$ for $n_1-t+1\leq j\leq n_1$ and hence $a_i=0$ for $i\in[n_1-t]$. 

    Finally, let $\bF\mathbf{A}=[\bff_1~\bff_2~\cdots~\bff_{n_1}]$. Since the $\bff_i$ are a basis for $\langle\bF\rangle$, $\mathbf{A}$ is invertible. It suffices to prove that the subspace spanned by any $r$ columns of $\bF\mathbf{A}$ intersects $\langle\bG\rangle$ only at $\Zero$. Suppose for contradiction that there exist $i_{1},\cdots,i_r$ such that $\Span(\bff_{i_1},\cdots,\bff_{i_r})\cap\langle\bG\rangle\neq\{\Zero\}$. This is equivalent to $\Span(\bff_{i_1},\cdots,\bff_{i_r})\cap \mathscr{U}\neq\{\Zero\}$, implying $\pi(\bff_{i_1}),\cdots,\pi(\bff_{i_r})$ are linearly dependent in $W$, which is a contradiction. 
\end{pf}
Let $t = \dim(\langle \bF \rangle \cap \langle \bG \rangle)$. The condition $\widehat{\cC}^* = \alpha/r_f$ implies $\alpha/r_f \leq (|\In(\gamma)|-r)/t$, which simplifies to $rr_f \leq \alpha r_f - t\alpha$. Applying the lemma with $n_1 = \alpha r_f$, $n_2 = \alpha r_g$, and the same $t$ (noting $\dim(\langle \bF' \rangle \cap \langle \bG' \rangle) = t\alpha$), we conclude that a matrix $\mathbf{A}$ satisfying 3) and 4) exists if $q^{\alpha r_f - rr_f - t\alpha + 1} > \binom{\alpha r_f - 1}{rr_f - 1}$.

\textbf{Decodability}:
The sink node receives $(M_S~K_S)\cdot \widehat{\mathbf{H}}_{\In(\gamma)}=M_S\cdot \widehat{\mathbf{H}}_{\In(\gamma)}^{M}$. From 3), we have $M_S\cdot \widehat{\mathbf{H}}_{\In(\gamma)}^{M}=M_S\cdot(\bF\otimes\bI_\alpha)\cdot\mathbf{A}=f(M_S)\cdot\mathbf{A}$. Since $\mathbf{A}$ is invertible, the sink node can decode $f(M_S)$ from the received messages.

\textbf{Security}:
We must show that the code remains secure when the wiretapper accesses any $r$ edges.
Let $\widehat{\mathbf{H}}_{\cE}$ be the global encoding matrix for all edges in the branch, then it can be written as
\begin{subequations}
    \begin{align}
        \widehat{\mathbf{H}}_{\cE}=&\begin{bmatrix}
        \widehat{\mathbf{H}}_{\Out(\sigma_1)}&\widehat{\mathbf{H}}_{\Out(\sigma_2)}&\cdots\widehat{\mathbf{H}}_{\Out(\sigma_V)}&\widehat{\mathbf{H}}_{\In(\gamma)}
    \end{bmatrix}\\
    =&\begin{bmatrix}
        \begin{matrix}
            \begin{matrix}
                \widehat{\mathbf{H}}_{\Out(\sigma_1)}^{(\sigma_1,M)}\\
        \widehat{\mathbf{H}}_{\Out(\sigma_1)}^{(\sigma_1,K)}
            \end{matrix}&\Zero&\cdots&\Zero\\ \Zero& \begin{matrix}
                \widehat{\mathbf{H}}_{\Out(\sigma_2)}^{(\sigma_2,M)}\\
        \widehat{\mathbf{H}}_{\Out(\sigma_2)}^{(\sigma_2,K)}
            \end{matrix}&\cdots&\Zero\\ \vdots&\vdots&\ddots&\vdots\\ \Zero&\Zero&\cdots& \begin{matrix}
                \widehat{\mathbf{H}}_{\Out(\sigma_V)}^{(\sigma_V,M)}\\
        \widehat{\mathbf{H}}_{\Out(\sigma_V)}^{(\sigma_V,K)}
            \end{matrix}
        \end{matrix}&\begin{matrix}
            \begin{matrix}
                \widehat{\mathbf{H}}_{\In(\gamma)}^{(\sigma_i,M)}\\
        \Zero
            \end{matrix}\\ \begin{matrix}
                \widehat{\mathbf{H}}_{\In(\gamma)}^{(\sigma_i,M)}\\
        \Zero
            \end{matrix}\\ \vdots\\  \begin{matrix}
                \widehat{\mathbf{H}}_{\In(\gamma)}^{(\sigma_i,M)}\\
        \Zero
            \end{matrix}
        \end{matrix}
    \end{bmatrix}.
    \end{align}
\end{subequations}
Consider any submatrix $\widehat{\mathbf{H}}_{\cE, rr_f}$ formed by $rr_f$ columns of $\widehat{\mathbf{H}}_{\cE}$. Then, proving the security is equivalent to showing
\begin{equation}\label{eq:case1-security}
    \langle\widehat{\mathbf{H}}_{\cE,rr_f}\rangle\cap\langle\mathbf{S}\rangle=\{\Zero\}.
\end{equation} 
where $\mathbf{S}$ is defined in \eqref{eq:S-def}. Suppose $\widehat{\mathbf{H}}_{\cE, rr_f}$ contains $r_{\sigma_i}$ columns from $\widehat{\mathbf{H}}_{\Out(\sigma_i)}$ for each $i$, and $r_\gamma$ columns from $\widehat{\mathbf{H}}_{\In(\gamma)}$, where $\sum_i r_{\sigma_i} + r_\gamma = rr_f$.

If $r_{\gamma}=rr_f$, from 4), we can directly obtain \eqref{eq:case1-security}. Now assume $r_{\gamma}<rr_f$, and for contradiction, we assume that there exists a non-zero vector $\bv\in\langle\widehat{\mathbf{H}}_{\cE,rr_f}\rangle\cap\langle\mathbf{S}\rangle$.  Since $\bv \in \langle \mathbf{S} \rangle$, it has the form
\begin{equation}\label{eq:case1-v}
    \bv=\begin{bmatrix}
        \bv^{(\sigma_1,M)}\\\bv^{(\sigma_1,K)}\\\vdots\\\bv^{(\sigma_V,M)}\\\bv^{(\sigma_V,K)}
    \end{bmatrix}=\begin{bmatrix}
        \bv^{(\sigma_1,M)}\\\Zero\\\vdots\\\bv^{(\sigma_V,M)}\\\Zero
    \end{bmatrix}.
\end{equation}
Because $\bv \in \langle \widehat{\mathbf{H}}_{\cE, rr_f} \rangle$, there exists a coefficient vector $\bw$ such that $\bv = \widehat{\mathbf{H}}_{\cE, rr_f} \bw$.
Furthermore, we can write $\widehat{\mathbf{H}}_{\cE,rr_f}$ and $\bw$ as
\begin{equation}
    \widehat{\mathbf{H}}_{\cE,rr_f}=\begin{bmatrix}
        \begin{matrix}
            \begin{matrix}
                \widehat{\mathbf{H}}_{\Out(\sigma_1)}^{(\sigma_1,M,r_{\sigma_1})}\\
        \widehat{\mathbf{H}}_{\Out(\sigma_1)}^{(\sigma_1,K,r_{\sigma_1})}
            \end{matrix}&\Zero&\cdots&\Zero\\ \Zero& \begin{matrix}
                \widehat{\mathbf{H}}_{\Out(\sigma_2)}^{(\sigma_2,M,r_{\sigma_2})}\\
        \widehat{\mathbf{H}}_{\Out(\sigma_2)}^{(\sigma_2,K,r_{\sigma_2})}
            \end{matrix}&\cdots&\Zero\\ \vdots&\vdots&\ddots&\vdots\\ \Zero&\Zero&\cdots& \begin{matrix}
                \widehat{\mathbf{H}}_{\Out(\sigma_V)}^{(\sigma_V,M,r_{\sigma_V})}\\
        \widehat{\mathbf{H}}_{\Out(\sigma_V)}^{(\sigma_V,K,r_{\sigma_V})}
            \end{matrix}
        \end{matrix}&\begin{matrix}
            \begin{matrix}
                \widehat{\mathbf{H}}_{\In(\gamma)}^{(\sigma_1,M,r_{\gamma})}\\
        \Zero
            \end{matrix}\\ \begin{matrix}
                \widehat{\mathbf{H}}_{\In(\gamma)}^{(\sigma_2,M,r_{\gamma})}\\
        \Zero
            \end{matrix}\\ \vdots\\  \begin{matrix}
                \widehat{\mathbf{H}}_{\In(\gamma)}^{(\sigma_V,M,r_\gamma)}\\
        \Zero
            \end{matrix}
        \end{matrix}
    \end{bmatrix},~~\bw=\begin{bmatrix}
        \bw^{(\sigma_1)}\\\bw^{(\sigma_2)}\\\vdots\\\bw^{(\sigma_V)}\\\bw^{(\gamma)}
    \end{bmatrix},
\end{equation}
where $\widehat{\mathbf{H}}_{\Out(\sigma_i)}^{(\sigma_i,M,r_{\sigma_i})}\in\mathbb{F}_q^{\alpha\times r_{\sigma_i}},\widehat{\mathbf{H}}_{\Out(\sigma_i)}^{(\sigma_i,K,r_{\sigma_i})}\in\mathbb{F}_q^{rr_f\times r_{\sigma_i}},\widehat{\mathbf{H}}_{\In(\gamma)}^{(\sigma_i,M,r_{\gamma})}\in\mathbb{F}_q^{\alpha\times r_\gamma}$ and $\bw^{(\sigma_i)}\in\mathbb{F}_q^{r_{\sigma_i}},\bw^{(\gamma)}\in\mathbb{F}_q^{r_{\gamma}}$.
Then, for every $i\in[V]$, we have
\begin{equation}\label{eq:case1-security-final}
    \Zero=\bv^{(\sigma_i,K)}=\widehat{\mathbf{H}}_{\Out(\sigma_i)}^{(\sigma_i,K,r_{\sigma_i})}\cdot\bw^{(\sigma_i)},
\end{equation}
where $\widehat{\mathbf{H}}_{\Out(\sigma_i)}^{(\sigma_i,K,r_{\sigma_i})}$ is formed by $r_{\sigma_i}$ columns of $\widehat{\mathbf{H}}_{\Out(\sigma_i)}^{(\sigma_i,K)}$. Consequently, \eqref{eq:case1-security-final} contradicts with the fact that $\widehat{\mathbf{H}}_{\Out(\sigma_i)}^{(\sigma_i,K)}$ is an MDS matrix, which completes the proof of security.

\textbf{Case 2}:
In this case, we construct an $(\ell,n)$ linear secure network for a branch under the condition $\widehat{\cC}^*=\frac{\ell}{n}<\frac{\alpha}{r_f}$. 

First, we present an example to illustrate the construction.
\begin{example}\label{ex:case2}
  Consider the network shown in Fig.~\ref{fig:233branch}, which is a branch for $(U,3,3)$-tree. Let $q>1296$ be a prime and the coefficient matrices of the target and security function be 
  \[\bF=\begin{bmatrix}
    1&2&0\\
    2&1&2\\
    0&1&1        
    \end{bmatrix},~~~
    \bG=\begin{bmatrix}
    1&2\\
    2&1\\
    0&1       
    \end{bmatrix}.\]
  Note that the matrix $\bF$ can be partitioned into two parts $\bF_1$ and $\bF_2$, such that $\Rank(\bF_1)=\dim(\langle\bF_1\rangle)=\dim(\langle\bF\rangle\cap\langle\bG\rangle)$. Specifically,
  \[\bF_1=\begin{bmatrix}
    1&2\\
    2&1\\
    0&1        
    \end{bmatrix},~~~
    \bF_2=\begin{bmatrix}
    0\\
    2\\
    1        
    \end{bmatrix}.\]
  When $r=2$, from Theorem~\ref{thm:upperbound1}, we can obtain an upper bound $\widehat{\cC}^*=(\alpha-r)/t=1/2$.
  Next, we construct a $(1,2)$ network code as follows. Note that each source node can transmit at most $n\alpha=6$ symbols to the intermediate node. Therefore, $\sigma_i$ can generates at most $5$ random keys to protect $1$ source symbol. Similar to Case 1, $\sigma_i$ sends $\by_{\Out(\sigma_i)}=(\mathbf{m}_i~\bk_i)\widehat{\mathbf{H}}_{\Out(\sigma_i)}^{(\sigma_i)}$ to the intermediate node, where $\widehat{\mathbf{H}}_{\Out(\sigma_i)}^{(\sigma_i)}\in\mathbb{F}_q^{6\times 6}$ is selected as
  \[\widehat{\mathbf{H}}_{\Out(\sigma_i)}^{(\sigma_i)}=\begin{bmatrix}
        1&0&0&0&0&0\\
        1&1&1&1&1&1\\
        1&2&3&4&5&6\\
        1&4&9&16&25&36\\
        1&8&27&64&125&216\\
        1&16&81&384&625&1296
    \end{bmatrix}.\]
  Since $\widehat{\mathbf{H}}_{\Out(\sigma_i)}^{(\sigma_i)}$ is full rank, the intermediate node can decode $\mathbf{m}_i$ and $\mathbf{k}_i$ for $i\in[3]$. Moreover, since the intermediate node can send at most $6$ symbols to $\gamma$ and the value of target function contains $3$ symbols, the transmitted symbols contain at most $3$ random keys. Therefore, the intermediate node generates $3$ new random keys from $\bk_1,\bk_2$ and $\bk_3$ as $k'_j=\sum_{i\in[3]}k_{i,j}$, for $j\in[3]$. This procedure can be represented in the matrix form as $(k_1',k_2',k_3')=\bk_S\cdot \bU$, where $\bU\in\mathbb{F}_q^{15\times 3}$ is selected as
  \[\bU^T=\begin{bmatrix}
      \bI_3&\Zero_{2\times3}&\bI_3&\Zero_{2\times3}&\bI_3&\Zero_{2\times3}
  \end{bmatrix}.\]
  Since only the symbols in $\mathbf{m}_S\bF_1$ are desired by the wiretapper, we regard the symbols in $\mathbf{m}_S\bF_2$ as additional random keys. Then, the intermediate node uses $\mathbf{m}_S\bF_2$ and $k'_j,j\in[3]$ to protect the symbols in $\mathbf{m}_S\bF_1$. Specifically, on the edges $\In(\gamma)$, $\by_{\In(\gamma)}=(\mathbf{m}_S~\bk_S)\widehat{\mathbf{H}}_{\In(\gamma)}=\mathbf{m}_S\widehat{\mathbf{H}}_{\In(\gamma)}^{(M)}+\bK_S\widehat{\mathbf{H}}_{\In(\gamma)}^{(K)}$, where $\widehat{\mathbf{H}}_{\In(\gamma)}^{(M)}$ and $\widehat{\mathbf{H}}_{\In(\gamma)}^{(K)}$ have the form
  \[\begin{bmatrix}
        \widehat{\mathbf{H}}_{\In(\gamma)}^{(M)}\\\widehat{\mathbf{H}}_{\In(\gamma)}^{(K)}
    \end{bmatrix}=\begin{bmatrix}
        \bF&\Zero_{3\times3}\\\Zero_{15\times 3}&\bU
    \end{bmatrix}\cdot\mathbf{D},~~\text{with}~~\mathbf{D}=
    \begin{bmatrix}
       1&0&0&0&0&0\\
       0&1&0&0&0&0\\
       0&0&0&1&1&1\\
       1&0&0&8&9&10\\
       0&1&0&64&81&100\\
       0&0&1&512&729&1000
    \end{bmatrix}.\]
   Note that $\mathbf{D}$ is invertible, hence $\gamma$ can decode $\mathbf{m}_S\bF=f(\mathbf{m}_S)$.
   
   The security can be checked as follows. In all transmitted symbols, the coefficient matrix of random keys $\widehat{\mathbf{H}}_{\cE}^{(K)}$ in the construction is
   
   \[\widehat{\mathbf{H}}_{\cE}^{(K)}=\begin{bmatrix}
       \widehat{\mathbf{H}}_{\Out(\sigma_1)}^{(\sigma_1,K)}&\Zero_{5\times 6}&\Zero_{5\times 6}&\widehat{\mathbf{H}}_{\In(\gamma)}^{(\sigma_1,K)}\\
       \Zero_{5\times 6}&\widehat{\mathbf{H}}_{\Out(\sigma_2)}^{(\sigma_2,K)}&\Zero_{5\times 6}&\widehat{\mathbf{H}}_{\In(\gamma)}^{(\sigma_2,K)}\\
       \Zero_{5\times 6}&\Zero_{5\times 6}&\widehat{\mathbf{H}}_{\Out(\sigma_3)}^{(\sigma_3,K)}&\widehat{\mathbf{H}}_{\In(\gamma)}^{(\sigma_3,K)}
   \end{bmatrix},\]
   with
   \[\widehat{\mathbf{H}}_{\Out(\sigma_i)}^{(\sigma_i,K)}=\begin{bmatrix}
       1&1&1&1&1&1\\
        1&2&3&4&5&6\\
        1&4&9&16&25&36\\
        1&8&27&64&125&216\\
        1&16&81&384&625&1296
   \end{bmatrix},~~~\widehat{\mathbf{H}}_{\In(\gamma)}^{(\sigma_i,K)}=\begin{bmatrix}
       1&0&0&8&9&10\\
       0&1&0&64&81&100\\
       0&0&1&512&729&1000\\
       0&0&0&0&0&0\\
       0&0&0&0&0&0
   \end{bmatrix}.\]
   If the wiretapper accesses $2$ edges ($4$ columns) between the first two layers, it cannot eliminate the random keys since $\widehat{\mathbf{H}}_{\Out(\sigma_i)}^{(\sigma_i,K)}$ is an MDS matrix. If the wiretapper accesses $1$ edge between the first two layers and $1$ edge in $\In(\gamma)$, it also cannot eliminate the random keys since any $2$ columns in the first $18$ columns in $\widehat{\mathbf{H}}_{\cE}^{(K)}$ are linearly independent with any $2$ columns in the last $6$ columns in $\widehat{\mathbf{H}}_{\cE}^{(K)}$. If the wiretapper accesses $2$ edges in $\In(\gamma)$, since $\by_{\In(\gamma)}$ consists of $6$ linear combinations of $\mathbf{m}_S\bF,k_1',k_2'$ and $k_3'$, and the coefficient matrix corresponding to the symbols, which the wiretapper do not want, is an MDS matrix, the security is hence guaranteed. 
\end{example}

\textbf{Notations}: 
Each source node $\sigma_i$ generates a source message $M_i$ of length $\ell$ and a random key $K_i$ of length $n\alpha-\ell$, both uniformly distributed over $\mathbb{F}_q$. For the edges in $\Out(\sigma_i),i\in[V]$, we still use the notation $\widehat{\mathbf{H}}_e\in\mathbb{F}_q^{Vn\alpha\times n},\widehat{\mathbf{H}}_{\Out(\sigma_i)}\in\mathbb{F}_q^{Vn\alpha\times n\alpha}$ to denote the global encoding matrices. Similar to Case 1, $\widehat{\mathbf{H}}_{\Out(\sigma_i)}=[\widehat{\mathbf{H}}_{\Out(\sigma_i)}^{(\sigma_1)};$ $\cdots;\widehat{\mathbf{H}}_{\Out(\sigma_i)}^{(\sigma_V)}]$, and $\widehat{\mathbf{H}}_{\Out(\sigma_i)}^{(\sigma_j)}$ has the form \eqref{eq:scheme2-case1-out} if $i=j$, otherwise, $\widehat{\mathbf{H}}_{\Out(\sigma_i)}^{(\sigma_j)}=\Zero$.

Consider the transmission on the edges in $\In(\gamma)$. Let $t\eqdef\dim(\langle\bF\rangle\cap\langle\bG\rangle)$.
Without loss of generality, we assume that $\bF=[\bF_1~\bF_2]$ with $\bF_1\in\mathbb{F}^{V\times t}$ and $\bF_2\in\mathbb{F}^{V\times (r_f-t)}$ such that $\langle\bF_1\rangle\cap\langle\bG\rangle=\langle\bF\rangle\cap\langle\bG\rangle$ and $\langle\bF_2\rangle\cap\langle\bG\rangle=\Zero$. The middle node uses the random keys generated by sources to produce $n\alpha-r_f\ell$ new keys through a matrix $\bU\in\mathbb{F}_q^{V(n\alpha-\ell)\times(n\alpha-r_f\ell)}$. The global encoding matrix $\widehat{\mathbf{H}}_{\In(\gamma)}\in\mathbb{F}_q^{Vn\alpha\times n\alpha}$ for $\In(\gamma)$ has the form \eqref{eq:scheme2-case1In}, let $\widehat{\mathbf{H}}_{\In(\gamma)}^{(M)}\in\mathbb{F}_q^{V\ell\times n\alpha}$ and $\widehat{\mathbf{H}}_{\In(\gamma)}^{(K)}\in\mathbb{F}_q^{V(n\alpha-\ell)\times n\alpha}$ be the submatrices corresponding to source messages and random keys. Moreover, these two matrices satisfy
\begin{equation}
    \begin{bmatrix}
        \widehat{\mathbf{H}}_{\In(\gamma)}^{(M)}\\\widehat{\mathbf{H}}_{\In(\gamma)}^{(K)}
    \end{bmatrix}=\begin{bmatrix}
        \bF\otimes \bI_\ell&\Zero_{V\ell\times(n\alpha-r_f\ell)}\\\Zero_{V(n\alpha-\ell)\times r_f\ell}&\bU
    \end{bmatrix}\cdot\mathbf{D},
\end{equation}
where $\mathbf{D}\in\mathbb{F}_q^{n\alpha\times n\alpha}$. Let $\widehat{\mathbf{H}}_{\cE}$ and $\widehat{\mathbf{H}}_{\Out}$ denote the global encoding matrices for the edges $\cE$ and $\cup_{i\in[V]}\Out(\sigma_i)$. Moreover, let the $\widehat{\mathbf{H}}_{\cE}^{(K)}$ and $\widehat{\mathbf{H}}_{\Out}^{(K)}$ be the submatrices corresponding to $K_S$ in $\widehat{\mathbf{H}}_{\cE}$ and $\widehat{\mathbf{H}}_{\Out}$.

\textbf{Construction}:
The code is specified by the following conditions on the encoding matrices.
\begin{enumerate}
    \item For each $i \in [V]$, $\widehat{\mathbf{H}}_{\Out(\sigma_i)}^{(\sigma_i)}$ is of full row rank;
    \item For each $i \in [V]$, $\widehat{\mathbf{H}}_{\Out(\sigma_i)}^{(\sigma_i,K)}$ is an MDS matrix, i.e., any $n\alpha-\ell$ columns are linearly independent; 
    \item $\Rank(\bU)=n\alpha-r_f\ell$;
    \item $\mathbf{D}$ is invertible and the last $n\alpha-t\ell$ rows of $\mathbf{D}$ form an MDS matrix;
    \item For $a\leq n\alpha-r_f\ell$, any $a$ columns of $\widehat{\mathbf{H}}_{\In(\gamma)}^{(K)}$ and any $nr-a$ columns of $\widehat{\mathbf{H}}_{\Out}^{(K)}$ are linearly independent.
\end{enumerate}

\textbf{Existence}: We can find matrices satisfying above conditions as follows. First, we set 
\begin{equation}
    \widehat{\mathbf{H}}_{\Out(\sigma_i)}^{(\sigma_i)}=\begin{bmatrix}
        \begin{matrix}
            \bI_{\ell}&\Zero_{\ell\times (n\alpha-\ell)}
        \end{matrix}\\\widehat{\mathbf{H}}_{\Out(\sigma_i)}^{(\sigma_i,K)}
    \end{bmatrix}=\begin{bmatrix}
        \begin{matrix}
            \bI_{\ell}&\Zero_{\ell\times (n\alpha-\ell)}
        \end{matrix}\\\mathbf{V}_{(n\alpha-\ell)\times n\alpha}(\ba)
    \end{bmatrix},
\end{equation}
where $\mathbf{V}_{(n\alpha-\ell)\times n\alpha}(\ba)$ is a Vandermonde matrix with $\ba=(1,2,\cdots,n\alpha)$. Hence, the conditions 1) and 2) are satisfied. Next, we select $\mathbf{U}$ as
\begin{equation}
    \bU=\begin{bmatrix}
        \bI_{n\alpha-r_f\ell}\\\Zero_{(r_f\ell-\ell)\times (n\alpha-r_f\ell)}\\\vdots\\\bI_{n\alpha-r_f\ell}\\\Zero_{(r_f\ell-\ell)\times (n\alpha-r_f\ell)}
    \end{bmatrix}.
\end{equation}
This implies that the middle node generates $n\alpha$ new keys $K_j'=\sum_{i\in[V]}K_{i,j}$, where $K_{i,j}$ is the $j$-th symbol in the random key generated by the $i$-th source node $K_i$. Therefore, the condition 3) is satisfied. At last, we select $\mathbf{D}=[\mathbf{D}_1;\mathbf{D}_2]$ with $\mathbf{D}_1\in\mathbb{F}_q^{t\ell\times n\alpha}$ and $\mathbf{D}_2\in\mathbb{F}_q^{(n\alpha-t\ell)\times n\alpha}$,
\begin{equation}
    \mathbf{D}_1=\begin{bmatrix}
    \bI_{t\ell}&\Zero_{t\ell\times (n\alpha-t\ell)}     
    \end{bmatrix},~~\mathbf{D}_2=\begin{bmatrix}
        \begin{matrix}
            \Zero_{(r_f\ell-t\ell)\times(n\alpha-r_f\ell)}\\\bI_{n\alpha-r_f\ell}
        \end{matrix}&\mathbf{V}_{(n\alpha-t\ell)\times (r_f\ell)}(\ba')
    \end{bmatrix},
\end{equation}
where $\mathbf{V}_{(n\alpha-t\ell)\times r_f\ell}(\ba')$ is a Vandermonde matrix with $\ba'=(n\alpha+1,\cdots,n\alpha+r_f\ell)$. If the field size $q>\max\{(n\alpha)^{n\alpha-\ell-1},(n\alpha+r_f\ell)^{n\alpha-t\ell-1}\}$ is a prime, then the matrix $\mathbf{D}_2$ is an MDS matrix and hence condition 4) is satisfied. At last, it suffices to check condition 5) is also satisfied. By the construction above, $\widehat{\mathbf{H}}_{\In(\gamma)}^{(K)}$and $\widehat{\mathbf{H}}_{\Out}^{(K)}$ have the following form
\begin{equation}
    \widehat{\mathbf{H}}_{\Out}^{(K)}=\bI_V\otimes
        \mathbf{V}_{(n\alpha-\ell)\times n\alpha}
    ,~~\widehat{\mathbf{H}}_{\In(\gamma)}^{(K)}=\mathbf{1}_{V}\otimes\begin{bmatrix}
        \begin{matrix}
            \bI_{n\alpha-r_f\ell}&\mathbf{V}_{(n\alpha-r_f\ell)\times r_f\ell}
        \end{matrix}\\\Zero_{(r_f\ell-\ell)\times n\alpha}
    \end{bmatrix},
\end{equation}
where $\mathbf{1}_{V}$ is a column vector of length $V$ with all elements $1$. 
Note that every column of $\widehat{\mathbf{H}}_{\In(\gamma)}^{(K)}$ can be written as a linear combination of its first $n\alpha-r_f\ell$ columns.
Moreover, since the matrix 
\begin{equation}
    \begin{bmatrix}
        \mathbf{V}_{(n\alpha-\ell)\times n\alpha}&\begin{matrix}
            \bI_{n\alpha-r_f\ell}\\\Zero_{(r_f\ell-\ell)\times(n\alpha-r_f\ell)}
        \end{matrix}
    \end{bmatrix}
\end{equation}
can be an MDS matrix for sufficiently large $q$, any $n\alpha-\ell$ columns are linearly independent. Together with the fact that $n\alpha-\ell>nr$, then we can conclude that condition 5) is satisfied.

\textbf{Decodability}: The sink node receives 
\begin{equation*}
    \bY_{\In(\gamma)}=[M_S~K_S]\cdot\widehat{\mathbf{H}}_{\In(\gamma)}=[M_S(\bF\otimes \bI_\ell)~~K_S\bU]\cdot\mathbf{D}.
\end{equation*}
Then the decodability follows from the invertibility of $\mathbf{D}$ in condition 4).

\textbf{Security}: Let $r_{\sigma}$ and $r_{\gamma}$ be the number of wiretap edges in $\cup_i\Out(\sigma_i)$ and $\In(\gamma)$. Then, $r_\sigma+r_\gamma=r$. 

If $r_\sigma=r$, then the security follows from the MDS property of $\widehat{\mathbf{H}}_{\Out(\sigma_i)}^{(\sigma_i,K)}$ as discussed in Case 1.

If $r_\sigma<r$, let $nr_\gamma=a$, then the $nr$ symbols obtained by the wiretapper contain random keys with the coefficient matrix formed by $a$ columns of $\widehat{\mathbf{H}}_{\In(\gamma)}^{(K)}$ and $nr-a$ columns of $\widehat{\mathbf{H}}_{\Out}^{(K)}$. 
If $a<nr-r_f\ell$, by condition 5), these columns are linearly independent, which implies the random keys cannot be eliminated. Hence, the security is guaranteed.
If $a>nr-r_f\ell$, suppose there exists a nonzero vector $\bv\in\langle\widehat{\mathbf{H}}_{\cE,nr}\rangle\cap\langle\mathbf{S}\rangle$, where $\mathbf{S}$ is defined in \eqref{eq:S-def}. As in \eqref{eq:case1-v}, the coefficients corresponding to $K_S$ in $\bv$ should be zero, that is, $\bv^{(K)}=\Zero$. At the same time, $\bv$ can be also represented by a linear combination of the columns in $\widehat{\mathbf{H}}_{\cE,nr}$, i.e., there exists a nonzero vector $\bw$ such that $\bv=\widehat{\mathbf{H}}_{\cE,nr}\bw$. Accordingly, $\bv^{(K)}=\widehat{\mathbf{H}}_{\cE,nr}^{(K)}\bw$. 
Without loss of generality, denote the $nr$ columns of $\widehat{\mathbf{H}}_{\cE,nr}^{(K)}$ as $\bc_1,\cdots,\bc_{nr}$.
Note that the last $a$ columns $\bc_{nr-a+1},\cdots,\bc_{nr}$ are columns from $\widehat{\mathbf{H}}_{\In(\gamma)}^{(K)}$. Hence, these $a$ columns can be represented by the $n\alpha-r_f\ell$ columns of them due to condition 4) and the choice of $\bU$. Let $\bw=(w_1,\cdots,w_{nr})$,  then $\Zero=\bv^{(K)}=\sum_{i=1}^{nr}w_i\bc_i=\sum_{i=1}^{nr-a+n\alpha-r_f\ell}w'_i\bc_i$, where $w_i'=w_i$ if $i\leq nr-a$ and $w'_i$ is a linear combination of $w_{nr-a+1},\cdots,w_{nr}$ for $nr-a+1\leq i\leq nr$. By condition 5), $\bc_1,\cdots,\bc_{nr-a+n\alpha-r_f\ell}$ are linearly independent, which implies $w_i'=0$ for all $i=1,2,\cdots,nr-a+n\alpha-r_f\ell$. Therefore, $w_i=0$ for $i\in[nr-a]$, that is, $\bv$ is a linear combination of $\widehat{\mathbf{H}}_{\In(\gamma)}$ or equivalently $r_\gamma=r$. 
In this case, the security follows from the MDS property of the last $n\alpha-t\ell$ rows of $\mathbf{D}$, which guarantees that $\bF_1\otimes\bI_{\ell}$ is protected by $\bF_2\otimes\bI_\ell$ and $K_j',j\in[n\alpha-r_f\ell]$. Hence, the security of the network code is proved.

Combining the above two cases, we can conclude that the upper bound in Corollary~\ref{cor:ub-branch} is achievable. 
\begin{thm}
    For the secure network function computation problem over the branch $(\cN,f,g,r)$, the secure computing capacity 
    \[\widehat{\cC}(\cN,f,g,r)=\min\left\{\frac{\alpha}{r_f},\min_{(C,W)\ \text{valid}}\frac{|C|-|W|}{t_{C,f,g}},\min_{(C,B,W)\ \text{valid}}\frac{|C\cup B|-|W|}{t_{C,f,g}}\right\},\]
    where $t_{C,f,g}=\dim\left(\langle \bF_{I_{C_1}}\rangle\bigoplus\cdots\langle \bF_{I_{C_p}}\rangle\bigcap\langle \bG\rangle\right).$
\end{thm}

\subsection{Linear secure network codes for $(U,V,\alpha)$-tree}
In last subsection, it has been shown that for each branch $\cN_i$, we can construct a linear secure network code with computing rate $\ell/n=\widehat{\cC}^*$, which matches the upper bound. In fact, for every $\ell',n'$ with $\ell'/n'<\ell/n$, we can obtain a secure network code with computing rate $\ell'/n'$.
To see this, let ${\rm{lcm}}(n',n)=n'p'=np$, where ${\rm{lcm}}(n',n)$ is the least common multiple of $n'$ and $n$. Then, we have $\frac{\ell't'}{n't'}=\frac{\ell'}{n'}$ and $\frac{\ell t}{nt}=\frac{\ell}{n}$, $\ell't'<\ell t$. By using the $(\ell,n)$ network code $t$ times, we can obtain an $(\ell t,nt)$ linear secure network code. By setting $\ell t-\ell't'$ symbols of the source messages to zero,  we obtain an  $(\ell't',nt)$ secure network code with computing rate $\frac{\ell't'}{nt}=\frac{\ell't'}{n't'}=\frac{ell'}{n'}$. 

Now consider two branches $\cN_1$ and $\cN_2$ with capacities $\widehat{\cC}(\cN_1,f_1,g_1,r_1) = \ell_1/n_1$ and $\widehat{\cC}(\cN_2,f_2,g_2,r_2) = \ell_2/n_2$. Using the argument above, a secure code with rate $\ell/n = \min\{\ell_1/n_1,\ell_2/n_2\}$ can be constructed for the combined system.

Consequently, we obtain the following theorem.
\begin{thm}
    For the secure network function computation problem over the tree $(\cN,f,g,r)$, the secure computing capacity 
    \[\widehat{\cC}(\cN,f,g,r)\geq\min_{\substack{(r_1,r_2,\cdots,r_U):\\\sum_{i\in[U]}r_i=r}}\left\{\min_{i\in[U]}\widehat{\cC}(\cN_i,f_i,g_i,r_i)\right\}=\min_{i\in[U]}\widehat{\cC}(\cN_i,f_i,g_i,r).\]
\end{thm}
\begin{pf}
    The equality holds because $\widehat{\cC}(\cN_i,f_i,g_i,r) \leq \widehat{\cC}(\cN_i,f_i,g_i,r_i)$ for any $r_i \leq r$. Let $\ell/n = \min_{i} \widehat{\cC}(\cN_i,f_i,g_i,r)$. For each branch, we can construct an $(\ell,n)$ linear secure network code as discussed before. Hence, for every wiretap set in $\cW$, the constructed code for each branch is secure with rate $\ell/n$. Then, the security of the overall code follows from Proposition~\ref{prop:tree2branch}.
\end{pf}

Proposition~\ref{prop:extension} shows that securing every branch is not necessary. This leads to the following refined result.
\begin{thm}
    Let $\cN$ be a tree network composed of $U$ branches $\{\cN_i : i \in [U]\}$. Suppose that for each branch $\cN_i$, there exists an $(\ell,n)$ linear secure network code with security level $\bar{r}_i$. Then the security level $r$ of the resulting $(\ell,n)$ code for the whole tree $\cN$ satisfies that for every $(r_1,r_2,\cdots,r_U)$ with $\sum_{i\in[U]}r_i=r$, there exists a subset $T\subseteq [U]$ such that the condition \eqref{eq:prop-condition} holds, and $r_i\leq \bar{r}_i, \forall~ i\in T$.
\end{thm}
\begin{pf}
    Decodability follows directly from the decodability of the branch codes. To verify security, note that for any wiretap set of size $r$, the conditions of the theorem guarantee the existence of a subset $T$ for which each branch code $\cN_i$ ($i \in T$) remains secure. By Proposition~\ref{prop:extension}, the overall code is secure for this wiretap set, and thus achieves security level $r$. 
\end{pf}

Finally, we present an example to illustrate  the construction.
\begin{figure}
    \centering
    \includegraphics[width=0.6\linewidth]{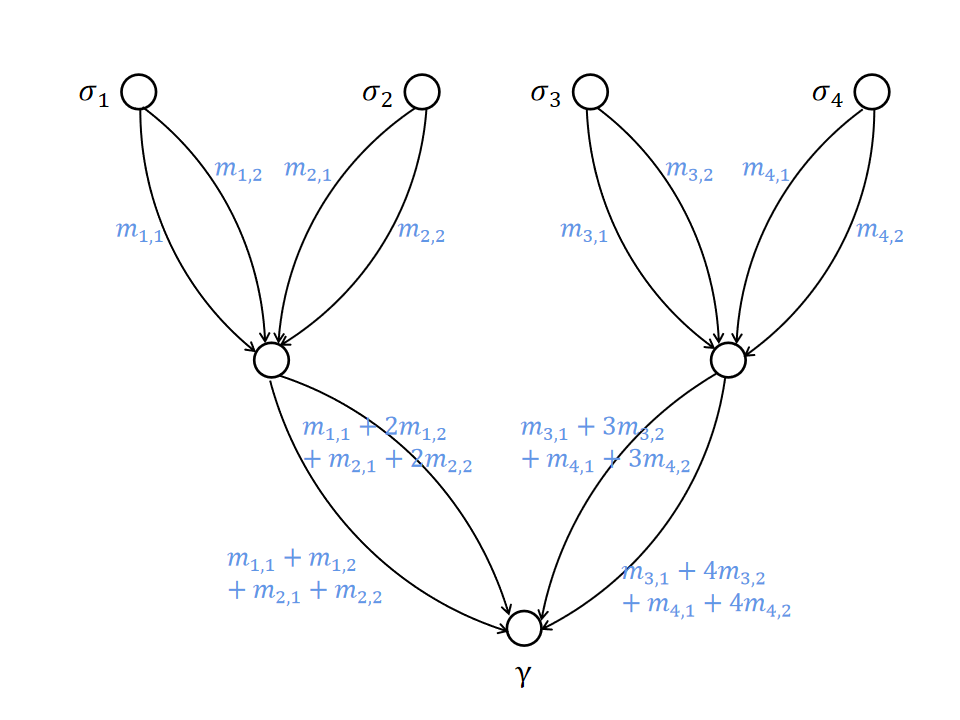}
    \caption{A $(2,2,2)$-tree and a secure network code with security level $r=2$.}
    \label{fig:tree1}
\end{figure}
\begin{figure}
    \centering
    \begin{minipage}{0.49\linewidth}
        \centering
        \includegraphics[width=1\linewidth]{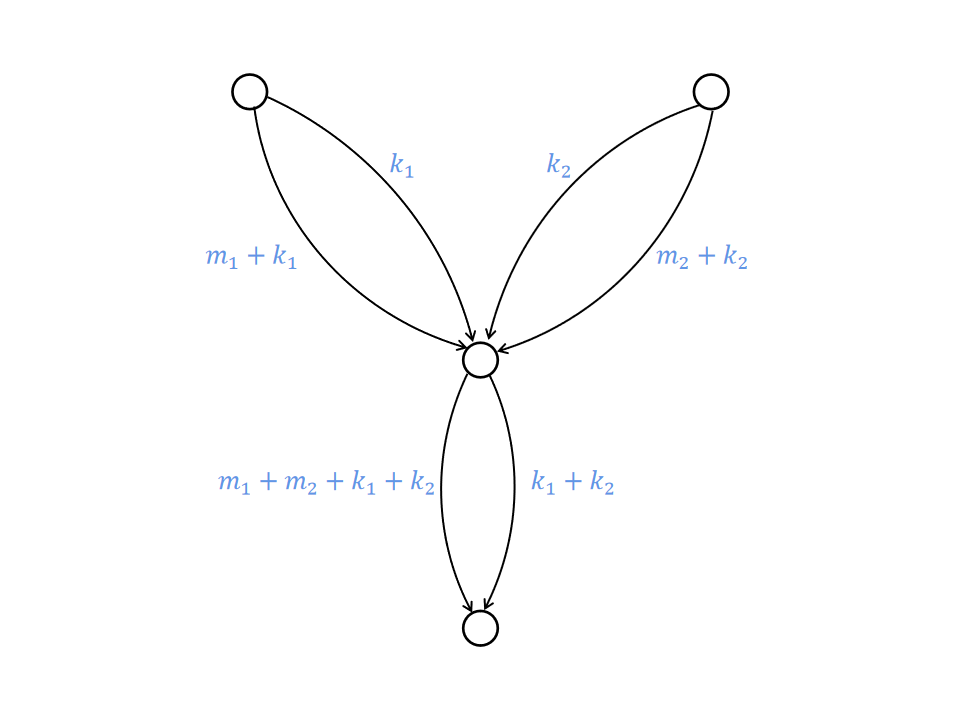}
        \caption{Network code with rate $1$ for each branch.}
        \label{fig:tree2}
    \end{minipage}
    \begin{minipage}{0.49\linewidth}
        \centering
        \includegraphics[width=1\linewidth]{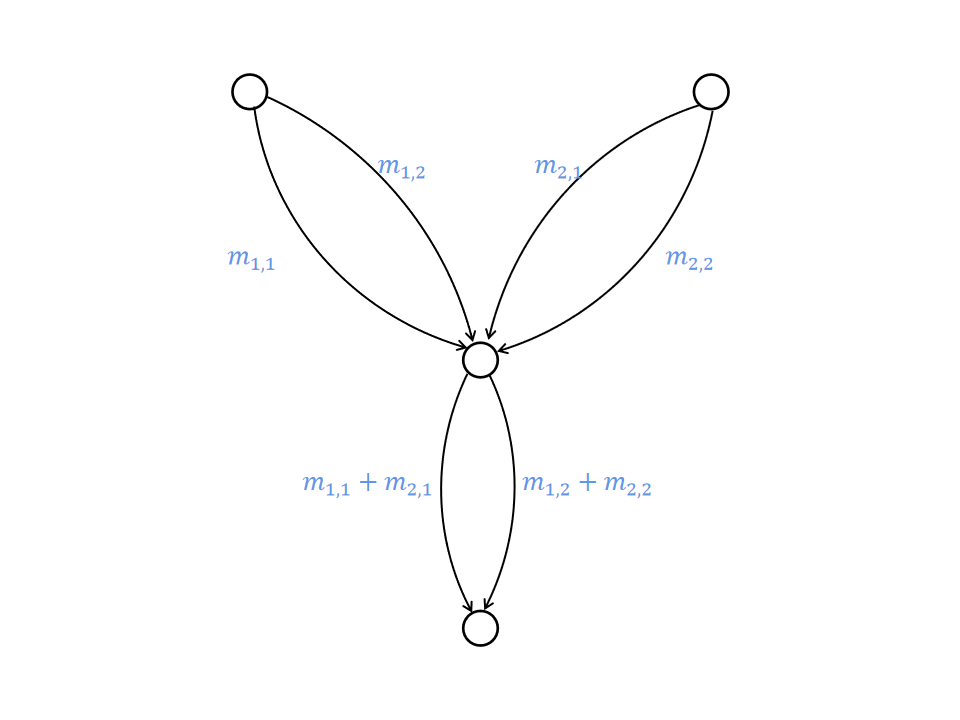}
        \caption{Network code with rate $2$ for each branch.}
        \label{fig:tree3}
    \end{minipage}
\end{figure}
\begin{example}
    Consider a $(2,2,2)$-tree as shown in Fig.~\ref{fig:tree1}, the target function and security function are over $\mathbb{F}_5$ and the coefficient matrices are
    \[\bF^T=\begin{bmatrix}
        1&1&1&1\\
        2&2&3&3
    \end{bmatrix},~~\bG^T=\begin{bmatrix}
        1&1&1&1
    \end{bmatrix}.\]
    When $r=1$, if we follow the idea of Proposition~\ref{prop:tree2branch} to construct secure network code over each branch, the secure computing capacity is given by $\widehat{\cC}(\cN_i,f_i,g_i,1)=1$. The corresponding network code for each branch is shown in Fig.~\ref{fig:tree2}. From Proposition~\ref{prop:extension}, we can select a $(2,1)$ network code for each branch as shown in Fig.~\ref{fig:tree3}, which is not secure for $(\cN_i,f_i,g_i,1)$. But for the tree, the network code is secure for $(\cN,f,g,1)$, since only one edge is wiretapped, and the wiretapper only obtains the information about $m_1,m_2$ or $m_3,m_4$, but cannot obtain them simultaneously. Thus, the wiretapper cannot obtain any information about $g(\mathbf{m}_S)$. Moreover, from the upper bound in Theorem~\ref{thm:csbnd}, we have
    \[\widehat{\cC}(\cN,f,g,1)\leq\cC(\cN,f)=2,\]
    which implies the constructed network code attains the upper bound in this case.

    However, when $r=2$, constructing network codes followed the idea in Proposition~\ref{prop:extension} is not enough. It requires to design the network codes for all branches together to get better performance. From Proposition~\ref{prop:extension}, since at least one branch should be secure for any choice of wiretap edges, we should select a $(1,1)$ secure network code as shown in Fig.~\ref{fig:tree2} for each branch. But in fact, if we design the network code for each branch together, we can construct a secure network code with rate $2$ as shown in Fig.~\ref{fig:tree1}. Note that if the wiretapper has access to exactly one edge in each branch, the network code for each branch in the construction is not secure, but the wiretapper still cannot obtain any information about $g(\mathbf{m}_S)$. This is because the coefficients of the transmitted messages are carefully designed. 
\end{example}

In summary, for the branch model we can construct vector linear secure network codes that achieve the upper bound. For $(U,V,\alpha)$-trees, the codes obtained via Proposition~\ref{prop:extension} attain the bound in some cases, but not in others (as illustrated by the example above). The underlying reason is that Proposition~\ref{prop:extension} only gives a sufficient condition for security. The example shows that a network code can be secure even if it is insecure on every individual branch.

\section{Conclusion}\label{sec:conclusion}
In this paper, we have investigated the secure network function computation problem for vector linear target and security functions. We have derived two new upper bounds on the secure computing capacity that apply to arbitrary network topologies and generalize existing results. For the source security problem considered in \cite{2024Guangsourcesecure}, our result offers a new upper bound, which has better performance in an explicit example. Following the standard procedure, we have presented a construction of linear secure network code when the target function is algebraic sum. Moreover, for vector linear target function and $(U,V,\alpha)$-trees, we have constructed a vector linear secure network code by constructing secure network codes for the branches. 

For the model considered in this paper, there are several interesting problems for future research. First, the upper bounds presented in this paper are not always tight. In fact, even for the network function computation problem without the security constraint, the exact value of the computing capacity for arbitrary vector linear target function and arbitrary networks is not known. Second, the minimum required field size for the secure network code constructed in Section~\ref{sec:lowerbound-sum} is not known yet. Third, since the secure network code constructed by the standard procedure always has a gap from the upper bound, new code constructions based on different idea should be investigated. Our tree-network code is one such attempt, though it requires a restrictive topology.

\bibliographystyle{IEEEtranS}
\bibliography{reference}

@ARTICLE{2019GYYL,
  author={Guang, Xuan and Yeung, Raymond W. and Yang, Shenghao and Li, Congduan},
  journal={IEEE Transactions on Information Theory}, 
  title={Improved Upper Bound on the Network Function Computing Capacity}, 
  year={2019},
  volume={65},
  number={6},
  pages={3790-3811},
  doi={10.1109/TIT.2019.2893107}}

@ARTICLE{2011AFKZ,
  author={Appuswamy, Rathinakumar and Franceschetti, Massimo and Karamchandani, Nikhil and Zeger, Kenneth},
  journal={IEEE Transactions on Information Theory},
  title={Network Coding for Computing: Cut-Set Bounds},
  year={2011},
  volume={57},
  number={2},
  pages={1015-1030},
  doi={10.1109/TIT.2010.2095070}}

@ARTICLE{2013AFKZ,
  author={Appuswamy, Rathinakumar and Franceschetti, Massimo and Karamchandani, Nikhil and Zeger, Kenneth},
  journal={IEEE Transactions on Information Theory},
  title={Linear Codes, Target Function Classes, and Network Computing Capacity},
  year={2013},
  volume={59},
  number={9},
  pages={5741--5753},
 }

@ARTICLE{2014Appuswamy,
  author={Appuswamy, Rathinakumar and Franceschetti, Massimo},
  journal={IEEE Transactions on Information Theory},
  title={Computing Linear Functions by Linear Coding Over Networks},
  year={2014},
  volume={60},
  number={1},
  pages={422-431},
  doi={10.1109/TIT.2013.2283075}}

@ARTICLE{2018HTYG,
  author={Huang, Cupjin and Tan, Zihan and Yang, Shenghao and Guang, Xuan},
  journal={IEEE Transactions on Information Theory}, 
  title={Comments on Cut-Set Bounds on Network Function Computation}, 
  year={2018},
  volume={64},
  number={9},
  pages={6454-6459},
  doi={10.1109/TIT.2018.2827405}}

@ARTICLE{2023WeiXuGe,
  author={Wei, Hengjia and Xu, Min and Ge, Gennian},
  journal={IEEE Transactions on Information Theory}, 
  title={Robust Network Function Computation}, 
  year={2023},
  volume={69},
  number={11},
  pages={7070-7081},
  doi={10.1109/TIT.2023.3296154}}

@ARTICLE{2024Guangsourcesecure,
  author={Guang, Xuan and Bai, Yang and Yeung, Raymond W.},
  journal={IEEE Transactions on Information Theory}, 
  title={Secure Network Function Computation for Linear Functions—Part I: Source Security}, 
  year={2024},
  volume={70},
  number={1},
  pages={676-697},
  doi={10.1109/TIT.2023.3328454}}

@article{2025Guangfunctionsecure,
  title={Secure Network Function Computation for Linear Functions, Part II: Target-Function Security},
  author={Bai, Yang and Guang, Xuan and Yeung, Raymond W},
  journal={arXiv preprint arXiv:2504.17514},
  year={2025}
}

@inproceedings{2004koetternetwork,
  title={Network codes as codes on graphs},
  author={Koetter, Ralf and Effros, Michelle and Ho, Tracey and M{\'e}dard, Muriel},
  booktitle={Proceeding of CISS},
  year={2004}
}

@article{2012rainetwork,
  title={On network coding for sum-networks},
  author={Rai, Brijesh Kumar and Dey, Bikash Kumar},
  journal={IEEE Transactions on Information Theory},
  volume={58},
  number={1},
  pages={50--63},
  year={2012},
  publisher={IEEE}
}

@article{2013ramamoorthycommunicating,
  title={Communicating the sum of sources over a network},
  author={Ramamoorthy, Aditya and Langberg, Michael},
  journal={IEEE Journal on Selected Areas in Communications},
  volume={31},
  number={4},
  pages={655--665},
  year={2013},
  publisher={IEEE}
}

@inproceedings{cai2002secure,
  title={Secure network coding},
  author={Cai, Ning and Yeung, Raymond W},
  booktitle={Proceedings IEEE International Symposium on Information Theory},
  pages={323},
  year={2002},
  organization={IEEE}
}

@article{cai2010secure,
  title={Secure network coding on a wiretap network},
  author={Cai, Ning and Yeung, Raymond W},
  journal={IEEE Transactions on Information Theory},
  volume={57},
  number={1},
  pages={424--435},
  year={2010},
  publisher={IEEE}
}

@article{el2012secure,
  title={Secure network coding for wiretap networks of type II},
  author={El Rouayheb, Salim and Soljanin, Emina and Sprintson, Alex},
  journal={IEEE Transactions on Information Theory},
  volume={58},
  number={3},
  pages={1361--1371},
  year={2012},
  publisher={IEEE}
}

@article{guang2018alphabet,
  title={Alphabet size reduction for secure network coding: a graph theoretic approach},
  author={Guang, Xuan and Yeung, Raymond W},
  journal={IEEE Transactions on Information Theory},
  volume={64},
  number={6},
  pages={4513--4529},
  year={2018},
  publisher={IEEE}
}

@article{silva2011universal,
  title={Universal secure network coding via rank-metric codes},
  author={Silva, Danilo and Kschischang, Frank R},
  journal={IEEE Transactions on Information Theory},
  volume={57},
  number={2},
  pages={1124--1135},
  year={2011},
  publisher={IEEE}
}

@article{guang2020local,
  title={Local-encoding-preserving secure network coding},
  author={Guang, Xuan and Yeung, Raymond W and Fu, Fang-Wei},
  journal={IEEE Transactions on Information Theory},
  volume={66},
  number={10},
  pages={5965--5994},
  year={2020},
  publisher={IEEE}
}

@article{bai2023multiple,
  title={Multiple Linear-Combination Security Network Coding},
  author={Bai, Yang and Guang, Xuan and Yeung, Raymond W},
  journal={Entropy},
  volume={25},
  number={8},
  pages={1135},
  year={2023},
  publisher={MDPI}
}

@article{bhattad2005weakly,
  title={Weakly secure network coding},
  author={Bhattad, Kapil and Narayanan, Krishna R and others},
  journal={NetCod, Apr},
  volume={104},
  pages={8--20},
  year={2005}
}

@ARTICLE{2022ZhaoSun,
  author={Zhao, Yizhou and Sun, Hua},
  journal={IEEE Transactions on Information Theory}, 
  title={Information Theoretic Secure Aggregation With User Dropouts}, 
  year={2022},
  volume={68},
  number={11},
  pages={7471-7484},
  doi={10.1109/TIT.2022.3192874}}

@ARTICLE{2024Zhaosun,
  author={Zhao, Yizhou and Sun, Hua},
  journal={IEEE Transactions on Information Theory}, 
  title={Secure Summation: Capacity Region, Groupwise Key, and Feasibility}, 
  year={2024},
  volume={70},
  number={2},
  pages={1376-1387},
  doi={10.1109/TIT.2023.3342571}}

@ARTICLE{2024WanYaoSunJiCaire,
  author={Wan, Kai and Yao, Xin and Sun, Hua and Ji, Mingyue and Caire, Giuseppe},
  journal={IEEE Transactions on Information Theory}, 
  title={On the Information Theoretic Secure Aggregation With Uncoded Groupwise Keys}, 
  year={2024},
  volume={70},
  number={9},
  pages={6596-6619},
  doi={10.1109/TIT.2024.3422087}}

@ARTICLE{2024WanSunJiCaire,
  author={Wan, Kai and Sun, Hua and Ji, Mingyue and Mi, Tiebin and Caire, Giuseppe},
  journal={IEEE Transactions on Information Theory}, 
  title={The Capacity Region of Information Theoretic Secure Aggregation With Uncoded Groupwise Keys}, 
  year={2024},
  volume={70},
  number={10},
  pages={6932-6949},
  doi={10.1109/TIT.2024.3393740}}

@article{2025yuanSun,
  title={Vector linear secure aggregation},
  author={Yuan, Xihang and Sun, Hua},
  journal={arXiv preprint arXiv:2502.09817},
  year={2025}
}

@article{hu2026capacity,
  title={On the Capacity Region of Individual Key Rates in Vector Linear Secure Aggregation},
  author={Hu, Lei and Ulukus, Sennur},
  journal={arXiv preprint arXiv:2601.03241},
  year={2026}
}

@article{2024zhangwansunwangoptimal,
  title={Optimal communication and key rate region for hierarchical secure aggregation with user collusion},
  author={Zhang, Xiang and Wan, Kai and Sun, Hua and Wang, Shiqiang and Ji, Mingyue and Caire, Giuseppe},
  journal={arXiv preprint arXiv:2410.14035},
  year={2024}
}

@article{2025LiZhangLvFan,
  title={Collusion-Resilient Hierarchical Secure Aggregation with Heterogeneous Security Constraints},
  author={Li, Zhou and Zhang, Xiang and Lv, Jiawen and Fan, Jihao and Chen, Haiqiang and Caire, Giuseppe},
  journal={arXiv preprint arXiv:2507.14768},
  year={2025}
}

@article{wei2024linear,
  title={Linear Network Coding for Robust Function Computation and Its Applications in Distributed Computing},
  author={Wei, Hengjia and Xu, Min and Ge, Gennian},
  journal={arXiv preprint arXiv:2409.10854},
  year={2024}
}

@ARTICLE{2024GuangZhang,
  author={Guang, Xuan and Zhang, Ruze},
  journal={IEEE Transactions on Information Theory}, 
  title={Zero-Error Distributed Compression of Binary Arithmetic Sum}, 
  year={2024},
  volume={70},
  number={5},
  pages={3100-3117},
  doi={10.1109/TIT.2023.3319976}}

@article{xu2022network,
  title={Network function computation with different secure conditions},
  author={Xu, Min and Ge, Gennian and Liu, Minqian},
  journal={arXiv preprint arXiv:2206.05468},
  year={2022}
}

@article{zhang2025information,
  title={Information-theoretic decentralized secure aggregation with collusion resilience},
  author={Zhang, Xiang and Li, Zhou and Li, Shuangyang and Wan, Kai and Ng, Derrick Wing Kwan and Caire, Giuseppe},
  journal={arXiv preprint arXiv:2508.00596},
  year={2025}
}

@article{li2025capacity,
  title={The capacity of collusion-resilient decentralized secure aggregation with groupwise keys},
  author={Li, Zhou and Zhang, Xiang and Zhao, Yizhou and Chen, Haiqiang and Fan, Jihao and Caire, Giuseppe},
  journal={arXiv preprint arXiv:2511.14444},
  year={2025}
}

@article{so2022lightsecagg,
  title={Lightsecagg: a lightweight and versatile design for secure aggregation in federated learning},
  author={So, Jinhyun and He, Chaoyang and Yang, Chien-Sheng and Li, Songze and Yu, Qian and E Ali, Ramy and Guler, Basak and Avestimehr, Salman},
  journal={Proceedings of Machine Learning and Systems},
  volume={4},
  pages={694--720},
  year={2022}
}

@article{jahani2023swiftagg+,
  title={SwiftAgg+: Achieving asymptotically optimal communication loads in secure aggregation for federated learning},
  author={Jahani-Nezhad, Tayyebeh and Maddah-Ali, Mohammad Ali and Li, Songze and Caire, Giuseppe},
  journal={IEEE Journal on Selected Areas in Communications},
  volume={41},
  number={4},
  pages={977--989},
  year={2023},
  publisher={IEEE}
}

@article{xu2025hierarchical,
  title={On hierarchical secure aggregation against relay and user collusion},
  author={Xu, Min and Han, Xuejiao and Wan, Kai and Ge, Gennian},
  journal={arXiv preprint arXiv:2511.20117},
  year={2025}
}

@INPROCEEDINGS{2024ZhouFu,
  author={Zhou, Qin and Fu, Fang-Wei},
  booktitle={2024 IEEE Information Theory Workshop (ITW)}, 
  title={Network Function Computation for Vector Linear Functions}, 
  year={2024},
  volume={},
  number={},
  pages={735-740},
  doi={10.1109/ITW61385.2024.10806965}}

@ARTICLE{2022LiXudiamond,
  author={Li, Dan and Xu, Yinfei},
  journal={IEEE Communications Letters}, 
  title={Computing Vector-Linear Functions on Diamond Network}, 
  year={2022},
  volume={26},
  number={7},
  pages={1519-1523},
  doi={10.1109/LCOMM.2022.3170974}}

@article{guang2025distributed,
  title={Distributed Source Coding for Compressing Vector-Linear Functions},
  author={Guang, Xuan and Sun, Xiufang and Zhang, Ruze},
  journal={arXiv preprint arXiv:2508.02996},
  year={2025}
}
\end{document}